\documentclass[11pt,a4paper]{article}

\usepackage{fullpage}
\usepackage[utf8]{inputenc}
\usepackage{amssymb}
\usepackage{amsmath}
\usepackage{amsthm}
\usepackage[round]{natbib}
\usepackage{url}
\usepackage{graphicx}
\usepackage{multirow}
\usepackage{setspace}
\usepackage{enumerate}
\usepackage{color}
\usepackage[colorlinks=true,linkcolor=blue,citecolor=blue,pdfborder={0 0 0}]{hyperref}
\usepackage{booktabs,rotating}
\usepackage{bm}
\usepackage{paralist}

\newcommand{\eps}{\varepsilon}
\newcommand{\N}{\mathbb{N}}
\newcommand{\R}{\mathbb{R}}
\newcommand{\Z}{\mathbb{Z}}
\newcommand{\dd}{\mathrm{d}}
\newcommand{\Cb}{\mathbb{C}}
\newcommand{\D}{\mathbb{D}}
\newcommand{\E}{\mathbb{E}}
\newcommand{\FF}{\mathcal{F}}
\newcommand{\M}{\mathcal{M}}
\newcommand{\B}{\mathbb{B}}
\newcommand{\G}{\mathbb{G}}
\newcommand{\U}{\mathbb{U}}
\newcommand{\Lc}{\mathcal{L}}

\newcommand{\Nc}{\mathcal{N}}
\newcommand{\Ex}{\mathrm{E}}
\newcommand{\Var}{\mathrm{Var}}
\newcommand{\Cov}{\mathrm{Cov}}

\newcommand{\1}{\mathbf{1}}

\newcommand{\ip}[1]{\lfloor #1 \rfloor}

\renewcommand{\Pr}{\mathbb{P}}

\newcommand{\p}{\overset{\Pr}{\to}}
\newcommand{\scsp}{\scriptscriptstyle \overset{\Pr}{\to}}

\newcommand{\scs}{\scriptscriptstyle}
\renewcommand{\mid}{\,|\,}

\theoremstyle{plain}
\newtheorem{prop}{Proposition}[section]

\newtheorem{cond}[prop]{Condition}

\numberwithin{equation}{section}

\title{Combining cumulative sum change-point detection tests for assessing the stationarity of univariate time series}

\author{Axel B\"ucher\,\footnote{Heinrich-Heine-Universität D\"usseldorf,
Mathematisches Institut,
Universit\"atsstr.~1, 40225 D\"usseldorf, Germany.
{E-mail:} \texttt{axel.buecher@hhu.de}}
\and Jean-David Fermanian\,\footnote{CREST-ENSAE, J120, 3, avenue Pierre-Larousse, 92245 Malakoff cedex, France. {E-mail:} \texttt{jean-david.fermanian@ensae.fr}}
\and Ivan Kojadinovic\,\footnote{CNRS / Universit\'e de Pau et des Pays de l'Adour / E2S UPPA, Laboratoire de math\'ematiques et applications -- IPRA, UMR 5142, B.P. 1155, 64013 Pau Cedex, France.
{E-mail:} \texttt{ivan.kojadinovic@univ-pau.fr}}
}

\begin{document}
\maketitle

\begin{abstract}
We derive tests of stationarity for univariate time series by combining change-point tests sensitive to changes in the contemporary distribution with tests sensitive to changes in the serial dependence. The proposed approach relies on a general procedure for combining dependent tests based on resampling. After proving the asymptotic validity of the combining procedure under the conjunction of null hypotheses and investigating its consistency, we study rank-based tests of stationarity by combining cumulative sum change-point tests based on the contemporary empirical distribution function and on the empirical autocopula at a given lag. Extensions based on tests solely focusing on second-order characteristics are proposed next. The finite-sample behaviors of all the derived statistical procedures for assessing stationarity are investigated in large-scale Monte Carlo experiments and illustrations on two real data sets are provided. Extensions to multivariate time series are briefly discussed as well.

\medskip


\medskip

\end{abstract}


\section{Introduction}

Testing the stationarity  of a time series is of great importance prior to any modeling. Existing approaches assessing whether a time series is stationary could roughly be grouped into two main categories: procedures that mostly work in the frequency domain, and those that mostly work in the time domain. Among the tests in the former group, one finds for instance approaches testing the constancy of a spectral functional \citep[see, e.g.,][]{PriSub69,Pap10}, procedures comparing a time-varying spectral density estimate with its stationary approximation \citep[see, e.g.,][]{DetPreVet11,PreVetDet13,PucPre16} and approaches based on wavelets \citep[see, e.g.,][]{vonNeu00,Nas13,CarNas13,CarNas16}. As far as the second category of tests is concerned, one mostly finds approaches based on the autocovariance / autocorrelation function such as \cite{LeeHaNa03}, \cite{DwiSub11}, \cite{JinWanWan15} and \cite{DetWuZho15}. In particular, the works of \cite{LeeHaNa03} and \cite{DetWuZho15} also clearly pertain to the change-point detection literature \citep[see, e.g.,][for an overview]{CsoHor97,AueHor13}. The latter should not come as a surprise. Indeed, any test for change-point detection may be seen as a test of stationarity designed to be sensitive to a particular type of departure from stationarity.

To illustrate the latter point, let $X_1,X_2,\dots$ be a stretch from a univariate time series and consider the classical \emph{cumulative sum} (CUSUM) test ``for a change in the mean'' \cite[see, e.g.,][]{Pag54,Phi87}. 
The latter is usually regarded as a test of
$$
H_0: X_1, X_2, \ldots \text{ have the same expectation}
$$
but it only holds its level asymptotically if $X_1,X_2, \dots$ is a stretch from a time series whose autocovariances at all lags are constant \citep{Zho13}. Without the latter assumption, a small p-value can only be used to conclude that $X_1,X_2, \ldots$ is not a stretch from a second-order stationary time series.
In other words, without the additional assumption of constant autocovariances, the classical CUSUM test ``for a change in the mean'' is merely a test of second-order stationarity that is particularly sensitive to a change in the expectation.

Obtaining a large p-value when carrying out the previously mentioned test should clearly not be interpreted as no evidence against second-order stationarity since a change in mean is only one possible departure from second-order stationarity. Following \cite{DetWuZho15}, complementing the previous test by tests for change-point detection particularly sensitive to changes in the variance and in the autocorrelation at some fixed lags may, in case of large p-values, comfort a practitioner in considering that $X_1,X_2,\dots$ might well be a stretch from a second-order stationary time series. The aim of this work is to adopt a similar perspective on assessing stationarity but without only restricting the analysis to second-order characteristics. In fact, all finite dimensional distributions induced by a time series could be potentially tested.

More formally, suppose we observe a stretch $X_1,\dots,X_{N}$ from a time series of univariate continuous random variables. For some $2\le h \le N$,  set $n=N-h+1$ and let $\bm Y_1^{\scs (h)},\dots,\bm Y_n^{\scs (h)}$ be $h$-dimensional random vectors defined by
\begin{equation}
\label{eq:Yi}
\bm Y_i^{(h)} = (X_i,\dots,X_{i+h-1}), \qquad i \in \{1,\dots,n\}.
\end{equation}
Note that the quantity $h$ is sometimes called the \emph{embedding dimension} and $h-1$ can be interpreted as the maximum lag under investigation. As an imperfect alternative, we shall focus on tests particularly sensitive to departures from the hypothesis
\begin{equation}
\label{eq:H0:Fh}
H_0^{(h)}: \,\exists \, F^{(h)} \text{ such that } \bm Y_1^{(h)},\dots,\bm Y_n^{(h)} \text{ have the distribution function (d.f.) } F^{(h)}.
\end{equation}

To derive such tests, a first natural approach would be to apply to the random vectors in~\eqref{eq:Yi} non-parametric CUSUM tests such as those based on differences of empirical d.f.s studied in~\cite{GomHor99}, \cite{Ino01} and \cite{HolKojQue13} (see also Section~\ref{sec:dftest} below), or on differences of empirical characteristic functions; see, e.g., \cite{HusMei06a} and \cite{HusMei06b}. However, preliminary numerical experiments (some of which are reported in Section~\ref{sec:MC}) revealed the low power of such an adaptation in the case of the empirical d.f.-based tests, especially when the non-stationarity of the underlying univariate time series is a consequence of changes in the serial dependence. These empirical conclusions, in line with those drawn in \cite{BucKojRohSeg14} in a related context, prompted us to consider the alternative approach consisting of assessing changes in the ``contemporary'' distribution (that is, of the $X_i$) separately from changes in the serial dependence.

Suppose that $H_0^{\scs (h)}$ in~\eqref{eq:H0:Fh} holds and recall that $X_1, \dots, X_{n+h-1}$ is assumed to be a stretch from a time series of univariate continuous random variables. Then, the common d.f.\ of $\bm Y_i^{\scs (h)}$ can be written \citep{Skl59} as
$$
F^{(h)}(\bm x)=C^{(h)} \{ G(x_1),\dots,G(x_h) \}, \qquad \bm x \in \R^h,
$$
where $C^{(h)}$ is the unique \emph{copula} (merely an $h$-dimensional d.f.\ with standard uniform margins) associated with $F^{(h)}$, and $G$ is the common marginal univariate d.f.\ of all the components of the $\bm Y_i^{\scs (h)}$, $i \in \{1,\dots,n\}$. The copula $C^{(h)}$ controls the dependence between the components of the~$\bm Y_i^{\scs (h)}$. Equivalently, it controls the \emph{serial dependence} up to lag $h-1$ in the time series, which is why it is sometimes called the lag $h-1$ \emph{serial copula} or \emph{autocopula} in the literature.

Notice further that, slightly abusing notation, the hypothesis $H_0^{\scs (h)}$ in~\eqref{eq:H0:Fh} can be written as $H_0^{\scs (1)} \cap H_{0,c}^{\scs (h)}$, where
\begin{equation}
\label{eq:H0:1}
H_0^{(1)}: \,\exists \, G \text{ such that } X_1,X_2, \dots \text{ have the d.f.\ } G,
\end{equation}
and
\begin{equation}
\label{eq:H0:Ch}
H_{0,c}^{(h)}: \,\exists \, C^{(h)} \text{ such that } \bm Y_1^{(h)},\dots,\bm Y_n^{(h)} \text{ have the copula } C^{(h)}.
\end{equation}
In other words, $H_0^{\scs (h)}$ in~\eqref{eq:H0:Fh} holds if all the $X_i$  have the same (contemporary) distribution and if all the $\bm Y_i^{\scs (h)}$ have the same copula.

A sensible strategy for assessing whether $H_0^{\scs (h)}$ in~\eqref{eq:H0:Fh} is plausible would thus naturally consist of combining two tests: a test particularly sensitive to departures from $H_0^{\scs (1)}$ in~\eqref{eq:H0:1} and a test particularly sensitive to departures from $H_{0,c}^{\scs (h)}$ in~\eqref{eq:H0:Ch}. For the former, as already mentioned, a natural candidate in the general context under consideration is the CUSUM test based on differences of empirical d.f.s studied in~\cite{GomHor99} and \cite{HolKojQue13}. We shall briefly revisit the latter approach in the setting of serially dependent observations. One of the main goals of this work is to derive a test that is particularly sensitive to departures from $H_{0,c}^{\scs (h)}$ in~\eqref{eq:H0:Ch}, that is,  to changes in the serial dependence. The idea is not new but seems to have been employed only with respect to second-order characteristics of a time series: see, e.g., \cite{LeeHaNa03} for tests on the autocovariance in a CUSUM setting, and \cite{DwiSub11} and \cite{JinWanWan15} for tests in a different setting. Specifically, one of the main contributions of this work is to propose a CUSUM test that is sensitive to departures from $H_{0,c}^{\scs (h)}$. It will be based on a serial version of the so-called \emph{empirical copula} that we should naturally refer to as the \emph{empirical autocopula} hereafter.

Because the aforementioned test based on empirical d.f.s (particularly sensitive to departures from $H_0^{\scs (1)}$ in~\eqref{eq:H0:1} by construction) and the test based on empirical autocopulas (designed to be sensitive to departures from $H_{0,c}^{\scs (h)}$ in~\eqref{eq:H0:Ch}) rely on the same type of resampling, bootstrap replicates on the underlying statistics $S_{n,G}$ and $S_{n,C^{(h)}}$ can be generated jointly to reproduce, approximately, the distribution of $(S_{n,G},S_{n,C^{(h)}})$ under stationarity. Under such an assumption, another main contribution of this work, that may be of independent interest, is a general procedure for combining dependent bootstrap-based tests, relying on appropriate extensions of well-known p-value combination methods such as those of \cite{Fis33} or \cite{StoEtAl49}. 

An interesting and desirable feature of the resulting global testing procedure is that it is rank-based. It is therefore expected to be quite robust in the presence of heavy-tailed observations. Still, in the case of Gaussian time series, some tests based on second-order characteristics might be more powerful. A natural competitor to our aforementioned global test could thus be obtained by combining tests particularly sensitive to changes in the expectation, variance and autocovariances up to lag $h-1$. Interestingly enough, CUSUM versions of such tests can be cast in the setting considered in \cite{BucKoj16b}: they can all be carried out using the same type of resampling and thus, as described in the previous paragraph, their (dependent) p-values can be combined, leading to a test that could be regarded as a test of second-order stationarity.

The paper is organized as follows. The proposed procedure for combining dependent boot\-strap-based tests is described in Section~\ref{sec:combine:tests}, conditions under which it is asymptotically valid under the conjunction of the component null hypotheses are stated and its consistency is theoretically investigated. The detailed description of the combined rank-based test involving empirical d.f.s and empirical autocopulas is given in Section~\ref{sec:rank}, along with theoretical results about its asymptotic validity under the null hypothesis of stationarity. The choice of the embedding dimension $h$ is discussed in Section~\ref{sec:hchoice}. The fourth section is devoted to related combined tests based on second-order characteristics: the corresponding testing procedures are provided and asymptotic validity results under the null are stated. Section~\ref{sec:MC} reports Monte Carlo experiments that are used to empirically study the previously described tests. Some illustrations on real-world data are presented in Section~\ref{sec:illus}. Finally, concluding remarks are provided in Section~\ref{sec:conc}, one of which, in particular, discusses multivariate extensions of the proposed tests. 

Auxiliary results and all proofs are deferred to a sequence of appendices. Additional theoretical and simulation results are provided in a supplementary material. The studied tests are implemented in the package {\tt npcp} \citep{npcp} for the \textsf{R} statistical system \citep{Rsystem}. In the rest of the paper, the arrow~`$\leadsto$' denotes weak convergence in the sense of Definition~1.3.3 in \cite{vanWel96}, while the arrow~`$\scsp$' denotes convergence in probability. All convergences are for $n\to\infty$ if not mentioned otherwise. Finally, given a set $S$, $\ell^\infty(S)$ denotes the space of all bounded real-valued functions on $S$ equipped with the uniform metric.

\section{A general procedure to combine dependent tests based on resampling}
\label{sec:combine:tests}

As argued in the introduction, to assess whether stationarity is likely to hold, it might be beneficial to combine several tests, each of which being designed to be sensitive to a particular form of non-stationarity. As the need for similar approaches may arise in other contexts than stationarity testing, in this section, we propose a very general strategy for combining tests based on resampling by relying on well-known p-value combination methods such as those of \cite{Fis33} or \cite{StoEtAl49}. Recall that, given $r$ p-values $p_1,\dots,p_r$ for right-tailed tests of corresponding null hypotheses $H_0^{\scs (1)}, \dots, H_0^{\scs (r)}$ with corresponding strictly positive weights $w_1,\dots,w_r$  that quantify the importance of each test in the combination, the latter method consists of computing, up to a rescaling term, the global statistic 
\begin{equation}
\label{eq:stouffer}
\psi_S(p_1,\dots,p_r) = \sum_{j=1}^r w_j \Phi^{-1}(1 - p_j),
\end{equation}
where $\Phi^{-1}$ is the quantile function of the standard normal. Large values provide evidence against the global null hypothesis $H_0 = H_0^{\scs (1)} \cap \dots \cap H_0^{\scs (r)}$. By analogy, the corresponding weighted version of the global statistic in Fisher's p-value combination method can be defined by
\begin{equation}
\label{eq:fisher}
\psi_F(p_1,\dots,p_r) = -2 \sum_{j=1}^r w_j \log(p_j).
\end{equation}
If the p-values $p_1, \dots, p_r$ are independent and uniformly distributed on $(0,1)$, then it can be verified that $\psi_S(p_1,\dots,p_r)$ or $\psi_F(p_1,\dots,p_r)$ are pivotal, giving rise to simple exact global tests. If the component tests are dependent, however, the distributions of the previous statistics are not pivotal and computing the corresponding global p-values is not straightforward anymore.

Let $\bm X_n$ denote the available data (apart from measurability, no assumptions are made on $\bm X_n$, but it is instructive to think of $\bm X_n$ as an $n$-tuple of possibly multivariate observations which may be serially dependent) and let $T_{n,1}=T_{n,1}(\bm X_n),\dots,T_{n,r}=T_{n,r}(\bm X_n)$ be the statistics, each $\R$-valued, of the $r$ tests to be combined.

We assume furthermore that, for any $j \in \{1,\dots,r\}$, large values of $T_{n,j}$ provide evidence against the hypothesis $H_0^{\scs (j)}$. As we continue, we let $\bm T_n = \bm T_n(\bm X_n)$ denote the $r$-dimensional random vector $(T_{n,1},\dots,T_{n,r}) = (T_{n,1}(\bm X_n),\dots,T_{n,r}(\bm X_n))$.

We suppose additionally that we have available a resampling mechanism which allows us to obtain a sample of $M$ bootstrap replicates $\bm T_n^{\scs [i]} = \bm T_n^{\scs [i]}(\bm X_n, \bm V_n^{\scs [i]})$, $i \in \{1,\dots,M\}$, of $\bm T_n$ where $\bm V_n^{\scs [1]}, \dots, \bm V_n^{\scs [M]}$ are independent and identically distributed (i.i.d.) $\R^n$-valued random vectors representing the additional sources of randomness involved in the resampling mechanism and such that, for any $i \in \{1,\dots,M\}$, $T_{n,j}^{\scs [i]}$ depends on the data $\bm X_n$ and $\bm V_n^{\scs [i]}$, that is, $T_{n,j}^{\scs [i]} = T_{n,j}^{\scs [i]}(\bm X_n, \bm V_n^{\scs [i]})$ for all $j \in \{1,\dots,r\}$. Note that the previous setup naturally implies that the components $T_{n,\scs 1}^{\scs [i]},\dots,T_{n,r}^{\scs [i]}$ of $\bm T_n^{\scs [i]}$ are bootstrap replicates of the components $T_{n,1},\dots,T_{n,r}$ of $\bm T_n$. The fact that all the components of $\bm T_n^{\scs [i]}$ depend on the same additional source of randomness $\bm V_n^{\scs [i]}$ makes it possible to expect that the bootstrap replicates $\bm T_n^{\scs [i]}$, $i \in \{1,\dots,M\}$, be, approximately, i.i.d.\ copies of $\bm T_n$ under the global null hypothesis $H_0 = H_0^{\scs (1)} \cap \dots \cap H_0^{\scs (r)}$. For the individual test based on $T_{n,j}$, $j \in \{1,\dots,r\}$, an approximate p-value could then naturally be computed as
\[
	\frac1M \sum_{i=1}^M \1(T_{n,j}^{[i]} \ge T_{n,j}).
\]

Let $\psi$ be a continuous function from $(0,1)^r$ to $\R$ that is decreasing in each of its $r$ arguments (such as $\psi_S$ or $\psi_F$ in~\eqref{eq:stouffer} and~\eqref{eq:fisher}, respectively). To compute an approximate $p$-value for the global statistic $\psi \{ p_{n,M}(T_{n,1}),\dots,p_{n,M}(T_{n,r}) \}$, we propose the following procedure:
\begin{enumerate}

\item Let $\bm T_n^{\scs [0]} = \bm T_{n}$.

\item Given a large integer $M$, compute the sample of $M$ bootstrap replicates $\bm T_n^{\scs [1]},\dots,\bm T_n^{\scs [M]}$ of~$\bm T_n^{\scs [0]}$.

\item Then, for all $i \in \{0,1,\dots,M\}$ and $j \in \{1,\dots,r\}$, compute
\begin{equation}
\label{eq:pval_T}
p_{n,M}(T_{n,j}^{[i]}) = \frac{1}{M+1} \bigg\{\frac{1}{2} + \sum_{k=1}^M \1 \left( T_{n,j}^{[k]} \geq T_{n,j}^{[i]} \right) \bigg\}.
\end{equation}

\item Next, for all $i \in \{0,1,\dots,M\}$, compute
\begin{equation}
\label{eq:WnMi}
W_{n,M}^{[i]} = \psi \{ p_{n,M}(T_{n,1}^{[i]}),\dots,p_{n,M}(T_{n,r}^{[i]}) \}.
\end{equation}
\item The global statistic is $W_{n,M}^{\scs [0]}$ and the corresponding approximate $p$-value is given by
\begin{equation}
\label{eq:pval_W}
p_{n,M}(W_{n,M}^{[0]}) =\frac{1}{M} \sum_{k=1}^M \1 \left( W_{n,M}^{[k]} \geq W_{n,M}^{[0]} \right).
\end{equation}
\end{enumerate}

Note that the quantities $p_{n,M}(T_{n,j}^{\scs [i]})$, $j \in \{1,\dots,r\}$, in Step~3 can be regarded as approximate p-values for the ``statistic values'' $T_{n,j}^{\scs [i]}$, $j \in \{1,\dots,r\}$. The offset by $1/2$ and the division by $M+1$ instead of $M$ in the formula is carried out to ensure that $p_{n,M}(T_{n,j}^{\scs [i]})$ belongs to the interval $(0,1)$ so that Step~4 is well-defined.

The next result, proved in Appendix~\ref{app:proofs}, provides conditions under which the global test based on $W_{n,M}^{\scs [0]}$ given by~\eqref{eq:WnMi} is asymptotically valid under the global null hypothesis $H_0 = H_0^{\scs (1)} \cap \dots \cap H_0^{\scs (r)}$ and the natural assumption that $M = M_n \to \infty$ as $n \to \infty$. Before proceeding, note that $W_{\scs n,M_n}^{\scs [0]}$ is a Monte Carlo approximation of the unobservable statistic 
\begin{equation}
\label{eq:Wn}
W_n = \psi\{\Pr(T_{n,1}^{\scs [1]} \geq T_{n,1} \mid \bm X_n),\dots,\Pr(T_{n,r}^{\scs [1]} \geq T_{n,r} \mid \bm X_n)\}.
\end{equation}

\begin{prop}
\label{prop:combined:general}
Let $M=M_n\to\infty$ as $n\to\infty$. Assume that $H_0 = H_0^{\scs (1)} \cap \dots \cap H_0^{\scs (r)}$ holds, that $\bm T_n$ converges weakly to $\bm T = (T_1,\dots,T_r)$, where $\bm T$ has a continuous d.f., and that either 
\begin{equation}
\label{eq:uncond:Tn}
(\bm T_n, \bm T_n^{[1]}, \bm T_n^{[2]}) \leadsto (\bm T, \bm T^{[1]}, \bm T^{[2]}),
\end{equation}
where $\bm T^{[1]}$ and $\bm T^{[2]}$ are independent copies of $\bm T$, or
\begin{equation}
\label{eq:cond:Tn}
\sup_{\bm x \in \R^r} | \Pr(\bm T_n^{[1]} \leq \bm x \mid \bm X_n) -  \Pr(\bm T_n \leq \bm x)|  \p 0.
\end{equation}
Then, for any $N \in \N$,
\begin{equation}
\label{eq:uncondM}
(W_{n,M_n}^{[0]}, W_{n,M_n}^{[1]},\dots, W_{n,M_n}^{[N]}) \leadsto (W,W^{[1]},\dots,W^{[N]}),
\end{equation}
where 
\begin{equation}
\label{eq:W}
W = \psi\{\bar F_{T_1}(T_1), \dots, \bar F_{T_r}(T_r) \}
\end{equation}
is the weak limit of $W_n$ in~\eqref{eq:Wn} with $\bar F_{T_j}(x) = \Pr(T_j \geq x)$, $x \in \R$, $j \in \{1,\dots,r\}$, and $W^{[1]},\dots,W^{[N]}$ are independent copies of $W$. Furthermore, if $\psi$ is chosen in such a way that the random variable $W$ has a continuous d.f., then
\begin{align}
\label{eq:condM}
&\sup_{x \in\R} | \Pr(W_{n,M_n}^{[1]} \le x \mid \bm X_n) - \Pr(W_n \le x) | \p 0, \\
\label{eq:MC}
&\sup_{x \in\R} \bigg| \frac{1}{M_n} \sum_{i=1}^{M_n} \1(W_{n,M_n}^{[i]} \le x)  - \Pr(W_n \le x) \bigg| \p 0,
\end{align}
and, as a consequence, $p_{n,M_{n}}(W_{\scs n,M_n}^{\scs [0]}) \leadsto \text{Uniform}(0,1)$, where $p_{n,M_{n}}(W_{\scs n,M_n}^{\scs [0]})$ is defined by~\eqref{eq:pval_W}.
\end{prop}

It is worthwhile mentioning that, by Lemma~2.2 of \cite{BucKoj18} and the assumption of continuity for the d.f.\ of $\bm T$, the statements~\eqref{eq:uncond:Tn} and~\eqref{eq:cond:Tn} are actually equivalent in the setting under consideration. Notice also that the resulting unconditional bootstrap consistency statement in~\eqref{eq:uncondM} does not require $W$ in~\eqref{eq:W} to have a continuous d.f. Proving the latter might actually be quite complicated as shall be illustrated in a particular case in Section~\ref{sec:comb:cop:df}.

We end this section by providing a result, proved in Appendix~\ref{app:proofs}, that states conditions under which the global test based on $W_{n,M}^{\scs [0]}$ given by~\eqref{eq:WnMi} leads to the rejection of the global null hypothesis $H_0 = H_0^{\scs (1)} \cap \dots \cap H_0^{\scs (r)}$.

\begin{prop}
  \label{prop:combined:alternative}
  Let $M=M_n\to\infty$ as $n\to\infty$. Assume that
  \begin{enumerate}[(i)]

  \item the combining function $\psi$ is of the form 
    \[
      \psi(p_1, \dots, p_r) = \sum_{j=1}^r w_j \varphi(p_j),
    \]
    where $\varphi$ is decreasing, non-negative and one-to-one from $(0,1)$ to $(0,\infty)$,
    
  \item there exists $j_0 \in \{1,\dots,r\}$ such that the null hypothesis $H_0^{\scs (j_0)}$ of $j_0$th test $T_{n,j_0}$ does not hold and $\Pr(T_{n,j_0}^{\scs [1]} \ge T_{n,j_0})$ converges to zero,

  \item for any $j \in \{1,\dots,r\}$, the sample of bootstrap replicates $T_{n,j}^{\scs [1]}, \dots, T_{n,j}^{\scs [M_n]}$ does not contain ties.
  \end{enumerate}
  Then, the approximate p-value  $p_{n,M_n}(W_{\scs n,M_n}^{\scs [0]})$ of the global test converges to zero in probability, where $p_{n,M_n}(W_{\scs n,M_n}^{\scs [0]})$ is defined by~\eqref{eq:pval_W}.
\end{prop}

Let us comment on the assumptions of the previous proposition. Assumption~$(i)$ is satisfied by the function $\psi_F$ defined by~\eqref{eq:fisher} but not by the function $\psi_S$ defined by~\eqref{eq:stouffer}. A result similar to Proposition~\ref{prop:combined:alternative}, which can be used to handle the function~$\psi_S$, is stated and proved in the supplementary material. Assumption~$(ii)$ can for instance be shown to hold under the hypothesis of one change in the contemporary d.f.\ of a time series when $T_{n,j_0}$ is a test statistic such as the one to be defined in Section~\ref{sec:dftest}, the observations are i.i.d., and the underlying resampling mechanism is a particular multiplier bootstrap. Specifically, in that case,  one can rely on Theorem~3 of \cite{HolKojQue13} to show that, under the hypothesis of one change in the contemporary d.f., $T_{n,j_0}$ diverges to infinity in probability while $T_{n,j_0}^{\scs [1]}$ is bounded in probability, implying that $T_{n,j_0}^{\scs [1]} - T_{n,j_0}$ diverges to $-\infty$ in probability, and thus that $\Pr(T_{n,j_0}^{\scs [1]} \ge T_{n,j_0})$ converges to zero. Finally, assumption~$(iii)$ appears empirically to be satisfied for most bootstrap-based tests for time series of continuous random variables.

\section{A rank-based combined test sensitive to departures from $H_{0}^{(h)}$}
\label{sec:rank}

The aim of this section is to use the results of the previous section to derive a global test of stationarity by combining a test that is particularly sensitive to departures from $H_0^{\scs (1)}$ in~\eqref{eq:H0:1} with a test that is particularly sensitive to departures from $H_{0,c}^{\scs (h)}$ in~\eqref{eq:H0:Ch}. We start by describing the latter test and provide conditions under which it is asymptotically valid under stationarity. The available data, denoted generically by $\bm X_n$ in Section~\ref{sec:combine:tests}, take here, as in the introduction, the form of a stretch $X_1,\dots,X_{n+h-1}$ from a univariate time series, where $h$ is the chosen embedding dimension and where each $X_i$ is assumed to have  a continuous d.f.

\subsection{A  copula-based test sensitive to changes in the serial dependence}

The test that we consider has the potential of being sensitive to all types of changes in the serial dependence up to lag $h-1$. Under $H_0^{\scs (h)}$ in~\eqref{eq:H0:Fh}, this serial dependence is completely characterized by the (auto)copula $C^{(h)}$ in~\eqref{eq:H0:Ch}. It is then natural to base the test on \emph{empirical (auto)copulas} \citep[see, e.g.,][]{Deh79,Deh81} calculated from portions of the data. For any $1 \leq k \leq l \leq n$, let
\begin{equation}
\label{eq:Cklh}
C_{k:l}^{(h)}(\bm u) = \frac{1}{l-k+1} \sum_{i=k}^l \prod_{j=1}^h \1\{ G_{k:l}(X_{i+j-1}) \leq u_j \}, \qquad \bm u \in [0,1]^h,
\end{equation}
where
\begin{equation}
\label{eq:Gkl}
G_{k:l}(x) = \frac{1}{l+h-k} \sum_{j=k}^{l+h-1} \1(X_j \leq x), \qquad x \in \R,
\end{equation}
with the convention that $C_{k:l}^{(h)} = 0$ if $k > l$. The quantity $C_{k:l}^{\scs (h)}$ is a non-parametric estimator of $C^{(h)}$ based on $\bm Y_k^{\scs (h)},\dots,\bm Y_l^{\scs (h)}$ that, as already mentioned, we shall call the \emph{lag $h-1$ empirical autocopula}. The latter was for instance used in \cite{GenRem04} for testing serial independence. It can be verified that it is a straightforward transposition of one of the usual definitions of the empirical copula  (when computed from a subsample) to the serial context under consideration.

\subsubsection{Test statistic}

The CUSUM statistic that we consider is
\begin{equation}
\label{eq:SnCh}
S_{n,C^{(h)}} = \sup_{s \in [0,1]} \int_{[0,1]^h} \left\{ \D_{n,C^{(h)}}(s,\bm{u}) \right\}^2 \dd C_{1:n}^{(h)} (\bm u) = \max_{1 \leq k \leq n-1} \int_{[0,1]^h} \left\{ \D_{n,C^{(h)}}(k/n,\bm{u}) \right\}^2 \dd C_{1:n}^{(h)} (\bm u),
\end{equation}
where, as mentioned earlier, $\ip{.}$ is the floor function,
\begin{equation}
\label{eq:Dnh}
\D_{n,C^{(h)}}(s,\bm{u}) = \sqrt{n} \lambda_n(0,s) \lambda_n(s,1) \left\{ C_{1:\ip{ns}}^{(h)}(\bm u) - C_{\ip{ns}+1:n}^{(h)} (\bm u) \right\}, \qquad (s, \bm u) \in [0,1]^{h+1},
\end{equation}
and $\lambda_n(s,t) = (\ip{nt}-\ip{ns}) / n$,  $(s,t) \in \Delta = \{ (s,t) \in [0,1]^2: s \le t\}$.

Under $H_0^{\scs (h)}$ in~\eqref{eq:H0:Fh}, the difference between $C_{\scs 1:k}^{\scs (h)}$ and $C_{\scs k+1:n}^{\scs (h)}$ should be small for all $k \in \{1,\dots,n-1\}$, resulting in small values of $S_{n,C^{(h)}}$. At the opposite, large values of $S_{n,C^{(h)}}$ provide evidence of non-stationarity. The coefficient $\sqrt{n} \lambda_n(0,s) \lambda_n(s,1)$ in~\eqref{eq:Dnh} is the classical normalizing term in the CUSUM approach. It ensures that, under suitable conditions, $S_{n,C^{(h)}}$ converges in distribution under the null hypothesis of stationarity. Analogously to what was explained in the introduction, the test based on $S_{n,C^{(h)}}$ should in general not be used to reject $H_{\scs 0,c}^{\scs (h)}$ in~\eqref{eq:H0:Ch}: It is merely a test of stationarity that is particularly sensitive to a change in the lag $h-1$ autocopula.

\subsubsection{Limiting null distribution}

The limiting null distribution of $S_{n,C^{(h)}}$ turns out to be a corollary of a recent result by \cite{BucKoj16} and \cite{BucKojRohSeg14}. Under $H_0^{\scs (h)}$ in~\eqref{eq:H0:Fh}, it can be verified that $\D_{n,C^{(h)}}$ in \eqref{eq:Dnh} can be written as
\begin{equation}
\label{eq:DnhH0}
  \D_{n,C^{(h)}}(s,\bm{u})   = \lambda_n(s,1) \, \Cb_{n,C^{(h)}}(0,s,\bm{u}) - \lambda_n(0,s) \, \Cb_{n,C^{(h)}}(s,1,\bm{u}), \qquad (s,\bm{u}) \in [0,1]^{h+1},
\end{equation}
where
\begin{equation}
  \label{eq:Cnh}
  \Cb_{n,C^{(h)}}(s, t, \bm{u}) = \sqrt{n} \, \lambda_n(s, t) \, \{ C_{\ip{ns}+1: \ip{nt}}^{(h)}(\bm{u}) - C^{(h)}(\bm{u}) \}, \qquad (s, t, \bm{u}) \in \Delta \times [0, 1]^h.
\end{equation}
Hence, the null weak limit of the empirical process $\D_{n,C^{(h)}}$ follows from that of $\Cb_{n,C^{(h)}}$, which we shall call the \emph{sequential empirical autocopula process}.

The following usual condition on the partial derivatives of $C^{(h)}$ \citep[see][]{Seg12} is considered as we continue.

\begin{cond}
\label{cond:pd}
For any $j \in \{1,\dots,h\}$, the partial derivative $\dot C_{\scs j}^{\scs (h)} = \partial C^{\scs (h)}/\partial u_j$ exists and is continuous on $V_{\scs j}^{\scs (h)} = \{ \bm{u} \in [0, 1]^h : u_j \in (0,1) \}$.
\end{cond}

Condition~\ref{cond:pd} is nonrestrictive in the sense that it is necessary so that the candidate weak limit of $\Cb_{n,C^{(h)}}$ exists pointwise and has continuous sample paths. In the sequel, following \cite{BucVol13}, for any $j \in \{1,\dots,h\}$, we define $\dot C_{\scs j}^{\scs (h)}$ to be zero on the set $\{ \bm{u} \in [0, 1]^h : u_j \in \{0,1\} \}$. Also, as we continue, for any $j \in \{1,\dots,h\}$ and any $\bm{u} \in [0, 1]^h$, $\bm{u}^{(j)}$ will stand for the vector of $[0, 1]^h$ defined by $u^{\scs (j)}_i = u_j$ if $i = j$ and 1 otherwise.

The null weak limit of $\Cb_{n,C^{(h)}}$ follows in turn from that of the sequential serial empirical process
\begin{equation}
\label{eq:seqep}
  \B_{n,C^{(h)}}(s, t, \bm{u})
  = \frac{1}{\sqrt{n}} \sum_{i=\ip{ns}+1}^{\ip{nt}} \bigg[  \prod_{j=1}^h \1\{ G(X_{i+j-1}) \leq u_j \} - C^{(h)}(\bm{u}) \bigg], \qquad (s, t,\bm{u}) \in \Delta \times [0, 1]^h,
\end{equation}
with the convention that $\B_{n,C^{(h)}}(s, t, \cdot) = 0$ if $\ip{nt} - \ip{ns} = 0$.

The following result, stating the weak limit of $\Cb_{n,C^{(h)}}$ and proved in Appendix~\ref{app:proofs}, is a consequence of the results of \cite{BucKoj16} and \cite{BucKojRohSeg14}. It considers $X_1,\dots,X_{n+h-1}$ as a stretch from a \emph{strongly mixing sequence}. For a sequence of random variables $(Z_i)_{i \in \Z}$, the $\sigma$-field generated by $(Z_i)_{a \leq i \leq b}$, $a, b \in \Z \cup \{-\infty,+\infty \}$, is denoted by $\FF_a^b$. The strong mixing coefficients corresponding to the sequence $(Z_i)_{i \in \Z}$ are then defined by $\alpha_0^Z = 1/2$,
\begin{equation}
\label{eq:alpha}
\alpha_r^Z = \sup_{p \in \Z} \sup_{A \in \FF_{-\infty}^p,B\in \FF_{p+r}^{+\infty}} \big| \Pr(A \cap B) - \Pr(A) \Pr(B) \big|, \qquad r \in \N, \, r > 0.
\end{equation}
The sequence $(Z_i)_{i \in \Z}$ is said to be \emph{strongly mixing} if $\alpha_r^Z \to 0$ as $r \to \infty$.

\begin{prop}
\label{prop:weak_Cnh_sm}
Let $X_1,\dots,X_{n+h-1}$ be drawn from a strictly stationary sequence $(X_i)_{i \in \Z}$ of continuous random variables whose strong mixing coefficients satisfy $\alpha_r^X = O(r^{-a})$ for some $a > 1$ as $r \to \infty$. Then, provided Condition~\ref{cond:pd} holds,
\begin{equation*}
  \sup_{(s, t,\bm{u})\in \Delta \times [0, 1]^h} \Big| \Cb_{n,C^{(h)}}(s, t, \bm{u}) - \B_{n,C^{(h)}}(s, t, \bm{u}) + \sum_{j=1}^h \dot C_j^{(h)}(\bm{u}) \, \B_{n,C^{(h)}}(s, t, \bm{u}^{(j)}) \Big| \p 0.
\end{equation*}
Consequently, $\Cb_{n,C^{(h)}} \leadsto \Cb_{C^{(h)}}$ in $\ell^\infty(\Delta \times [0, 1]^h)$, where, for any $(s, t, \bm{u}) \in \Delta \times [0, 1]^h$,
\begin{equation}
\label{eq:CbCh}
\Cb_{C^{(h)}}(s, t, \bm{u}) = \B_{C^{(h)}}(s, t, \bm{u}) - \sum_{j=1}^h \dot C_j^{(h)}(\bm{u}) \, \B_{C^{(h)}}(s, t, \bm{u}^{(j)}),
\end{equation}
and $\B_{C^{(h)}}$ in $\ell^\infty(\Delta \times [0, 1]^h)$, a tight centered Gaussian process, is the weak limit of $\B_{n,C^{(h)}}$ in~\eqref{eq:seqep}.
\end{prop}

Since they are not necessary for the subsequent derivations, the expressions of the covariances of $\B_{C^{(h)}}$  and $\Cb_{C^{(h)}}$ are not provided. The latter can however be deduced from the above mentioned references.

The next result, proved in Appendix~\ref{app:proofs}, and partly a simple consequence of the previous proposition and the continuous mapping theorem, gives the limiting distribution of $S_{n,C^{(h)}}$ under the null hypothesis of stationarity.

\begin{prop} \label{prop:acs}
Under the conditions of Proposition~\ref{prop:weak_Cnh_sm}, $\D_{n,C^{(h)}} \leadsto \D_{C^{(h)}}$ in $\ell^\infty([0,1]^{h+1})$, where, for any $(s,\bm{u}) \in [0,1]^{h+1}$,
\begin{equation}
  \label{eq:DCh}
  \D_{C^{(h)}}(s,\bm{u}) =  \Cb_{C^{(h)}}(0,s,\bm{u}) - s \, \Cb_{C^{(h)}}(0,1,\bm{u}),
\end{equation}
and $\Cb_{C^{(h)}}$ is defined by~\eqref{eq:CbCh}. As a consequence, we get
\begin{equation}
\label{eq:SCh}
  S_{n,C^{(h)}} \leadsto S_{C^{(h)}} = \sup_{s \in [0,1]} \int_{[0, 1]^h} \{ \D_{C^{(h)}}(s,\bm{u}) \}^2 \, \dd C^{(h)}(\bm{u}).
\end{equation}
Moreover, the distribution of $S_{C^{(h)}}$ is absolutely continuous with respect to the Lebesgue measure.
\end{prop}

\subsubsection{Bootstrap and computation of approximate p-values}

The null weak limit of $S_{n,C^{(h)}}$ in~\eqref{eq:SCh} is unfortunately untractable. Starting from Proposition~\ref{prop:weak_Cnh_sm} and adapting the approach of \cite{BucKoj16} and \cite{BucKojRohSeg14}, we propose to base the computation of approximate p-values for $S_{n,C^{(h)}}$ on \emph{multiplier} resampling versions of $\Cb_{n,C^{(h)}}$ in~\eqref{eq:Cnh}. For any $m \in \N$ and any $(s,t,\bm{u}) \in \Delta \times [0,1]^h$, let
\begin{equation}
\label{eq:hatCbnm}
  \hat{\Cb}_{n,C^{(h)}}^{[m]}(s, t, \bm{u}) = \hat{\B}_{n,C^{(h)}}^{[m]}(s, t, \bm{u}) -  \sum_{j=1}^h \dot C_{j,1:n}^{(h)}(\bm{u}) \, \hat{\B}_{n,C^{(h)}}^{[m]}(s, t, \bm{u}^{(j)}),
\end{equation}
where
$$
\dot{C}_{j,1:n}^{(h)}(\bm{u}) = \frac{C_{1:n}^{(h)}( \bm{u} + h \bm{e}_j ) - C_{1:n}^{(h)}( \bm{u} - h \bm{e}_j )}{\min(u_j+h,1) - \max(u_j-h, 0)}
$$
with $\bm e_j$ the $j$-th unit vector and
\begin{equation}
\label{eq:hatBnm}
\hat{\B}_{n,C^{(h)}}^{[m]}(s, t,\bm{u}) = \frac{1}{\sqrt{n}} \sum_{i=\ip{ns}+1}^{\ip{nt}} \xi_{i,n}^{[m]} \bigg[ \prod_{j=1}^h \1\{ G_{1:n}(X_{i+j-1}) \leq u_j \} - C_{1:n}^{(h)}(\bm{u}) \bigg],
\end{equation}
with $C_{1:n}^{\scs (h)}$ and $G_{1:n}$ defined by~\eqref{eq:Cklh} and~\eqref{eq:Gkl}, respectively. The sequences of random variables $(\xi_{\scs i,n}^{\scs [m]})_{i \in \Z}$, $m \in \N$, appearing in the expressions of the processes $\hat{\B}_n^{\scs (h),[m]}$ in~\eqref{eq:hatBnm}, $m \in \N$, are independent copies of what was called a \emph{dependent multiplier sequence} in \cite{BucKoj16}. Details on that definition, on how such a sequence can be generated and on how a respective block length parameter can be chosen adaptively are presented in Appendix~\ref{app:dep}.

Next, starting from~\eqref{eq:hatCbnm} and having~\eqref{eq:DnhH0} in mind, multiplier resampling versions of $\D_{n,C^{(h)}}$ are then naturally given, for any $m \in \N$ and $(s,\bm u) \in [0,1]^{h+1}$, by
\begin{align*}
\nonumber
  \hat{\D}_{n,C^{(h)}}^{[m]}(s,\bm{u}) &= \lambda_n(s,1) \, \hat{\Cb}_{n,C^{(h)}}^{[m]}(0, s, \bm{u}) - \lambda_n(0,s) \, \hat{\Cb}_{n,C^{(h)}}^{[m]}(s, 1, \bm{u}) \\
&= \hat{\Cb}_{n,C^{(h)}}^{[m]}(0,s,\bm{u}) - \lambda_n(0,s) \, \hat{\Cb}_{n,C^{(h)}}^{[m]}(0, 1, \bm{u}).
\end{align*}
Corresponding multiplier resampling versions of the statistic $S_{n,C^{(h)}}$ in~\eqref{eq:SnCh} are finally
\begin{equation}
\label{eq:SnChm}
\hat{S}_{n,C^{(h)}}^{[m]} = \sup_{s \in [0,1]} \int_{[0, 1]^h} \{ \hat{\D}_{n,C^{(h)}}^{[m]} (s, \bm{u}) \}^2 \, \dd C_{1:n}^{(h)}(\bm{u}),
\end{equation}
which suggests computing an approximate p-value for $S_{n,C^{(h)}}$ as $M^{-1} \sum_{m=1}^M \1 \big( \hat{S}_{n,C^{(h)}}^{[m]} \geq S_{n,C^{(h)}} \big)$ for some large integer $M$.

The following proposition establishes the asymptotic validity of the multiplier resampling scheme under the null hypothesis of stationarity. The proof is given in Appendix~\ref{app:proofs}.

\begin{prop}
\label{prop:S_mult}
Assume that $X_1,\dots,X_{n+h-1}$ are drawn from a strictly stationary sequence $(X_i)_{i \in \Z}$ of continuous random variables whose strong mixing coefficients satisfy $\alpha_r^X = O(r^{-a})$ as $r \to \infty$ for some $a > 3+3h/2$, and $(\xi_{\scs i,n}^{\scs [1]})_{i \in \Z},(\xi_{\scs i,n}^{\scs [2]})_{i \in \Z}, \dots$ are independent copies of a dependent multiplier sequence satisfying~($\M 1$)--($\M 3$) in Appendix~\ref{app:dep} with $\ell_n = O(n^{1/2 - \gamma})$ for some $0 < \gamma < 1/2$. Then, for any $M \in \N$, 
$$
  \big(\Cb_{n,C^{(h)}}, \hat{\Cb}_{n,C^{(h)}}^{[1]}, \dots, \hat{\Cb}_{n,C^{(h)}}^{[M]} \big)
  \leadsto
  \big(\Cb_{C^{(h)}}, \Cb_{C^{(h)}}^{[1]}, \dots, \Cb_{C^{(h)}}^{[M]} \big)
$$
in $\{\ell^\infty(\Delta \times [0, 1]^h)\}^{M+1}$, where $\Cb_{C^{(h)}}$ is defined by~\eqref{eq:CbCh}, and $\Cb_{C^{(h)}}^{\scs [1]},\dots,\Cb_{C^{(h)}}^{\scs [M]}$ are independent
copies of $\Cb_{C^{(h)}}$. As a consequence, for any $M \in \N$,
$$
  \big( \D_{n,C^{(h)}}, \hat{\D}_{n,C^{(h)}}^{[1]}, \dots, \hat{\D}_{n,C^{(h)}}^{[M]} \big)
  \leadsto
  \big( \D_{C^{(h)}}, \D_{C^{(h)}}^{[1]}, \dots, \D_{C^{(h)}}^{[M]} \big)
$$
in $\{ \ell^\infty([0,1]^{h+1}) \}^{M+1}$, where $\D_{C^{(h)}}$ is defined by~\eqref{eq:DCh} and $\D_{C^{(h)}}^{\scs [1]}, \dots,  \D_{C^{(h)}}^{\scs [M]}$ are independent copies of $\D_{C^{(h)}}$. Finally, for any $M \in \N$,
\begin{equation*}
  \big( S_{n,C^{(h)}}, \hat{S}_{n,C^{(h)}}^{[1]}, \dots, \hat{S}_{n,C^{(h)}}^{[M]} \big)   \leadsto
  \big( S_{C^{(h)}},S_{C^{(h)}}^{[1]}, \dots, S_{C^{(h)}}^{[M]} \big),
\end{equation*}
where $S_{C^{(h)}}$ is defined by~\eqref{eq:SCh} and $S_{C^{(h)}}^{\scs [1]},\dots,S_{C^{(h)}}^{\scs [M]}$ are independent copies of $S_{C^{(h)}}$.
\end{prop}

Notice that, by Lemma~2.2 of \cite{BucKoj18} and the continuity of the d.f.\ of $S_{C^{(h)}}$ (see Proposition~\ref{prop:acs} above), the last statement of Proposition~\ref{prop:S_mult} is equivalent to the following more classical formulation of bootstrap consistency:
$$
\sup_{x \in \R} | \Pr(\hat S_{n,C^{(h)}}^{[1]} \leq x \mid \bm X_n) - \Pr(S_{n,C^{(h)}} \leq x)| \p 0.
$$
Furthermore, Lemma~4.2 in \cite{BucKoj18} ensures that the test based on $S_{n,C^{(h)}}$ with approximate p-value $p_{n,M}(S_{n,C^{(h)}}) = M^{-1} \sum_{m=1}^M \1 \big( \hat{S}_{n,C^{(h)}}^{\scs [m]} \geq S_{n,C^{(h)}} \big)$ holds its level asymptotically under the null hypothesis of stationarity as $n$ and $M$ tend to the infinity. By Corollary~4.3 in the same reference, this implies that $p_{n,M_n}(S_{n,C^{(h)}}) \leadsto \text{Uniform}(0,1)$ when $n\to \infty$, for any sequence $M_n \to \infty$.

\subsection{A d.f.-based test sensitive to changes in the contemporary distribution}
\label{sec:dftest}

We propose to combine the previous test with a test particularity  sensitive to departures from $H_0^{\scs (1)}$ in~\eqref{eq:H0:1}. As mentioned in the introduction, a natural candidate is the CUSUM test studied in~\cite{GomHor99} and extended in \cite{HolKojQue13}. For the sake of a simpler presentation, we proceed as if the only available observations were $X_1,\dots,X_n$, thereby ignoring the remaining $h-1$ ones. The test statistic can then be written as
\begin{equation}
\label{eq:Sng}
S_{n,G} = \sup_{s \in [0,1]} \int_{\R} \left\{ \E_n(s,x) \right\}^2 \dd G_{1:n}(x),
\end{equation}
where
\begin{equation}
\label{eq:En}
\E_n(s,x) = \sqrt{n} \lambda_n(0,s) \lambda_n(s,1) \left\{ G_{1:\ip{ns}}(x) - G_{\ip{ns}+1:n} (x) \right\}, \qquad (s, x) \in [0,1] \times \R,
\end{equation}
and, for any $1 \leq k \leq l \leq n$, $G_{k:l}$ is defined as in~\eqref{eq:Gkl} but with $h=1$. As one can see, the test involves the comparison of the empirical d.f.\ of $X_1,\dots,X_k$ with the one of $X_{k+1},\dots,X_n$ for all $k \in \{1,\dots,n-1\}$. Under $H_0^{\scs (1)}$ in~\eqref{eq:H0:1}, it can be verified that $\E_n$ in \eqref{eq:En} can be written as
\begin{equation*}
  \E_n(s,x) = \G_n(s,x) - \lambda_n(0,s) \, \G_n(1,x), \qquad (s, x) \in [0,1] \times \R,
\end{equation*}
where
\begin{equation}
  \label{eq:Gn}
  \G_n(s, x) = \sqrt{n} \, \lambda_n(0, s) \, \{ G_{1: \ip{ns}}(x) - G(x) \}, \qquad (s, x) \in [0,1] \times \R.
\end{equation}

The following result, proved in Appendix~\ref{app:proofs} and providing the null weak limit of $S_{n,G}$ in~\eqref{eq:Sng}, is partly an immediate consequence of Theorem~1 of \cite{Buc15} and of the continuous mapping theorem.

\begin{prop}
\label{prop:weak_Gn_sm}
Let $X_1,\dots,X_{n}$ be drawn from a strictly stationary sequence $(X_i)_{i \in \Z}$ of continuous random variables whose strong mixing coefficients satisfy $\alpha_r = O(r^{-a})$ for some $a > 1$, as $r \to \infty$. Then, $\G_n \leadsto \G$ in $\ell^\infty([0, 1] \times \R)$, where $\G$ is a tight centered Gaussian process with covariance function
$$
\Cov\{\G(s,x), \G(t, y) \} = \min(s,t) \sum_{k \in \Z} \Cov\{\1(X_0 \leq x) \1(X_k \leq y) \}.
$$
Consequently, $\E_n \leadsto \E$ in $\ell^\infty([0, 1] \times \R)$, where
\begin{equation}
\label{eq:E}
\E(s,x) = \G(s,x) - s \G(1,x), \qquad (s, x) \in [0,1] \times \R,
\end{equation}
and $S_{n,G} \leadsto S_G$ with
\begin{equation}
\label{eq:SG}
S_G = \sup_{s \in [0,1]} \int_{\R} \left\{ \E(s,x) \right\}^2 \dd G(x).
\end{equation}
Moreover, the distribution of $S_G$ is absolutely continuous with respect to the Lebesgue measure.
\end{prop}

Following \cite{GomHor99}, \cite{HolKojQue13} and \cite{BucKoj16}, we shall compute approximate p-values for $S_{n,G}$ using multiplier resampling versions of $\G_n$ in~\eqref{eq:Gn}. Let $(\xi_{i,n}^{\scs [m]})_{i \in \Z}$, $m \in \N$, be independent copies of the same dependent multiplier sequence. For any $m \in \N$ and any $(s,x) \in [0,1] \times \R$, let
\begin{align}
  \nonumber
  \hat{\G}_n^{[m]}(s, x) &=  \frac{1}{\sqrt{n}} \sum_{i=1}^{\ip{ns}} \xi_{i,n}^{[m]} \left\{ \1(X_i \leq x) - G_{1:n}(x) \right\}, \\ 
  \nonumber
  \hat \E_n^{[m]}(s,x) &= \G_n^{[m]}(s,x) - \lambda_n(0,s) \, \G_n^{[m]}(1,x), \\
  \label{eq:SnGm}
  \hat{S}_{n,G}^{[m]} &= \sup_{s \in [0,1]} \int_{\R} \left\{ \hat \E_n^{[m]}(s,x) \right\}^2 \dd G_{1:n}(x).
\end{align}
An approximate p-value for $S_{n,G}$ will then be computed as $p_{n,M}(S_{n,G}) = M^{-1} \sum_{m=1}^M \1 \big( \hat{S}_{n,G}^{\scs [m]} \geq S_{n,G} \big)$ for some large integer $M$.
The asymptotic validity of this approach under the null hypothesis of stationarity can be shown as for the test based on $S_{n,C^{(h)}}$ presented in the previous section. The result is a direct consequence of Corollary~2.2 in \cite{BucKoj16}; see also Proposition~\ref{prop:combined1} in the next section. In particular, $p_{n,M_n}(S_{n,G}) \leadsto \text{Uniform}(0,1)$ when $n \to \infty$, for any sequence $M_n\to\infty$.

\subsection{Combining the two tests}
\label{sec:comb:cop:df}

To combine the two tests, we use the general procedure described in Section~\ref{sec:combine:tests} with $r=2$, $T_{n,1} = S_{n,C^{(h)}}$ and $T_{n,2} = S_{n,G}$, for some suitable function $\psi:(0,1)^2 \to \R$ such as $\psi_S$ in~\eqref{eq:stouffer} or $\psi_F$ in~\eqref{eq:fisher}. To be able to apply Proposition~\ref{prop:combined:general}, we need to find conditions under which $\bm T_n = (T_{n,1}, T_{n,2})$ and its bootstrap replicates satisfy~\eqref{eq:uncond:Tn} or, equivalently,~\eqref{eq:cond:Tn}. A natural prerequisite is to compute the $M$ bootstrap replicates of $T_{n,1} = S_{n,C^{(h)}}$ and $T_{n,2} = S_{n,G}$ in~\eqref{eq:SnChm} and~\eqref{eq:SnGm}, respectively, using the same $M$ dependent multiplier sequences. Since a moving average approach is used to generate such sequences, it follows from~\eqref{eq:movave} that it is sufficient to impose that the same $M$ initial independent normal sequences be used for both tests. In practice, prior to using~\eqref{eq:movave} to generate the $M$ independent copies of the same dependent multiplier sequence, we estimate the key bandwidth parameter $\ell_n$ from $X_1,\dots,X_{n+h-1}$ using the approach proposed in \citet[Section 5.1]{BucKoj16}, briefly overviewed in Appendix~\ref{app:dep}.

The next result, proven in Appendix~\ref{app:proofs},  provides conditions under which~\eqref{eq:uncond:Tn} holds.

\begin{prop}
\label{prop:combined1}
Under the conditions of Proposition~\ref{prop:S_mult}, for any $M \in \N$,
$$
  \big((\D_{n,C^{(h)}},\E_n), (\hat{\D}_{n,C^{(h)}}^{[1]}, \hat{\E}_n^{[1]}), \dots, (\hat{\D}_{n,C^{(h)}}^{[M]}, \hat{\E}_n^{[M]}) \big)
  \leadsto
  \big((\D_{C^{(h)}}, \E), (\D_{C^{(h)}}^{[1]}, \E^{[1]}), \dots, (\D_{C^{(h)}}^{[M]},  \E^{[M]}) \big)
$$
in $\{\ell^\infty([0, 1] \times \R)\}^{2(M+1)}$, where $\D_{C^{(h)}}$ and $\E$ are defined by~\eqref{eq:DCh} and~\eqref{eq:E}, respectively, and $(\D_{C^{(h)}}^{\scs [1]},\E^{[1]}),\dots,(\D_{C^{(h)}}^{\scs [M]},\E^{[M]})$ are independent copies of $(\D_{C^{(h)}},\E)$. Note that we do not specify the joint law of $(\D_{C^{(h)}},\E)$; it will only be important that $(\hat{\D}_{n,C^{(h)}}^{\scs [m]}, \hat{\E}_n^{\scs [m]})$, $m \in \{1,\dots,M\}$, can be considered to have the same joint law as $(\D_{C^{(h)}},\E)$ asymptotically. As a consequence, 
$$
  \big( (S_{n,C^{(h)}},S_{n,G}),( \hat{S}_{n,C^{(h)}}^{[1]},  \hat{S}_{n,G}^{[1]}), \dots, (\hat{S}_{n,C^{(h)}}^{[M]}, \hat{S}_{n,G}^{[M]}) \big)
  \leadsto
  \big( (S_{C^{(h)}},S_G), (S_{C^{(h)}}^{[1]},S_G^{[1]}), \dots, (S_{C^{(h)}}^{[M]},  S_G^{[M]}) \big),
$$
where $S_{C^{(h)}}$ and $S_G$ are defined by~\eqref{eq:SCh} and~\eqref{eq:SG}, respectively, and where the random vectors $(S_{C^{(h)}}^{\scs [1]},S_G^{\scs [1]}), \dots, (S_{C^{(h)}}^{\scs [M]},S_G^{\scs [M]})$ are independent copies of $(S_{C^{(h)}},S_G)$.
\end{prop}

A consequence of the previous proposition is that the unconditional bootstrap consistency statement in~\eqref{eq:uncondM} holds under the conditions of Proposition~\ref{prop:S_mult}. To conclude that the conditional statements given in~\eqref{eq:condM} and~\eqref{eq:MC} hold has well, it is necessary to establish that $W$, given generically by~\eqref{eq:W}, has a continuous d.f. Proving the latter might actually be quite complicated: unlike $S_{C^{(h)}}$ in~\eqref{eq:SCh} and $S_G$ in~\eqref{eq:SG}, $W$ is not a convex function of some Gaussian process, whence the general results from \cite{DavLif84} and the references therein are not applicable. Proving the absolute continuity of the vector $(S_{C^{(h)}},S_G)$ could be a first step but the latter does not seem easy either: available results in the literature are mostly based on complicated conditions from Malliavin Calculus, see, e.g., Theorem~2.1.2 in \cite{Nua06}. For these reasons, we do not pursue such investigations any further in this paper. Nonetheless, we conjecture that $W$ will have a continuous d.f.\ in all except a few very pathological situations.

Under suitable conditions on alternative models, it can further be shown that at least one of the statistics $S_{n,G}$ or $S_{n,C^{(h)}}$ (for $h$ suitably chosen) diverges to infinity in probability at rate~$n$. For instance, for $S_{n,G}$, under the assumption of at most one change in the contemporary d.f.\ of the time series, the latter can be shown by adapting to the serially dependent case the arguments used in \citet[Proof of Theorem 3~(i)]{HolKojQue13}. Further details are omitted for the sake of brevity. As far as bootstrap replicates of $S_{n,G}$ or $S_{n,C^{(h)}}$ are concerned, based on our extensive simulation results, we conjecture that, for many alternative models, the bootstrap replicates are of lower order than $O_\Pr(n)$. As a consequence, assuming the aforementioned results, and when the combining function $\psi$ is $\psi_F$ in~\eqref{eq:fisher}, one can rely on Proposition~\ref{prop:combined:alternative} to show the consistency of the test based on~$W_{\scs n,M_n}^{\scs [0]}$ in~\eqref{eq:WnMi}.

\subsection{On the choice of the embedding dimension $h$}
\label{sec:hchoice}

The methodology described in the previous sections depends on the embedding dimension~$h$. In this section, we will provide some intuition about the trade-off between the choice of small and large values of $h$. Based on the developed arguments, and on the large-scale simulation study in Section~\ref{sec:MC} and in the supplementary material, we will make a practical suggestion at the end of this section.

Let us start by considering arguments in favour of choosing a large value of $h$. For that purpose, note that stationarity is equivalent to $H_0^{\scs (1)}$ in~\eqref{eq:H0:1} and $H_{\scs 0,c}^{\scs (h)}$ in~\eqref{eq:H0:Ch} for all $h\ge 2$, and that a test based on the embedding dimension $h$ can only detect alternatives for which $H_{\scs 0,c}^{\scs (h)}$ does not hold. Hence, since $H_{\scs 0,c}^{\scs (2)} \Leftarrow H_{\scs 0,c}^{\scs (3)} \Leftarrow \dots$, we would like to choose $h$ as large as possible to be consistent against as many alternatives as possible. Note that, at the same time, the potential gain in moving from $h$ to $h+1$ should decrease with $h$: first, the larger $h$, the less likely it seems  that real-life phenomena satisfy $H_{\scs 0,c}^{\scs (h)}$ but not $H_{\scs 0,c}^{\scs (h+1)}$; second, from a model-engineering perspective, the larger the value of~$h$, the more difficult and artificial it becomes to construct sensible models that satisfy $H_{\scs 0,c}^{\scs (h)}$ but not $H_{\scs 0,c}^{\scs (h+1)}$. To illustrate the latter point, constructing such a model on the level of copulas would amount to finding (at least two) different $(h+1)$-dimensional copulas $C^{(h+1)}$ that have the same lower-dimensional (multivariate) margins. More formally and given the serial context under consideration, this would mean finding a model such that
\[
C^{(h+1)}(1,\ldots,1,u_i,\ldots,u_{i+k-1},1,\ldots,1)=C^{(k)}(u_i,\ldots,u_{i+k-1}),
\]
for all $k\in \{ 2,\ldots, h\}$, $i \in \{1,\ldots,h-k+2\}$, $u_i,\ldots,u_{i+k-1} \in [0,1]$, for some given $k$-dimensional copulas $C^{(k)}$.
This problem is closely related to the so-called compatibility problem (\citealp{Nel06}, Section 3.5) and,
to the best of our knowledge, has not yet a general solution.
Some necessary conditions can be found in \citet[Theorem~4]{Rus85} for the case of copulas that are absolutely continuous with respect to the Lebesgue measure on the unit hypercube. As another example, consider as a starting point the autoregressive process $X_i=\beta X_{i-h} + \eps_i$, where the noises $\eps_i\sim \mathcal N(0,\tau^2)$ are i.i.d. and where $|\beta|<1$. The components of the vectors $\bm Y^{\scs (h)}_i=(X_i, \dots,X_{i+h-1})$ are then i.i.d.\ $\mathcal N(0,\tau^2/(1-\beta^2))$. Hence, $C^{(h)}$ is the independence copula and $H_0^{\scs (h)}$ in~\eqref{eq:H0:Fh} is met, while $H_{0,c}^{\scs (h+1)}$ in~\eqref{eq:H0:Ch} would not be met should the parameters $\tau$ and $\beta$ change (smoothly or abruptly) in such a way that $\tau^2/(1-\beta^2)$ stays constant; a rather artificial example. More generally, one could argue that, the larger $h$, the more artificial instances of common time series models (such as ARMA- or GARCH-type models) for which $H_{\scs 0,c}^{\scs (h)}$ holds but not $H_{\scs 0,c}^{\scs (h+1)}$ seem to be.

The previous paragraph suggests to choose $h$ as large as possible, even if the marginal gain of an increase of $h$ becomes smaller for larger and larger $h$. At the opposite, there are also good reasons for choosing $h$ rather small. Indeed, for many  sensible models, the power of the test based on $S_{n,C^{(h)}}$ {in~\eqref{eq:SnCh} is a decreasing function of $h$, at least from some small value onwards. This observation will for instance be one of the results of our simulation study in Section~\ref{sec:MC} (see, e.g., Figure~\ref{fig:h}), but it can also be supported by more theoretical arguments. Indeed, consider for instance the following simple alternative model: $X_1,X_2,\dots$ have the same d.f.\ $G$ and, for some  $s^* \in (0,1)$, $\bm Y_i^{\scs (h)}$, $i \in \{1,\dots,\ip{ns^*}-\ip{h/2}\}$, have copula $C_1^{\scs (h)}$ and $\bm Y_i^{\scs (h)}$, $i \in \{\ip{ns^*}+1+\ip{h/2},\dots,n\}$, have copula $C_2^{\scs (h)} \ne C_1^{\scs (h)}$. For simplicity, we do not specify the laws of the $\bm Y_i^{\scs (h)}$ for $i\in \{\ip{ns^*}-\ip{h/2}+1,\ldots, \ip{ns^*}+\ip{h/2}\}$ (these observations induce negligible effects in the following reasoning), whence, asymptotically, we can do ``as if'' $\bm Y_i^{\scs (h)}$, $i \in \{1,\dots,\ip{ns^*}\}$, have copula $C_1^{\scs (h)}$ and $\bm Y_i^{\scs (h)}$, $i \in \{\ip{ns^*}+1,\dots,n\}$, have copula $C_2^{\scs (h)}$. Under this model and additional regularity conditions, we obtain that
\[
n^{-1} S_{n,C^{(h)}} \leadsto \kappa_h \equiv \{ s^*(1-s^*) \}^2 \int_{[0,1]^h}  \{ C_1^{(h)} (\bm u) - C_2^{(h)}(\bm u) \}^2 \, \dd C_{s^*}^{(h)}(\bm u),
\]
where $C_{\scs s^*}^{\scs (h)} =  s^* C_1^{\scs (h)} + (1-s^*)  C_2^{\scs (h)}$. In other words, the dominating term in an asymptotic expansion of $S_{n,C^{(h)}}$ diverges to infinity at rate $n$, with scaling factor $\kappa_h$ depending on $h$. Since we conjecture that the bootstrap replicates of $S_{n,C^{(h)}}$ are of lower order than $O_\Pr(n)$ for any $h$, we further conjecture that the power curves of the test will  be controlled to a large extent by the ``signal of non-stationarity'' $\kappa_h$. The impact of $h$ on this quantity is ambiguous, but, in many sensible models, it is decreasing in $h$ eventually, inducing a sort of ``curse of dimensionality''. This results in a smaller power of the corresponding test for larger $h$ and fixed sample size $n$, as will be empirically confirmed in several scenarios considered in the Monte Carlo experiments of Section~\ref{sec:MC} and in the supplementary material.

Additionally, several arguments lead us to assume that smaller values of~$h$ also yield a better approximation of the nominal level. From an empirical perspective, this will be confirmed for all the scenarios under stationarity in our Monte Carlo experiments. While we are not aware of any theoretical result for our quite general serially dependent setting (that would include the dependent multiplier bootstrap), some results are available for the i.i.d.\ or non-bootstrap case. For instance, \cite{CheCheKat13} provide bounds on the approximation error of i.i.d.\ sum statistics by an i.i.d.\ multiplier bootstrap; the bounds are increasing in the dimension~$h$. Moreover, the asymptotics of our test statistics relying on the asymptotics of empirical processes, we would be interested in a good approximation of empirical processes by their limiting counterparts. As shown in \cite{DedMerRio14} for the case of beta-mixing random variables, the approximation error by strong approximation techniques is again increasing in $h$.

Globally, the above arguments as well as the results of the simulation study in Section~\ref{sec:MC} below and in the supplementary material suggest that a rather small value of $h$, for instance in \{2,3,4\}, should be sufficient to test strong stationarity in many situations. Such a choice would provide relatively powerful tests for many interesting alternatives without strongly suffering from the curse of dimensionality. Depending on the ultimate interest, one might also consider choosing $h$ differently, e.g., as the ``forecast horizon''. Finally, a natural research direction would consist of developing data-driven procedures for choosing $h$, for instance following ideas developed in \cite{EscLob09} for testing serial correlation in a time series. However, such an analysis appears to be a research topic in itself and lies beyond the scope of the present paper.

\section{A combined test of second-order stationarity}
\label{sec:sotests}

Starting from the general framework considered in \cite{BucKoj16b} and proceeding as in Section~\ref{sec:rank}, one can derive a combined test of second-order stationarity. Given the embedding dimension $h \geq 2$ and the available univariate observations $X_1,\dots,X_{n+h-1}$, let $\bm Z_i^{\scs (q)}$, $i \in \{1,\dots,n\}$, be the random variables defined by
\begin{equation}
\label{eq:Zi}
\bm Z_i^{(q)} = \left\{
\begin{array}{ll}
X_i, & \mbox{if } q=1, \\
(X_i,X_{i+q-1}), & \mbox{if } q \in \{2,\dots,h\}.
\end{array}
\right.
\end{equation}
Let $\phi$ be a symmetric, measurable function on $\R \times \R$ or on $\R^2 \times \R^2$. Then, the \emph{$U$-statistic of order~2 with kernel $\phi$} obtained from the subsample $\bm Z_k^{\scs (q)},\dots,\bm Z_l^{\scs (q)}$, $1 \leq k < l \leq n$, is given by
\begin{equation}
\label{eq:Uklqphi}
U_{k:l,q,\phi,} = \frac{1}{\binom{l-k+1}{2}} \sum_{k \leq i < j \leq l} \phi(\bm Z_i^{(q)},\bm Z_j^{(q)}).
\end{equation}
We focus on CUSUM tests for change-point detection based on the generic statistic
\begin{equation}
\label{eq:Snqphi}
S_{n,q,\phi} = \max_{2 \leq k \leq n-2} | \U_{n,q,\phi} (k/n) | = \sup_{s \in [0,1]} | \U_{n,q,\phi}(s) |,
\end{equation}
where
\begin{equation*}
\U_{n,q,\phi}(s) = \sqrt n \lambda_n(0,s) \lambda_n(s,1) ( U_{1:\ip{ns},q,\phi} - U_{\ip{ns}+1:n,q,\phi} )  \qquad \mbox{if } s \in [2/n,1-2/n],
\end{equation*}
and $\U_{n,q,\phi}(s)=0$ otherwise.

With the aim of assessing whether second-order stationarity is plausible, the following possibilities for $q \in \{1,\dots,h\}$ and the kernel $\phi$ are of interest: If $q=1$ and $\phi(z,z') = m(z, z') = z$, $z,z' \in \R$, the statistic $S_{n,q,\phi} = S_{n,1,m}$ is (asymptotically equivalent to) the classical CUSUM statistic that is particularly sensitive to changes in the expectation of $X_1,\dots,X_n$. Similarly, setting $q=1$ and $\phi(z,z') = v(z, z') = (z - z')^2 / 2$, $z,z' \in \R$, gives rise to the statistic $S_{n,1,v}$ particularly sensitive to changes in the variance of the observations. For $q \in \{2,\dots,h\}$, setting $\phi(\bm z, \bm z') = a(\bm z, \bm z') = (z_1 - z_1') (z_2 - z_2') / 2$, $\bm z, \bm z' \in \R^2$, results in the CUSUM statistic $S_{n,q,a}$ sensitive to changes in the autocovariance at lag $q-1$.

From \cite{BucKoj16b}, CUSUM tests based on $S_{n,1,m}$, $S_{n,1,v}$  and $S_{n,q,a}$, $q \in \{ 2, \dots, h\}$, sensitive to changes in the expectation, variance and autocovariances, respectively, can all be carried out using a resampling scheme based on dependent multiplier sequences. As a consequence, they can be combined by proceeding as in Sections~\ref{sec:combine:tests} and~\ref{sec:comb:cop:df}. Specifically, for the generic test based on $S_{n,q,\phi}$, let $(\xi_{i,n}^{\scs [m]})_{i \in \Z}$, $m \in \N$, be independent copies of the same dependent multiplier sequence and, for any $m \in \N$ and $s \in [0,1]$, let
\begin{equation*}
\hat \U_{n,q,\phi}^{[m]} (s)  = \frac{2}{\sqrt{n}} \sum_{i=1}^{\ip{ns}} \xi_{i,n}^{(m)} \hat \phi_{1,1:n} (\bm Z_i^{(q)}) - \lambda_n(0,s) \times \frac{2}{\sqrt{n}} \sum_{i=1}^n \xi_{i,n}^{(m)} \hat \phi_{1,1:n}(\bm Z_i^{(q)}), \qquad \mbox{if } s \in [2/n,1-2/n],
\end{equation*}
and $\hat \U_{n,q,\phi}^{[m]} (s) = 0$ otherwise, where
\begin{equation*}
\hat \phi_{1,1:n}(\bm Z_i^{(q)}) = \frac{1}{n-1} \sum_{j=1, j \neq i}^n \phi(\bm Z_i^{(q)}, \bm Z_j^{(q)}) - U_{1:n,q,\phi}, \qquad i \in \{1,\dots,n\},
\end{equation*}
with $U_{1:n,q,\phi}$ defined by~\eqref{eq:Uklqphi}. Then, multiplier replications of $S_{n,q,\phi}$ are given by
\begin{equation*}
\hat S_{n,q,\phi}^{[m]} = \max_{2 \leq k \leq n-2} | \hat \U_{n,q,\phi}^{[m]} (k/n) | = \sup_{s \in [0,1]} | \hat \U_{n,q,\phi}^{[m]}(s) |, \qquad m \in \N,
\end{equation*}
and an approximate p-value for $S_{n,q,\phi}$ can be computed as $p_{n,M}(S_{n,q,\phi}) = M^{-1} \sum_{m=1}^M \1 \big( \hat S_{n,q,\phi}^{\scs [m]} \geq S_{n,q,\phi} \big)$ for some large integer $M \in \N$.

To obtain a test of second-order stationarity, we use again the combining procedure of Section~\ref{sec:combine:tests}, this time, with $r = h + 1$, $T_{n,1} = S_{n,1,m}$, $T_{n,2} = S_{n,1,v}$ and $T_{n,q+1} = S_{n,q,a}$, $q \in \{2,\dots,h\}$, for some function $\psi:(0,1)^{h+1} \to \R$ decreasing in each of its arguments such as $\psi_S$ in~\eqref{eq:stouffer} or $\psi_F$ in~\eqref{eq:fisher}. As in Section~\ref{sec:comb:cop:df}, to compute bootstrap replicates of the components of $\bm T_n = (T_{n,1},\dots,T_{n,r})$, we use the same $M$ dependent multiplier sequences. Specifically, we first estimate $\ell_n$ from $X_1,\dots,X_{n}$ as explained in \citet[Section 2.4]{BucKoj16b} for $\phi=m$. Then, with the obtained value of $\ell_n$, we generate $M$ independent copies of the same dependent multiplier sequence using~\eqref{eq:movave} and compute the corresponding multiplier replicates $\hat S_{n,q,\phi}^{\scs [1]},\dots,\hat S_{n,q,\phi}^{\scs [M]}$ for $q=1$ and $\phi \in \{m,v\}$, and for $q \in \{ 2, \dots, h\}$ and $\phi = a$. 

As in Section~\ref{sec:comb:cop:df}, to establish the asymptotic validity of the global test under stationarity using Proposition~\ref{prop:combined:general}, we need to establish conditions under which, using the notation of Section~\ref{sec:combine:tests}, $\bm T_n = (T_{n,1},\dots,T_{n,r})$ and its bootstrap replicates satisfy~\eqref{eq:uncond:Tn} or, equivalently,~\eqref{eq:cond:Tn}. The latter can be proved by starting from Proposition~2.5 in \cite{BucKoj16b} and by proceeding as in the proofs of the results stated in Section~\ref{sec:comb:cop:df}. For the sake of simplicity, the conditions in the following proposition require that $X_1,\dots,X_{n+h-1}$ is a stretch from an absolutely regular sequence. Indeed, assuming that $(X_i)_{i \in \Z}$ is only strongly mixing leads to significantly more complex statements. Recall that the absolute regularity coefficients corresponding to a sequence $(Z_i)_{i \in \Z}$ are defined by
$$
\beta_r^Z = \sup_{p \in \Z} \Ex \sup_{A \in \FF_{p+r}^\infty} | \Pr(A \mid \FF_{-\infty}^p ) - \Pr(A) |, \qquad r \in \N, \, r > 0,
$$
where $\FF_a^b$ is defined above~\eqref{eq:alpha}. The sequence $(Z_i)_{i \in \N}$ is then said to be \emph{absolutely regular} if $\beta_r \to 0$ as $r \to \infty$. As $\alpha_r^Z \leq \beta_r^Z$, absolute regularity implies strong mixing.

\begin{prop}
Let $X_1,\dots,X_{n+h-1}$ be drawn from a strictly stationary sequence $(X_i)_{i \in \Z}$ such that $\Ex\{|X_1|^{2(4+\delta)}\} < \infty$ for some $\delta > 0$. Also, let $(\xi_{\scs i,n}^{\scs [1]})_{i \in \Z}$ and $(\xi_{\scs i,n}^{\scs [2]})_{i \in \Z}$ be independent copies of a dependent multiplier sequence satisfying~($\M 1$)--($\M 3$) in Appendix~\ref{app:dep} with $\ell_n = O(n^{1/2 - \gamma})$ for some $1/(6+2\delta) < \gamma < 1/2$. Then, if 
$\beta_r^X = O(r^{-b})$ for some $b > 2(4 + \delta)/\delta$ as $r \to \infty$,~\eqref{eq:uncond:Tn} or, equivalently,~\eqref{eq:cond:Tn}, hold.

\end{prop}

\section{Monte Carlo experiments}
\label{sec:MC}

Extensive simulations were carried out in order to try to answer several fundamental questions (hereafter in bold) regarding the tests proposed in Sections~\ref{sec:rank} and~\ref{sec:sotests}. For the sake of readability, we only present a small subset of the performed Monte Carlo experiments in detail and refer the reader to the supplementary material for more results. Before formulating the questions, we introduce abbreviations for the components tests whose behavior we investigated:
\vspace{-\medskipamount}
\begin{itemize}\parskip0pt
\item d for the d.f.\ test based on $S_{n,G}$ in~\eqref{eq:Sng},
\item c for the empirical autocopula test at lag $h-1$ based on $S_{n,C^{(h)}}$ in~\eqref{eq:SnCh} (the value of $h$ will always be clear from the context),
\item m for the sample mean test based on $S_{n,m}^{\scs (1)}$ defined generically by~\eqref{eq:Snqphi},
\item v for the variance test based on $S_{n,v}^{\scs (1)}$ defined generically by~\eqref{eq:Snqphi}, and
\item a for the autocovariance test at lag $q-1$ based on $S_{n,a}^{\scs (q)}$, $q \in \{2,\dots,h\}$, defined generically by~\eqref{eq:Snqphi} (the value of $q$ will always be clear from the context).
\end{itemize}
\vspace{-\medskipamount}

\noindent With these conventions, the following abbreviations are used for the combined tests:
\vspace{-\medskipamount}
\begin{itemize}\parskip0pt
\item dc: equally weighted combination of the tests d and c for embedding dimension $h$ or, equivalently, lag $h-1$,
\item va: combination of the test v with weight 1/2 and the autocovariance tests a for lags $q \in \{1,\dots,h-1\}$ with equal weights $1/\{2(h-1)\}$,
\item mva: combination of the test m with weight 1/3, of the variance test v with weight 1/3 and the autocovariance tests a for lags $q \in \{1,\dots,h-1\}$ with equal weights $1/\{3(h-1)\}$,
\item dcp: combination of the test d with weight 1/2 with pairwise bivariate empirical autocopula tests for lags $1,\dots,h-1$ with equal weights $1/\{2(h-1)\}$; in other words, the d.f.\ test based on $S_{n,G}$ in~\eqref{eq:Sng} is combined with $S_{n,C^{\scs (2)}}$ in~\eqref{eq:SnCh} and $S_{n,\tilde C^{\scs (3)}},\dots,S_{n,\tilde C^{(h)}}$, where the latter are the analogues of $S_{n, C^{\scs (2)}}$ but for lags $2,\dots,h-1$ (that is, they are computed from~\eqref{eq:Zi} for $q \in \{3,\dots,h\}$).
\end{itemize}
\vspace{-\medskipamount}
The above choices for the weights are arbitrary and thus clearly debatable. An ``optimal'' strategy for the choice of the weights is beyond the scope of this work. For the function $\psi$ in Sections~\ref{sec:rank} and~\ref{sec:sotests}, we only consider $\psi_F$ in~\eqref{eq:fisher} as the use of $\psi_S$ in~\eqref{eq:stouffer} sometimes gave inflated levels.

Let us now state the fundamental questions concerning the studied tests that we attempted to answer empirically by means of a large number of Monte Carlo experiments.

\paragraph{Do the studied component and combined tests maintain their level?} As is explained in detail in the supplementary material, ten strictly stationarity models, including ARMA, GARCH and nonlinear autoregressive models with either normal or Student $t$ with 4 degrees of freedom innovations, were used to generate observations under the null hypothesis of stationarity. The rank-based tests of Section~\ref{sec:rank}, that is, d, c, dc and dcp, were never found to be too liberal, while some of the second-order tests of Section~\ref{sec:sotests}, namely, v, va and mva, were found to reject stationarity too often for a particular GARCH model mimicking S\&P500 daily log-returns.

\paragraph{How do the rank-based tests of Section~\ref{sec:rank} compare to the second-order tests of Section~\ref{sec:sotests} in terms of power?} 
As presented in detail in the supplementary material, to investigate the power of the tests, eight models connected to the literature on locally stationary processes were considered alongside with four models more in line with the change-point detection literature. All tests were found to have reasonable power for at least one (and, usually, several) of the alternatives under consideration. The combined rank-based tests proposed in Section~\ref{sec:rank}, that is, dc or dcp, were found, overall, to be more powerful than the combined second-order tests, namely, va or mva, even in situations involving changes in the second-order characteristics of the underlying time series. 

\paragraph{How are the powers of the proposed component and combined tests related?} For the sake of illustration, we only focus on the component tests d and c, and the combined test dc, and consider three simple data generating models:
\vspace{-\medskipamount}
\begin{itemize}
\parskip0pt
\item[D($\sigma$) -] ``Change in the contemporary distribution only'': The $n/2$ first observations are i.i.d.\ from the $N(0,\sigma^2)$ distribution and the $n/2$ last observations are i.i.d.\ from the $N(0,1)$ distribution.

\item[S($\beta$) -] ``Change in the serial dependence only'': The $n/2$ first observations are i.i.d.\ standard normal and the $n/2$ last observations are drawn from an AR(1) model with parameter $\beta$ and centered normal innovations with variance $(1-\beta^2)$. The contemporary distribution is thus constant and equal to the standard normal.

\item[DS($\sigma$, $\beta$) -] ``Change in the contemporary distribution and the serial dependence'': The $n/2$ first observations are i.i.d.\ from the $N(0,\sigma^2)$ distribution and the $n/2$ last observations are drawn from an AR(1) model with parameter $\beta$ and $N(0,1)$ innovations.

\end{itemize}
At the 5\% significance level, the rejection percentages of the null hypothesis of stationarity computed from 1000 samples of size $n=128$ from model D($\sigma$), S($\beta$) or DS($\sigma$, $\beta$) for various values of $\sigma$ and $\beta$ are given in Table~\ref{relationship} for the tests d, c and dc for $h=2$. As one can see from the first four rows of the table, when one of the component tests has hardly any power, a ``dampening effect'' occurs for the combined test. However, when the two components tests tend to detect changes, most of the time, simultaneously, a ``reinforcement effect'' seems to occur for the combined test as can be seen from the last two rows of the table. 

\begin{table}[t!]
\centering
\caption{Percentages of rejection of the null hypothesis of stationarity computed from 1000 samples of size $n=128$ from model D($\sigma$), S($\beta$) or DS($\sigma$, $\beta$) for various values of $\sigma$ and $\beta$. The meaning of the abbreviations d, c, dc is given in Section~\ref{sec:MC}.} 
\label{relationship}
\begin{tabular}{lrrr}
  \hline
  \multicolumn{1}{c}{} & \multicolumn{3}{c}{$h=2$ or lag 1} \\ \cmidrule(lr){2-4} Model & d & c & dc \\ \hline
D(2): `Small change in contemporary dist.\ only' & 33.6 & 2.2 & 16.4 \\ 
  D(3): `Large change in contemporary dist.\ only' & 81.6 & 1.6 & 59.2 \\ 
  S(0.3): `Small change in serial dep.\ only' & 6.4 & 19.6 & 16.6 \\ 
  S(0.9): `Large change in serial dep.\ only' & 13.8 & 64.2 & 62.8 \\ 
  DS(2, 0.4): `Small change in both' & 17.2 & 28.8 & 35.4 \\ 
  DS(4, 0.7): `Large change in both' & 75.6 & 70.0 & 92.6 \\ 
   \hline
\end{tabular}
\end{table}

\paragraph{Is the combined test dc truly more powerful than a simple multivariate extension of the test d designed to be directly sensitive to departures from $H_0^{(h)}$ in~\eqref{eq:H0:Fh}?} Note that to implement the latter test for a given embedding dimension $h$, it suffices to proceed as in Section~\ref{sec:dftest} but by using the $h$-dimensional empirical d.f.s of the $h$-dimensional random vectors $\bm Y_i^{\scs (h)}$ in~\eqref{eq:Yi} instead of the one-dimensional empirical d.f.s generically given by~\eqref{eq:Gkl}. Let dh be the abbreviation of this test. To provide an empirical answer to the above question, we consider a similar setup as previously. The rejection percentages of the null hypothesis of stationarity computed from 1000 samples of size $n=128$ from model D($\sigma$), S($\beta$) or DS($\sigma$, $\beta$) for various values of $\sigma$ and $\beta$ are given in Table~\ref{joint} for the tests dc, dcp and dh for $h \in \{2,3\}$. As one can see, the test dh seems to have hardly any power when the non-stationarity is only due to a change in the serial dependence. Furthermore, even when the non-stationarity results from a change in the contemporary distribution, the test dh appears to be less powerful, overall, than the combined tests dc and dcp.

\begin{table}[t!]
\centering
\caption{Percentages of rejection of the null hypothesis of stationarity computed from 1000 samples of size $n=128$ from model D($\sigma$), S($\beta$) or DS($\sigma$, $\beta$) for various values of $\sigma$ and $\beta$. The meaning of the abbreviations dc, dcp and dh is given in Section~\ref{sec:MC}.} 
\label{joint}
\begin{tabular}{lrrrrr}
  \hline
  \multicolumn{1}{c}{} & \multicolumn{2}{c}{$h=2$} & \multicolumn{3}{c}{$h=3$ or lag 2} \\ \cmidrule(lr){2-3} \cmidrule(lr){4-6} Model & dc & dh & dc & dcp & dh \\ \hline
D(2): `Small change in contemporary dist.\ only' & 16.4 & 21.8 & 17.8 & 26.6 & 24.8 \\ 
  D(3): `Large change in contemporary dist.\ only' & 59.2 & 52.4 & 58.8 & 73.0 & 44.0 \\ 
  S(0.3): `Small change in serial dep.\ only' & 16.6 & 7.2 & 18.2 & 13.0 & 9.0 \\ 
  S(0.9): `Large change in serial dep.\ only' & 62.8 & 15.6 & 63.0 & 65.0 & 16.0 \\ 
  DS(2, 0.4): `Small change in both' & 35.4 & 20.6 & 42.2 & 34.8 & 30.0 \\ 
  DS(4, 0.7): `Large change in both' & 92.6 & 67.6 & 92.4 & 91.6 & 71.6 \\ 
   \hline
\end{tabular}
\end{table}

\paragraph{What is the influence of the choice of the embedding dimension $h$ on the empirical levels and the powers of the proposed tests?} The extensive simulations results available in the supplementary material indicate that, under the null hypothesis of stationarity, the tests c and dc tend, overall, to become more and more conservative as $h$ increases for fixed sample size~$n$. For fixed $h$, the empirical levels seem to get closer to the 5\% nominal level as $n$ increases, as expected theoretically. To convey some intuitions on the influence on $h$ on the empirical power under non-stationarity, we consider again the same setup as before and plot the rejection percentages of c and dc computed from 1000 samples of size $n=128$ from models D(2), D(3), S(0.3), S(0.9), DS(2,0.4) and DS(4,0.7) against the embedding dimension~$h$. As one can see from Figure~\ref{fig:h}, for the models under consideration, the empirical powers of the tests c and dc essentially decrease as $h$ increases. Additional simulations presented in the supplementary material and involving an AR(2) model instead of an AR(1) model for the serial dependence show that a similar pattern occurs from $h=3$ onwards in that case. Indeed, as discussed in Section~\ref{sec:hchoice}, for many models including those that were just mentioned, the power of the tests appears to be a decreasing function of $h$, at least from some small value of $h$ onwards.  

\begin{figure}[t!]
\begin{center}
\includegraphics*[width=1\linewidth]{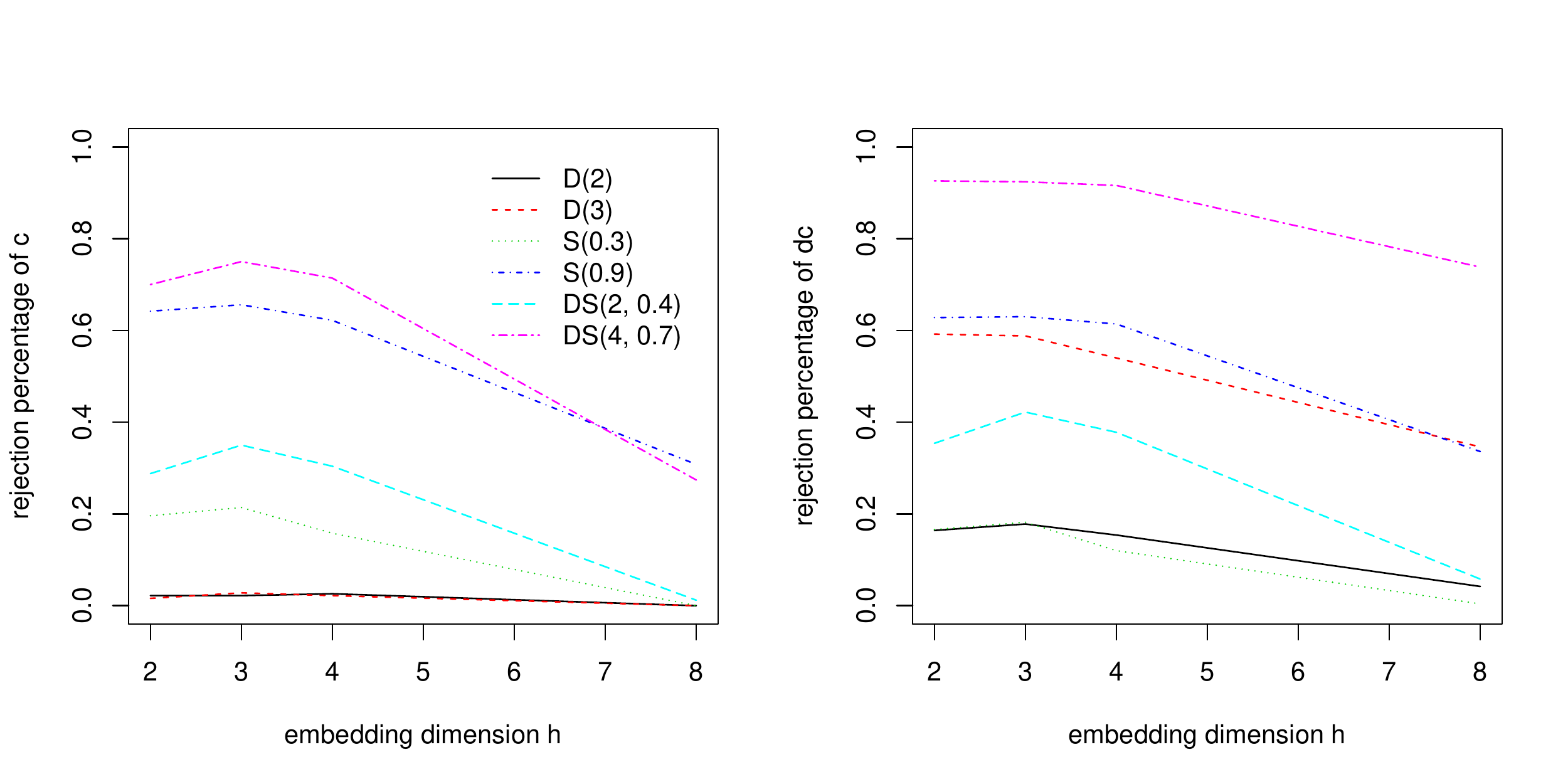}
\caption{\label{fig:h} Rejection percentages of c and dc against the embedding dimension $h \in \{2,3,4,8\}$ computed from 1000 samples of size $n=128$ from models D(2), D(3), S(0.3), S(0.9), DS(2,0.4) and DS(4,0.7).}
\end{center}
\end{figure}

\paragraph{How do the studied tests compare to existing competitors?}

As mentioned in the introduction, many tests of stationarity were proposed in the literature. Unfortunately, only a few of them seem to be implemented in statistical software. In the supplementary material, we report the results of Monte Carlo experiments investigating the finite-sample behavior of the tests of \cite{PriSub69}, \cite{Nas13} and \cite{CarNas13} that are implemented in the \textsf{R} packages \texttt{fractal} \citep{fractal}, \texttt{locits} \citep{locits} and \texttt{costat} \citep{costat}, respectively. Note that we did not consider the test of \cite{CarNas16} (implemented in the \textsf{R} package \texttt{BootWPTOS}) because we were not able to understand how to initialize the arguments of the corresponding \textsf{R} function. Under stationarity, unlike the rank-based tests d, c, dc and dcp, the three aforementioned tests were found to be too liberal for at least one of the considered models. Their behavior under the null turned out to be even more disappointing when heavy tailed innovations were used. In terms of empirical power, the results presented in the supplementary material allow in principle for a direct comparison with the results reported in \cite{CarNas13} and \cite{DetPreVet11}. Since the tests available in \textsf{R} considered in \cite{CarNas13} are far from maintaining their levels, a comparison in terms of power with these tests is clearly not meaningful. As far as the tests of \cite{DetPreVet11} are concerned, they appear, overall, to be more powerful for some of the considered models. It is however unknown whether they hold their levels when applied to stationary heavy-tailed observations as only Gaussian time series were considered in the simulations of \cite{DetPreVet11}.

\section{Illustrations}
\label{sec:illus}

By construction, the tests based on the sample mean, variance and autocovariance proposed in Section~\ref{sec:sotests} are only sensitive to changes in the second-order characteristics of a time series. The results of the simulations reported in the previous section and in the supplementary material seem to indicate that the latter tests do not always maintain their level (for instance, in the presence of conditional heteroskedasticity) and that the rank-based tests proposed in Section~\ref{sec:rank} are more powerful, even in situations only involving changes in the second-order characteristics. Therefore, we recommend the use of the rank-based tests in general.

To illustrate their application, we consider two real datasets, both available in the \textsf{R} package \texttt{copula} \citep{copula}. The first one consists of daily log-returns of Intel, Microsoft and General Electric stocks for the period from 1996 to 2000. It was used in \citet[Chapter~5]{McNFreEmb05} to illustrate the fitting of elliptical copulas. The second dataset was initially considered in \cite{GreGenGen08} to illustrate the so-called \emph{copula--GARCH} approach \citep[see, e.g.,][]{CheFan06,Pat06}. It consists of bivariate daily log-returns computed from three years of daily prices of crude oil and natural gas for the period from July 2003 to July 2006.

Prior to applying the methodologies described in the aforementioned references, it is crucial to assess whether the available data can be regarded as stretches from stationary multivariate time series. As multivariate versions of the proposed tests would need to be thoroughly investigated first (see the discussion in the next section), as an imperfect alternative, we applied the studied univariate versions to each component time series. The results are reported in Table~\ref{illus}. For the sake of simplicity, we shall ignore the necessary adjustment of p-values or global significance level due to multiple testing.

\begin{table}[t!]
\centering
\caption{Approximate p-values (multiplied by 100) of the rank-based tests of stationarity proposed in Section~\ref{sec:rank} for embedding dimension $h \in \{2,3,4\}$ applied to the component times series of the trivariate log-return data considered in \citet[Chapter~5]{McNFreEmb05} and the bivariate log-return data considered in \cite{GreGenGen08}. The daily log-returns of the Intel, Microsoft and General Electric stocks are abbreviated by INTC, MSFT and GE, respectively. The meaning of the abbreviations d, c, dc and dcp is given in Section~\ref{sec:MC}. The columns c2 and c3 report the results for the bivariate analogues of the test based on $S_{n,C^{\scs (2)}}$ defined by~\eqref{eq:SnCh} (which arise in the combined test dcp) for lags 2 and 3.}
\label{illus}
\begin{tabular}{lrrrrrrrrrrr}
  \hline
   \multicolumn{2}{c}{} & \multicolumn{2}{c}{$h=2$ or lag 1} & \multicolumn{4}{c}{$h=3$ or lag 2} & \multicolumn{4}{c}{$h=4$ or lag 3} \\ \cmidrule(lr){3-4} \cmidrule(lr){5-8} \cmidrule(lr){9-12} Variable & d & c & dc & c & dc & c2 & dcp & c & dc & c3 & dcp \\ \hline
INTC & 0.0 & 2.0 & 0.0 & 4.8 & 0.0 & 32.5 & 0.0 & 7.9 & 0.0 & 30.2 & 0.0 \\
  MSFT & 0.2 & 92.3 & 2.2 & 80.7 & 0.8 & 47.3 & 0.0 & 86.4 & 0.1 & 37.2 & 0.0 \\
  GE & 0.1 & 62.1 & 0.7 & 15.9 & 0.1 & 67.2 & 0.0 & 22.4 & 0.6 & 16.7 & 0.1 \\
   \hline \ oil & 89.6 & 22.1 & 52.5 & 55.3 & 84.0 & 46.5 & 67.8 & 89.0 & 97.2 & 5.6 & 49.0 \\
  gas & 5.0 & 16.5 & 3.9 & 17.4 & 5.4 & 90.5 & 7.4 & 43.9 & 8.8 & 85.2 & 6.2 \\
   \hline
\end{tabular}
\end{table}

As one can see from the results of the combined tests dc and dcp for embedding dimension $h \in \{2,3,4\}$, there is strong evidence against stationarity in the component series of the trivariate log-return data considered in \citet[Chapter~5]{McNFreEmb05}. For all three series, the very small p-values of the combined tests are a consequence of the very small p-value of the test d focusing on the contemporary distribution. For the Intel stock (line INTC), it is also a consequence of the small p-value of the test c for $h=2$. Although it is for instance very tempting to conclude that the non-stationarity in the log-returns of the Intel stock is due to $H_0^{\scs (1)}$ in~\eqref{eq:H0:1} and $H_{0,c}^{\scs (2)}$ in~\eqref{eq:H0:Ch} not being satisfied, such a reasoning is not formally valid without additional assumptions, as explained in the introduction. From the second horizontal block of Table~\ref{illus}, one can also conclude that there is no evidence against stationarity in the log-returns of the oil prices and only weak evidence against stationarity in the log-returns of the gas prices.

\section{Concluding remarks} \label{sec:conc}

Unlike some of their competitors that are implemented in various \textsf{R} packages, the rank-based tests of stationarity proposed in Section~\ref{sec:rank} were never observed to be too liberal for the rather typical sample sizes considered in this work. As discussed in Section~\ref{sec:hchoice}, and as empirically confirmed by the experiments of Section~\ref{sec:MC} and the supplementary material, the tests are nevertheless likely to become more conservative and less powerful as the embedding dimension $h$ is increased. 
The latter led us to make the rather general recommendation that they should be typically used with a small value of the embedding dimension $h$ such as 2, 3 or 4. It is however difficult to assess the breadth of that recommendation and it might be meaningful for the practitioner to consider the issue of the choice of $h$ in all its subtlety as attempted in the discussion of Section~\ref{sec:hchoice}.

While, unsurprisingly, the recommended tests seem to display good power for alternatives connected to the change-point detection literature, their power was not observed to be very high, overall, for the locally stationary alternatives considered in our Monte Carlo experiments. A promising approach to improve on the latter aspect would be to derive extensions of the tests allowing the comparison of blocks of observations in the spirit of \cite{HusSla01} and of \cite{KirMuh16}: once the time series is divided into moving blocks of  equal length, the main idea is to compare successive pairs of blocks by means of a statistic based on a suitable extension of the process in~\eqref{eq:Dnh} (if the focus is on serial dependence) or in~\eqref{eq:En} (if the focus is on the contemporary distribution), and to finally aggregate the statistics for each pair of blocks.

Additional future research may consist of extending the proposed tests to multivariate time series. To fix ideas, let us focus on lag $h-1$ and consider a stretch $\bm X_i=(X_{i,1}, \dots, X_{i,d})$, $i \in \{1,\dots,n+h-1\}$ from a continuous $d$-dimensional time series. A straightforward extension of the approach considered in this work is first to define the $d\times h$-dimensional random vectors $\bm Y_i^{\scs (h)} = (\bm X_i, \dots, \bm X_{i+h-1})$, $i \in \{1,\dots,n\}$.  As argued in the introduction and in Section~\ref{sec:MC}, it will then be helpful in terms of finite sample power properties to split the hypothesis  $H_0^{\scs (h)}$ in~\eqref{eq:H0:Fh} into suitable sub-hypotheses. For $A\subset \{1, \dots, d\}$ and $B\subset \{0,\dots, h-1\}$, let
\begin{align*}
H_0^{(1)}(A): &\,\exists \, G^{A} \text{ such that } (X_{1,j})_{j\in A},\dots, (X_{n-h+1,j})_{j\in A}  \text{ have d.f.\ } G^{A}, \\
H_{0,c}^{(h)}(A,B): &\,\exists \, C^{(h), A,B} \text{ such that } (X_{1+s,j})_{s\in B, j \in A},\dots,(X_{n+s,j})_{s\in B, j \in A} \text{ have copula } C^{(h), A,B}.
\end{align*}
Letting $\bar d = \{1, \dots, d\}$ and $\bar h=\{0, \dots, h-1\}$, Sklar's theorem suggests the decomposition $H_0^{\scs (h)} = H_0^{\scs (1)}(\{1\}) \cap \dots \cap H_0^{\scs (1)}(\{d\}) \cap H_{0,c}^{\scs (h)}(\bar d, \bar h)$. However, preliminary numerical experiments indicate that a straightforward extension of the approach proposed in Section~\ref{sec:comb:cop:df} to this combined hypothesis does not seem to be very powerful. The latter might be due to the curse of dimensionality identified in Section~\ref{sec:hchoice} and the fact that, under stationarity, the $d \times h$-dimensional copula $C^{\scs (h), \bar d, \bar h}$ of the $\bm Y_i^{\scs (h)}$ arising in the aforementioned decomposition does not solely control the serial dependence in the time series but also the cross-sectional dependence. As a consequence, alternative combination strategies would need to be investigated in the multivariate case. As an imperfect alternative, one might for instance consider the following hypothesis
\begin{align*}
\Big( \textstyle \bigcap_{j =1}^d H_0^{(1)}(\{j\})   \Big) \cap  \Big( H_{0,c}^{(h)}(\bar d, \{0\})\Big)  \cap \Big( \bigcap_{j =1}^d H_{0,c}^{(h)}(\{j\}, \bar h) \Big),
\end{align*}
a combined test of which would be sensible to any changes in the marginals, the contemporary dependence or the marginal serial dependence. One may easily include further hypotheses related to cross-sectional and cross-serial dependencies, like for instance $\bigcap_{i\ne j\in \bar d}H_{0,c}^{\scs (h)}(\{i,j\}, \{0,1\})$. The amount of potential adaptations appears to be very large, whence a further investigation, in particular from a finite-sample point-of-view, is beyond the scope of this paper.


\appendix

\section{Dependent multiplier sequences}
\label{app:dep}

A sequence of  random variables $(\xi_{i,n})_{i \in \Z}$ is a \emph{dependent multiplier sequence} if the three following conditions are fulfilled:
\begin{enumerate}[({$\M$}1)]
\item The sequence $(\xi_{i,n})_{i \in \Z}$ is independent of the available sample $X_1,\dots,X_{n+h-1}$ and strictly stationary with $\Ex(\xi_{0,n}) = 0$, $\Ex(\xi_{0,n}^2) = 1$ and $\sup_{n \geq 1} \Ex(|\xi_{0,n}|^\nu) < \infty$ for all $\nu \geq 1$.
\item There exists a sequence $\ell_n \to \infty$ of strictly positive constants such that $\ell_n = o(n)$ and the sequence $(\xi_{i,n})_{i \in \Z}$ is $\ell_n$-dependent, i.e., $\xi_{i,n}$ is independent of $\xi_{i+p,n}$ for all $p > \ell_n$ and $i \in \N$.
\item There exists a function $\varphi:\R \to [0,1]$, symmetric around 0, continuous at $0$, satisfying $\varphi(0)=1$ and $\varphi(x)=0$ for all $|x| > 1$ such that $\Ex(\xi_{0,n} \xi_{p,n}) = \varphi(p/\ell_n)$ for all $p \in \Z$.
\end{enumerate}
Roughly speaking, such sequences extend to the serially dependent setting the multiplier sequences that appear in the \emph{multiplier central limit theorem} \citep[see, e.g.,][Theorem 10.1 and Corollary 10.3]{Kos08}. The latter result lies at the heart of the proof of the asymptotic validity of many types of bootstrap schemes for independent observations. In particular and as it shall become clearer below, the bandwidth parameter $\ell_n$ plays a role somehow similar to the block length in the block bootstrap of \cite{Kun89}.

Two ways of generating dependent multiplier sequences are discussed in \citet[Section 5.2]{BucKoj16}. Throughout  this work, we use the so-called \emph{moving average approach} based on an initial independent and identically distributed (i.i.d.) standard normal sequence and Parzen's kernel
$$
\kappa(x) = (1 - 6x^2 + 6|x|^3) \1(|x| \leq 1/2) + 2(1-|x|)^3 \1(1/2 < |x| \leq 1), \quad x \in \R.
$$
Specifically, let $(b_n)$ be a sequence of integers such that $b_n \to \infty$, $b_n = o(n)$ and $b_n \geq 1$ for all $n \in \N$. Let $Z_1,\dots,Z_{n+2b_n-2}$ be i.i.d. $\Nc(0,1)$. Then, let $\ell_n=2b_n-1$ and, for any $j \in \{1,\dots,\ell_n\}$, let $w_{j,n} = \kappa\{(j-b_n)/b_n\}$ and $\tilde w_{j,n} = w_{j,n} ( \sum_{j'=1}^{\ell_n} w_{j',n}^2 )^{-1/2}$. Finally, for all $i \in \{1,\dots,n\}$, let
\begin{equation}
\label{eq:movave}
\xi_{i,n} = \sum_{j=1}^{\ell_n} \tilde w_{j,n} Z_{j+i-1}.
\end{equation}
Then, as verified in \citet[Section 5.2]{BucKoj16}, the infinite size version of $\xi_{1,n},\dots,\xi_{n,n}$ satisfies Assumptions~($\M 1$)-($\M 3$), when $n$ is sufficiently large.

As can be expected, the bandwidth parameter $\ell_n$ (or, equivalently, $b_n$) will have a crucial influence on the finite-sample performance of the tests studied in this work. In practice, for the rank-based (resp.\ second-order) tests of Section~\ref{sec:rank} (resp.\ Section~\ref{sec:sotests}), we apply to the available univariate sequence $X_1,\dots,X_{n+h-1}$ the data-adaptive procedure proposed in \citet[Section 5.1]{BucKoj16} \citep[resp.][Section 2.4]{BucKoj16b}, which is based on the seminal work of \cite{PapPol01}, \cite{PolWhi04} and \cite{PatPolWhi09}, among others. Roughly speaking, the latter amounts to choosing $\ell_n$ as $K_n n^{1/5}$, which asymptotically minimizes a certain integrated mean squared error, for a constant $K_n$ that can be estimated from $X_1,\dots,X_{n+h-1}$. 

Monte Carlo experiments studying the finite-sample behavior of the data-adaptive procedure of \citet[Section 5.1]{BucKoj16} for estimating the bandwidth parameter $b_n$ can be found in \citet[Section 6]{BucKoj16}. A small simulation showing how the automatically-chosen bandwidth parameter $b_n$ is affected by the strength of the serial dependence in an AR(1) model is presented in the supplementary material.

\section{Proofs}
\label{app:proofs}

\begin{proof}[Proof of Proposition~\ref{prop:combined:general}]
As we continue, we adopt the notation $\bar F^*_{T_j}(x) = \Pr(T_{n,j}^{\scs [1]} \geq x \mid \bm X_n)$, $x \in \R$, $j \in \{1,\dots,r\}$. Note in passing that the functions $\bar F^*_{T_j}$ are random and that we can rewrite $W_n$ in~\eqref{eq:Wn} as $W_n = \psi\{ \bar F_{T_1}^*(T_{n,1}), \dots, \bar F_{T_r}^*(T_{n,r}) \}$. In addition, recall that $\bar F_{T_j}(x) = \Pr(T_j \geq x)$, $x \in \R$, $j \in \{1,\dots,r\}$. Combining either~\eqref{eq:uncond:Tn} or~\eqref{eq:cond:Tn} with Lemma~2.2 in \cite{BucKoj18} and Problem 23.1 in \cite{Van98}, we obtain that
\begin{equation}
\label{eq:aeSj}
\sup_{x \in \R} | \bar F_{T_j}^*(x) - \bar F_{T_j}(x)| \p 0, \quad j \in \{1,\dots,r\}.
\end{equation}
Furthermore, Lemma~2.2 in \cite{BucKoj18} implies that~\eqref{eq:aeSj} is equivalent to
\begin{equation}
\label{eq:aeSj2}
\sup_{x \in \R} \Big|\textstyle  \frac{1}{M_n} \sum_{i=1}^{M_n} \1(T_{n,j}^{[i]} \geq x) - \bar F_{T_j}(x) \Big| \p 0, \quad j \in \{1,\dots,r\}.
\end{equation}
Again, from Lemma~2.2 in \cite{BucKoj18}, we also have that~\eqref{eq:uncond:Tn} or~\eqref{eq:cond:Tn} imply that
$$
(\bm T_n, \bm T_n^{[1]}, \dots, \bm T_n^{[N]}) \leadsto (\bm T, \bm T^{[1]}, \dots, \bm T^{[N]}),
$$
for all $N \in \N$, where $\bm T^{[1]},\dots,\bm T^{[N]}$ are independent copies of $\bm T$. Combining this last result with the continuous mapping theorem, we immediately obtain that, for any $N\in\N$,
\begin{equation}
\label{eq:jointFS}
(\bar{\bm F}_T(\bm T_n), \bar{\bm F}_T(\bm T_n^{[1]}), \dots, \bar{\bm F}_T(\bm T_n^{[N]})) \leadsto (\bar{\bm F}_T(\bm T), \bar{\bm F}_T(\bm T^{[1]}), \dots, \bar{\bm F}_T(\bm T^{[N]})),
\end{equation}
where $\bar{\bm F}_T(\bm x) = (\bar F_{T_1}(x_1),\dots,\bar F_{T_r}(x_r))$, $\bm x \in \R^r$. Combining~\eqref{eq:jointFS} with~\eqref{eq:aeSj2}, the continuity of~$\psi$ and the continuous mapping theorem, we obtain that~\eqref{eq:uncondM} holds for all $N \in \N$.

From now on, assume that $W$ has a continuous d.f. As a straightforward consequence of~\eqref{eq:aeSj} and the continuous mapping theorem, the weak convergence in \eqref{eq:jointFS} implies that, for any $N\in\N$,
\begin{equation*}
(W_n,W_n^{[1]},\dots,W_n^{[N]}) \leadsto (W,W^{[1]},\dots,W^{[N]}),
\end{equation*}
where $W_n$ is defined by~\eqref{eq:Wn} and $W_n^{\scs [i]} = \psi\{ \bar F_{T_1}^*(T_{n,1}^{\scs [i]}), \dots, \bar F_{T_r}^*(T_{n,r}^{\scs [i]}) \}, i \in \{1,\dots,N\}.$
The previous display has the following two consequences: first, by Problem 23.1 in \cite{Van98},\begin{equation}
\label{eq:unif}
\sup_{x \in \R} | \Pr(W_n \leq x) - \Pr(W \leq x)| \p 0.
\end{equation}
Second, since $W_n^{\scs [1]},\dots,W_n^{\scs [N]}$ are identically distributed and independent conditionally on the data, by Lemma~2.2 in \cite{BucKoj18}, we have that
\begin{equation}
\label{eq:cond}
\sup_{x \in \R} | \Pr(W_n^{[1]} \leq x\mid \bm X_n) - \Pr(W_n \leq x)| \p 0.
\end{equation}

Let us next prove~\eqref{eq:condM}. In view of~\eqref{eq:cond}, 
 it suffices to show that
\begin{equation}
\label{eq:goal}
\sup_{x \in\R} | \Pr(W_{n,M_n}^{[1]} \le x \mid \bm X_n) -  \Pr(W_n^{[1]} \leq x \mid \bm X_n) | \p 0.
\end{equation}
Using the fact that, for any $a,b,x\in\R$ and $\eps>0$,
\begin{equation}
\label{eq:tool}
| \1(a \le x) - \1(b \le x) | \le \1(|x-a| \le \eps) +\1(|a-b|> \eps),
\end{equation}
we have that
\begin{multline*}
\sup_{x \in\R} \big | \Pr (W_n^{[1]} \le x \mid \bm X_n) -  \Pr(W_{n,M_n}^{[1]} \leq x \mid \bm X_n) \big |
\le
\sup_{x \in\R} \Pr( |W_n^{[1]} - x | \le \eps \mid \bm X_n )  \\
+  \Pr( | W_n^{[1]} - W_{n,M_n}^{[1]} | > \eps \mid \bm X_n ).
\end{multline*}
From~\eqref{eq:unif} and~\eqref{eq:cond}, $\sup_{x \in \R} \Pr( |W_n^{\scs [1]} - x | \le \eps\mid \bm X_n)$ converges in probability to $\sup_{x \in \R} \Pr ( |W - x | \le \eps )$ which can be made arbitrary small by decreasing $\eps$.
From~\eqref{eq:aeSj},~\eqref{eq:aeSj2},~\eqref{eq:jointFS} and the continuous mapping theorem, we obtain that $W_n^{\scs [1]} - W_{\scs n,M_n}^{\scs [1]} =o_\Pr(1)$, which implies that
\begin{equation}
\label{eq:tool2}
 \Pr( | W_n^{[1]} - W_{n,M_n}^{[1]} | > \eps ) = \Ex \{ \Pr( | W_n^{[1]} - W_{n,M_n}^{[1]} | > \eps \mid \bm X_n ) \} \to 0,
\end{equation}
and thus that $\Pr( | W_n^{\scs [1]} - W_{\scs n,M_n}^{\scs [1]} | > \eps \mid \bm X_n ) =o_\Pr(1)$. Hence,~\eqref{eq:goal} holds and thus so does~\eqref{eq:condM}.

Finally, let show that~\eqref{eq:MC} holds. Since $W_n^{\scs [1]},\dots,W_n^{\scs [N]}$ are identically distributed and independent conditionally on the data, by Lemma~2.2 in \cite{BucKoj18}, we have that~\eqref{eq:cond} implies
\begin{equation}
\label{eq:MCunobs}
\sup_{x \in \R} \bigg| \frac{1}{M_n} \sum_{i=1}^{M_n} \1(W_n^{[i]} \le x) - \Pr(W_n \leq x)\bigg| \p 0.
\end{equation}
Whence~\eqref{eq:MC} is proved if we show that
\begin{equation}
\label{eq:goal2}
\sup_{x \in\R} \bigg| \frac{1}{M_n} \sum_{i=1}^{M_n} \1(W_{n,M_n}^{[i]} \le x)  - \frac{1}{M_n} \sum_{i=1}^{M_n} \1(W_n^{[i]} \le x) \bigg| \p 0.
\end{equation}
Using again~\eqref{eq:tool}, the term on the left of the previous display is smaller than
$$
\sup_{x \in\R} \frac{1}{M_n} \sum_{i=1}^{M_n} \1( | W_n^{[i]} - x| \le \eps) +  \frac{1}{M_n} \sum_{i=1}^{M_n} \1( | W_n^{[i]} - W_{n,M_n}^{[i]} | \ge \eps).
$$
From~\eqref{eq:MCunobs} and~\eqref{eq:unif}, the first term converges in probability to $\sup_{x \in \R} \Pr ( |W - x | \le \eps )$ which can be made arbitrary small by decreasing $\eps$. The second term converges in probability to zero by Markov's inequality: for any $\lambda > 0$,
\begin{align*}
\Pr\left\{ \frac{1}{M_n} \sum_{i=1}^{M_n} \1( | W_n^{[i]} - W_{n,M_n}^{[i]} | \ge \eps) > \lambda \right\}
&\leq
\lambda^{-1} \Ex \left\{  \frac{1}{M_n} \sum_{i=1}^{M_n} \1( | W_n^{[i]} - W_{n,M_n}^{[i]} | \ge \eps) \right\} \\
&\leq
\lambda^{-1} \Pr( | W_n^{[1]} - W_{n,M_n}^{[1]} | \ge \eps) \to 0
\end{align*}
since the $W_n^{\scs [i]} - W_{\scs n,M_n}^{\scs [i]}$ are identically distributed and by~\eqref{eq:tool2}. Therefore,~\eqref{eq:goal2} holds and, hence, so does~\eqref{eq:MC}. Note that, from the fact that $\bm T$ and $W$ have continuous d.f.s, we could have alternatively proved the analogue statement with `$\le$' replaced by `$<$'. As a consequence, we immediately obtain that $p_{n,M_n}(W_{\scs n,M_n}^{\scs [0]})$ has the same weak limit as $\bar F_{W_n}(W_{\scs n,M_n}^{\scs [0]})$, where $\bar F_{W_n}(w) = \Pr(W_n \ge w)$, $w \in \R$. By the analogue to~\eqref{eq:unif} with `$\le$' replaced by `$<$', the latter has the same asymptotic distribution as $\bar F_W(W_{\scs n,M_n}^{\scs [0]})$, where $\bar F_W(w) = \Pr(W \ge w)$, $w \in \R$. By the weak convergence $W_{\scs n,M_n}^{\scs [0]} \leadsto W$ following from~\eqref{eq:uncondM} and the continuous mapping theorem, $\bar F_W(W_{\scs n,M_n}^{\scs [0]})$ is asymptotically standard uniform.
\end{proof}

\begin{proof}[Proof of Proposition~\ref{prop:combined:alternative}]
Notice first that assumption~$(ii)$ implies that the corresponding approximate p-value $p_{n,M_n}(T_{n,j_0}^{\scs [0]})$ given by~\eqref{eq:pval_T} converges to zero in probability. Indeed,
$$
\E\{ p_{n,M_n}(T_{n,j_0}^{[0]}) \} = \frac1{M_n+1} \left\{ \frac12 + \sum_{k=1}^{M_n} \Pr(T_{n,j_0}^{[k]} \ge T_{n,j_0}^{[0]}) \right\} = \Pr(T_{n,j_0}^{[1]} \ge T_{n,j_0}^{[0]}) + O(M_n^{-1}).
$$
Next, a consequence of assumption~$(iii)$ is that, for any $j \in \{1, \dots, r\}$,
\[
\big( p_{n,M_n}(T_{n,j}^{[1]}) , \dots, p_{n,M_n}(T_{n,j}^{[M_n]})  \big)
\]
is a permutation of the vector 
\[
\big( \tfrac{3}{2M_n + 2} , \dots, \tfrac{2M_n + 1}{2M_n+ 2}   \big).
\]
It follows that, for any $x \in (0,1)$,
\begin{equation}
  \label{eq:emp:df}
\frac{1}{M_n} \sum_{k=1}^{M_n} \1 \big\{ p_{n,M_n}(T_{n,j}^{[k]}) \leq x \big\} = \frac{1}{M_n} \sum_{k=1}^{M_n} \1 \big( \tfrac{2k + 1}{2M_n+ 2} \leq x \big) = x + O(M_n^{-1}), 
\end{equation}
where $\ip{.}$ is the floor function. Then, let $\bar w = \max_{j \in \{1,\dots,r\}} w_j$. Starting from \eqref{eq:pval_W}, and relying on assumptions $(i)$ and $(iii)$, we successively obtain
\begin{align*}
p_{n,M_n}(W_{n,M_n}^{[0]})
&=
\frac1{M_n} \sum_{k=1}^{M_n} \1 \Big[ \sum_{j=1}^r w_j \varphi \big\{ p_{n,M_n}(T_{n,j}^{[k]}) \big\} \ge \sum_{j=1}^r  w_j \varphi \big\{ p_{n,M_n}(T_{n,j}^{[0]}) \big\} \Big] \\
&\le 
\frac1{M_n} \sum_{k=1}^{M_n}  \1 \Big[ \bar w \sum_{j=1}^r \varphi \big\{ p_{n,M_n}(T_{n,j}^{[k]}) \big\} \ge  w_{j_0} \varphi \big\{ p_{n,M_n}(T_{n,j_0}^{[0]}) \big\} \Big] \\
&\le 
\frac1{M_n} \sum_{k=1}^{M_n} \1 \Big[ \exists\, j \in \{1, \dots, r \} : r \bar w \varphi \big\{ p_{n,M_n}(T_{n,j}^{[k]}) \big\} \ge  w_{j_0} \varphi \big\{ p_{n,M_n}(T_{n,j_0}^{[0]}) \big\} \Big] \\
&\le 
\frac1{M_n} \sum_{k=1}^{M_n}  \sum_{j=1}^r \1 \Big[ \varphi \big\{ p_{n,M_n}(T_{n,j}^{[k]}) \big\} \ge \tfrac{w_{j_0}}{r \bar w} \varphi \big\{ p_{n,M_n}(T_{n,j_0}^{[0]}) \big\} \Big] \\
  &=                                                                                     \frac{r}{M_n} \sum_{k=1}^{M_n}  \1 \Big[ \varphi \big( \tfrac{2k+1}{2M_n+ 2} \big) \ge \tfrac{w_{j_0}}{r \bar w} \varphi \big\{ p_{n,M_n}(T_{n,j_0}^{[0]}) \big\} \Big]
  \\
  &=                                                                                     \frac{r}{M_n} \sum_{k=1}^{M_n}  \1 \left( \tfrac{2k+1}{2M_n+ 2} \le \varphi^{-1} \Big[ \tfrac{w_{j_0}}{r \bar w} \varphi \big\{ p_{n,M_n}(T_{n,j_0}^{[0]}) \big\} \Big] \right) \\
  &= r \varphi^{-1} \Big[ \tfrac{w_{j_0}}{r \bar w} \varphi \big\{ p_{n,M_n}(T_{n,j_0}^{[0]}) \big\} \Big] +O(M_n^{-1}) \p 0,
\end{align*}
where the last statement follows from~\eqref{eq:emp:df} and the fact that $p_{n,M_n}(T_{n,j_0}^{\scs [0]}) \p 0$.
\end{proof}

\begin{proof}[Proof of Proposition~\ref{prop:weak_Cnh_sm}]
The result is a consequence of Proposition~3.3 in \cite{BucKojRohSeg14} and the fact that the strong mixing coefficients of the sequence $(\bm Y_i^{\scs (h)})_{i \in \Z}$ defined through~\eqref{eq:Yi} can be expressed from those of the sequence $(X_i)_{i \in \Z}$ as $\alpha_r^{\bm Y} = \alpha_{\scs (r-h+1) \vee 0}^{X}, r \in \N$, where $\vee$ is the maximum operator.
\end{proof}

\begin{proof}[Proof of Proposition~\ref{prop:acs}]
The assertions concerning weak convergence are simple consequences of the continuous mapping theorem and Proposition~\ref{prop:weak_Cnh_sm}. It remains to show that $\Lc(S_{C^{(h)}})$, the distribution of $S_{C^{(h)}}$, is absolutely continuous with respect to the Lebesgue measure. For that purpose, note that, with probability one, the sample paths of $\D_{C^{(h)}}$ are elements of $\mathcal C([0,1]\times [0,1]^{h})$, the space of continuous real-valued functions on $[0,1]\times [0,1]^{h}$. We may write $S_{C^{(h)}}= \{ f(\D_{C^{(h)}}) \}^2$, where
\[
f:\mathcal C([0,1]^{h+1}) \to \R, \quad f(g) = \sup_{s \in [0,1]}  \left\{ \int_{[0, 1]^h} g^2(s,\bm{u})  \, \dd C^{(h)}(\bm{u}) \right\}^{1/2},
\]
and it is sufficient to show that  $\Lc \{ f(\D_{C^{(h)}}) \}$ is absolutely continuous. Now, if $\mathcal C([0,1]^{h+1})$ is equipped with the supremum norm $\|\cdot\|_\infty$, then $f$ is continuous and convex. We may hence apply Theorem 7.1 in \cite{DavLif84}:  $\Lc \{ f(\D_{C^{(h)}}) \}$ is concentrated on $[a_0,\infty)$ and absolutely continuous on $(a_0, \infty)$, where
\[
a_0 = \inf\{ f(g) : g \text{ belongs to the support of }  \Lc (\D_{C^{(h)}}) \}.
\]
It hence remains to be shown that $\Lc \{ f(\D_{C^{(h)}}) \}$ has no atom at $a_0$.
First of all, note that $a_0=0$. Indeed, by Lemma 1.2(e) in \cite{DerFehMatSch03}, we have $\Pr(\|\D_{C^{(h)}}\|_\infty\le \eps) >0$ for any $\eps>0$. Hence, for any $\eps>0$, there exist functions $g$ in the support of the distribution of $\D_{C^{(h)}}$  such that $f(g)\le \eps$, whence $a_0=0$ as asserted. Moreover, $f(\D_{C^{(h)}})=0$ holds if and only if $\mathbb D_{C^{(h)}}(s,\bm u) = 0$ for any $s\in[0,1]$ and any $\bm u$ in the support of the distribution induced by $C^{(h)}$ (by continuity of the sample paths). Then, choose an arbitrary point $\bm u^*$ in the latter support such that $\sigma^2 = \Var\{\mathbb C_{C^{(h)}}(0,1,\bm u^*)\}>0$. A straightforward calculation shows that $\mathbb C_{C^{(h)}}(0,1/2, \bm u^*)$ and $\mathbb C_{C^{(h)}}(1/2,1, \bm u^*)$ are uncorrelated and have the same variance $\tfrac12 \sigma^2$.
 Hence,
\begin{align*}
\Var \{ \mathbb D_{C^{(h)}}(\tfrac12, \bm u^*) \}
&=
\Var\{ \tfrac12 \mathbb C_{C^{(h)}}(0,1/2, \bm u^*) - \tfrac12   \mathbb C_{C^{(h)}}(1/2,1 ,\bm u^*) \} \\
&=
\tfrac14 \Var\{ \mathbb C_{C^{(h)}}(0,1/2, \bm u^*) \} + \tfrac14   \Var \{ \mathbb C_{C^{(h)}}(1/2,1, \bm u^*) \}
= \tfrac14 \sigma^2 >0.
\end{align*}
As  consequence,
$
\Pr(f(\D_{C^{(h)}})=0) \le \Pr(\mathbb D_{C^{(h)}}(\tfrac12 ,\bm u^*) = 0 ) = 0,
$
which finally implies that $\Lc(f(\D_{C^{(h)}}))$ and therefore  $\Lc(S^{\scs (h)})$  is absolutely continuous.
\end{proof}

\begin{proof}[Proof of Proposition~\ref{prop:S_mult}]
The result is a consequence of Proposition~4.2 in \cite{BucKojRohSeg14} and the fact that the strong mixing coefficients of the sequence $(\bm Y_i^{\scs (h)})_{i \in \Z}$ can be expressed as $\alpha_r^{\bm Y} = \alpha_{\scs (r-h+1) \vee 0}^{X}, r \in \N$.
\end{proof}

\begin{proof}[Proof of Proposition~\ref{prop:weak_Gn_sm}]
The assertions concerning weak convergence are simple consequences of Theorem~1 of \cite{Buc15} and of the continuous mapping theorem. Absolute continuity of $S_G$ can be shown along similar lines as for $S_{C^{(h)}}$ in Proposition~\ref{prop:acs}.
\end{proof}

\begin{proof}[Proof of Proposition~\ref{prop:combined1}]
To prove the first claim, one first needs to show that the finite-dimensional distributions of $\big(\D_{n,C^{(h)}}, \hat{\D}_{\scs n,C^{(h)}}^{\scs [1]}, \dots, \hat{\D}_{\scs n,C^{(h)}}^{\scs [M]}, \E_n, \hat{\E}_n^{\scs [1]}, \dots, \hat{\E}_n^{\scs [M]} \big)$ converge weakly to those of $\big(\D_{C^{(h)}}, \D_{\scs C^{(h)}}^{\scs [1]}, \dots, \D_{\scs C^{(h)}}^{\scs [M]}, \E, \E^{[1]}, \dots, \E^{[M]} \big)$. The proof is a more notationally involved version of the proof of Lemma~A.1 in \cite{BucKoj16}. Joint asymptotic tightness follows from Proposition~\ref{prop:S_mult} as well as from the fact that, for any $m \in \N$, $\hat{\E}_n^{\scs [m]} \leadsto \E^{[m]}$ in $\ell^\infty([0,1] \times \R)$ as a consequence of Corollary~2.2 in \cite{BucKoj16} and the continuous mapping theorem.
\end{proof}

\section*{Acknowledgments}

The authors would like to thank two anonymous Referees and a Co-Editor for their constructive and insightful comments on an earlier version of this manuscript. Axel Bücher gratefully acknowledges support by the Collaborative Research Center ``Statistical modeling of nonlinear dynamic processes'' (SFB 823) of the German Research Foundation. Parts of this paper were written when Axel Bücher was a postdoctoral researcher at Ruhr-Universität Bochum, Germany. Jean-David Fermanian's work was supported by the grant ``Investissements d’Avenir'' (ANR-11-IDEX0003/Labex Ecodec) of the French National Research Agency.

\bibliographystyle{chicago}
\bibliography{biblio}


\newpage
\thispagestyle{empty}

\begin{center}
{\LARGE Supplementary material for  \\ [2mm] ``Combining cumulative sum change-point detection tests for assessing the stationarity of univariate time series''}
\vspace{1cm}

{\large Axel B\"ucher\footnote{Heinrich-Heine-Universität D\"usseldorf,
Mathematisches Institut,
Universit\"atsstr.~1, 40225 D\"usseldorf, Germany.
{E-mail:} \texttt{axel.buecher@hhu.de}}
\hspace{1cm}
Jean-David Fermanian\footnote{CREST-ENSAE, J120, 3, avenue Pierre-Larousse, 92245 Malakoff cedex, France. {E-mail:} \texttt{jean-david.fermanian@ensae.fr}}
\hspace{1cm}
Ivan Kojadinovic\footnote{CNRS / Universit\'e de Pau et des Pays de l'Adour, Laboratoire de math\'ematiques et applications -- IPRA, UMR 5142, B.P. 1155, 64013 Pau Cedex, France.
{E-mail:} \texttt{ivan.kojadinovic@univ-pau.fr}}

\vspace{.2cm} \today 
\vspace{1cm}
}
\end{center}

\begin{abstract}
After providing additional results stating conditions under which the procedure for combining dependent tests described in Section~\ref{sec:combine:tests} is consistent, we briefly illustrate the data-adaptive procedure used to estimate the bandwidth parameter arising in dependent multiplier sequences and report the results of additional Monte Carlo experiments investigating the finite-sample performance of the tests that were proposed in Sections~\ref{sec:rank} and~\ref{sec:sotests}.  
\end{abstract}


\section{Additional results on the consistency of the procedure for combining dependent tests}

The following result is an analogue of Proposition~\ref{prop:combined:alternative} that allows one to consider the function~$\psi_S$ in~\eqref{eq:stouffer} as combining function $\psi$ in the proposed global testing procedure.

\begin{prop}
  Let $M=M_n\to\infty$ as $n\to\infty$. Assume that
  \begin{enumerate}[(i)]

  \item the combining function $\psi$ is of the form 
    \[
      \psi(p_1, \dots, p_r) = \sum_{j=1}^r w_j \varphi(p_j),
    \]
    where $\varphi$ is decreasing, non-negative on $(0,1/2)$ and one-to-one from $(0,1)$ to $(-\infty,\infty)$,
    
  \item the global statistic $W_{n,M_n}^{[0]}$ diverges to infinity in probability, 

  \item for any $j \in \{1,\dots,r\}$, the sample of bootstrap replicates $T_{n,j}^{\scs [1]}, \dots, T_{n,j}^{\scs [M_n]}$ does not contain ties.
  \end{enumerate}
  Then, the approximate p-value  $p_{n,M_n}(W_{\scs n,M_n}^{\scs [0]})$ of the global test converges to zero in probability, where $p_{n,M_n}(W_{\scs n,M_n}^{\scs [0]})$ is defined by~\eqref{eq:pval_W}.
\end{prop}

\begin{proof}
Let $\bar w = \max_{j \in \{1,\dots,r\}} w_j$ and let $\varphi^+ = \varphi \vee 0$, where $\vee$ is the maximum operator. Starting from \eqref{eq:pval_W}, and relying on assumptions $(i)$ and $(iii)$, we successively obtain
\begin{align*}
p_{n,M_n}(W_{n,M_n}^{[0]})
&=
\frac1{M_n} \sum_{k=1}^{M_n} \1 \Big[ \sum_{j=1}^r w_j \varphi \big\{ p_{n,M_n}(T_{n,j}^{[k]}) \big\} \ge W_{n,M_n}^{[0]} \big\} \Big] \\
&\le 
\frac1{M_n} \sum_{k=1}^{M_n}  \1 \Big[ \bar w \sum_{j=1}^r \varphi^+ \big\{ p_{n,M_n}(T_{n,j}^{[k]}) \big\} \ge  W_{n,M_n}^{[0]} \Big] \\
&\le 
\frac1{M_n} \sum_{k=1}^{M_n} \1 \Big[ \exists\, j \in \{1, \dots, r \} : r \bar w \varphi^+ \big\{ p_{n,M_n}(T_{n,j}^{[k]}) \big\} \ge  W_{n,M_n}^{[0]} \Big] \\
  &\le \frac1{M_n} \sum_{k=1}^{M_n}  \sum_{j=1}^r \1 \Big[ \varphi^+ \big\{ p_{n,M_n}(T_{n,j}^{[k]}) \big\} \ge (r \bar w)^{-1} W_{n,M_n}^{[0]} \Big] \\
  &=                                                                                     \frac{r}{M_n} \sum_{k=1}^{M_n}  \1 \Big\{ \varphi^+ \big( \tfrac{2k+1}{2M_n+ 2} \big) \ge (r \bar w)^{-1} W_{n,M_n}^{[0]} \Big\}.
\end{align*}
Using assumption~$(i)$, we have that, for $x > 0$, 
\begin{align*}
  \frac1{M_n} \sum_{k=1}^{M_n}  \1 \Big\{ \varphi^+ \big( \tfrac{2k+1}{2M_n+ 2} \big) \ge x \Big\} &= \frac1{M_n} \sum_{k=1}^{\ip{M_n/2}}  \1 \Big\{ \varphi \big( \tfrac{2k+1}{2M_n+ 2} \big) \ge x \Big\} + O(M_n^{-1}) \\
                                                                                                   &= \frac1{M_n} \sum_{k=1}^{\ip{M_n/2}}  \1 \Big\{ \tfrac{2k+1}{2M_n+ 2} \le \varphi^{-1}(x) \Big\} + O(M_n^{-1}) \\
  &= \varphi^{-1}(x)/2 + O(M_n^{-1})
\end{align*}
since $\varphi^{-1}(x) \in (0,1/2)$. As a consequence, we obtain that
$$
p_{n,M_n}(W_{n,M_n}^{[0]}) = r \varphi^{-1} \big\{ (r \bar w)^{-1} W_{n,M_n}^{[0]} \big\} /2 + O(M_n^{-1}) \p 0.
$$
\end{proof}

The next results provides sufficient conditions so that assumption~$(ii)$ of the previous proposition is satisfied. 

\begin{prop}
  Let $M=M_n\to\infty$ as $n\to\infty$. Assume that
  \begin{enumerate}[(i)]
    
  \item the combining function $\psi$ is of the form 
    \[
      \psi(p_1, \dots, p_r) = \sum_{j=1}^r w_j \varphi(p_j),
    \]
    where $\varphi$ is decreasing, non-negative on $(0,1/2)$ and one-to-one from $(0,1)$ to $(-\infty,\infty)$,

  \item there exists $r_0 \in \{1,\dots,r-1\}$ such that
    \begin{enumerate}[(a)]
    \item $H_0^{(1)} \cap \dots \cap H_0^{(r_0)}$ holds, $(T_{n,1},\dots,T_{n,r_0})$ converges weakly to a random vector having a continuous d.f., and
    $$
    \sup_{\bm x \in \R^{r_0}} | \Pr(T_{n,1}^{[1]} \leq x_1, \dots, T_{n,r_0}^{[1]} \leq x_{r_0} \mid \bm X_n) -  \Pr(T_{n,1} \leq x_1, \dots, T_{n,r_0} \leq x_{r_0})|  \p 0,
    $$
    \item for any $j \in \{r_0+1,\dots,r\}$, $H_0^{(j)}$ does not hold and $\Pr(T_{n,j}^{\scs [1]} \ge T_{n,j})$ converges to zero.
    \end{enumerate}
    
  \end{enumerate}
  Then, $W_{n,M_n}^{\scs [0]}$ diverges to infinity in probability. 
\end{prop}

\begin{proof}
  As a consequence of assumption~$(ii)$~$(a)$ and Proposition~\ref{prop:combined:general}, we have that the random variable $\sum_{j=1}^{r_0} w_j \varphi\{p_{n,M_n}(T_{n,j}^{\scs [0]})\}$ converges in distribution, and thus, that
  \begin{equation}
    \label{eq:bounded:proba}
    \sum_{j=1}^{r_0} w_j \varphi\{p_{n,M_n}(T_{n,j}^{\scs [0]})\} = O_\Pr(1).
  \end{equation}
  From assumption~$(ii)$~$(b)$, proceeding as in the proof of Proposition~\ref{prop:combined:alternative}, we immediately obtain that $p_{n,M_n}(T_{n,j}^{\scs [0]})$ converges to zero in probability for all $j \in \{r_0+1,\dots,r\}$, which combined with assumption~$(i)$ and the continuous mapping theorem implies that
  \begin{equation}
    \label{eq:div:proba}
    \sum_{j=r_0+1}^{r} w_j \varphi\{p_{n,M_n}(T_{n,j}^{\scs [0]})\} \p \infty.
  \end{equation}
  The desired result follows from~\eqref{eq:bounded:proba} and~\eqref{eq:div:proba}.
\end{proof}

\section{Data-adaptive bandwidth parameter of dependent multiplier sequences for an AR(1) model}

As explained in Appendix~\ref{app:dep}, the bandwidth parameter $\ell_n$ (or, equivalently, $b_n$) arising in the generation of dependent multiplier sequences through~\eqref{eq:movave} will have a crucial influence on the finite-sample performance of the tests studied in this work. In this section, we conduct a small simulation to illustrate the finite-sample properties of the data-adaptive procedure proposed in \citet[Section 5.1]{BucKoj16}. For $\beta \in \{0,0.1,\dots,0.9\}$, we generate 1000 samples of size $n=128$ from an AR(1) model with standard normal innovations and parameter~$\beta$. From each sample, we estimate the bandwidth parameter $b_n$. The mean and standard deviations of the 1000 estimates are represented in Figure~\ref{fig:b} against the value of $\beta$. As can be seen, the stronger the serial dependence, the larger $b_n$ is on average, suggesting that $b_n$ (or, equivalently, $\ell_n$) do indeed play a role similar to the block length in the block bootstrap of \cite{Kun89}.

\begin{figure}[t!]
\begin{center}
\includegraphics*[width=0.5\linewidth]{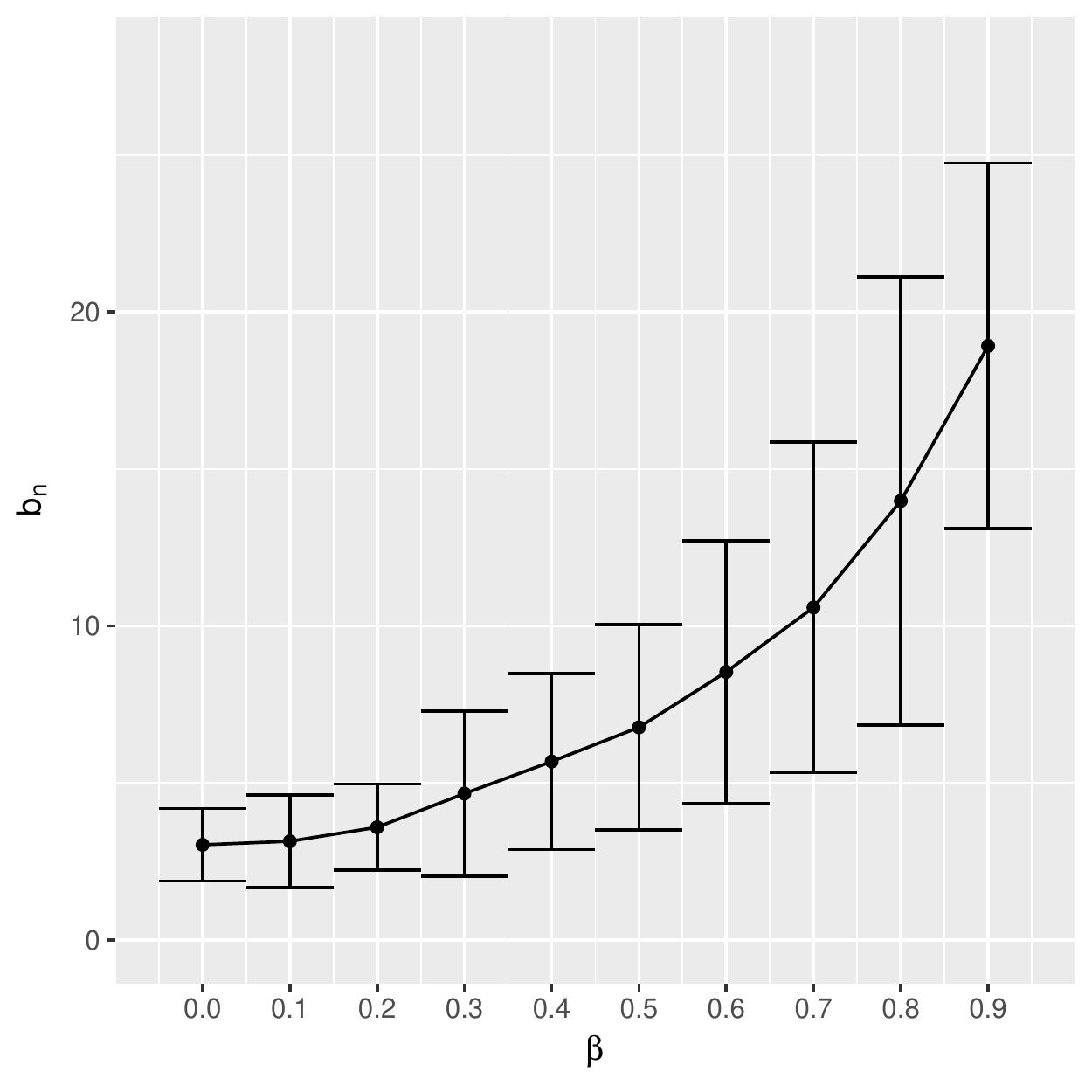}
\caption{\label{fig:b} Mean and standard deviation of 1000 estimates of the bandwidth parameter $b_n$ arising in dependent multiplier sequences computed from samples of size $n=128$ from an AR(1) model with standard normal innovations and parameter $\beta$.}
\end{center}
\end{figure}

\section{Additional Monte Carlo experiments}

After describing the main data generating models for studying the finite-sample properties of the proposed tests, we investigate in detail the behavior of some competitor tests available in \textsf{R}, study how well the proposed tests hold their level and finally assess their power under various alternatives, some belonging to the locally stationary process literature, others more in line with the change-point detection literature.

\subsection{Data generating processes}

The following ten strictly stationary models were used to generate observations under the null hypothesis of stationarity. Either standard normal or standardized Student $t$ with 4 degrees of freedom innovations were considered (standardization refers to the fact that the Student $t$ with 4 degrees of freedom distribution was rescaled to have variance one). The first seven models were considered in \cite{Nas13} (and are denoted by S1--S7 therein):
\vspace{-\medskipamount}
\begin{enumerate}[{N}1 -][10]\setcounter{enumi}{0}\parskip0pt
\item i.i.d.\ observations from the innovation distribution.
\item AR(1) model with parameter $0.9$.
\item AR(1) model with parameter $-0.9$.
\item MA(1) model with parameter $0.8$.
\item MA(1) model with parameter $-0.8$.
\item ARMA(1, 0, 2) with the AR coefficient $-0.4$, and the MA coefficients $(-0.8, 0.4)$.
\item AR(2) with parameters $1.385929$ and $-0.9604$. This process is stationary, but close to the ``unit root'': a ``rough'' stochastic process with spectral peak near $\pi/4$.
\item GARCH(1,1) model with parameters $(\omega,\beta,\alpha)=(0.012,0.919,0.072)$. The latter values were estimated by \cite{JonPooRoc07} from SP500 daily log-returns.
\item the exponential autoregressive model considered in \cite{AueTjo90} whose generating equation is
\begin{equation*}
X_{t} = \{ 0.8 - 1.1 \exp ( - 50 X_{t-1}^2 ) \} X_{t-1} + 0.1 \eps_{t}.
\end{equation*}
\item the nonlinear autoregressive model used in \citet[Section 3.3]{PapPol01} whose generating equation is
\begin{equation*}
X_t = 0.6 \sin( X_{t-1} ) + \eps_{t}.
\end{equation*}
\end{enumerate}
\vspace{-\medskipamount}
For all these models, a burn-in period of 100 observations was used.

To simulate observations under the alternative hypothesis of non-stationarity, models connected to the literature on locally stationary processes were considered first. The first four are taken from \cite{DetPreVet11} and were used to generate univariate series $X_{1,n},\dots,X_{n,n}$ of length $n \in \{128,256,512\}$ by means of the following equations:
\vspace{-\medskipamount}
\begin{enumerate}[{A}1 -][10]\setcounter{enumi}{0}\parskip0pt

\item $X_{t,n} = 1.1 \cos\{1.5 - \cos(4 \pi t / n)\} \eps_{t-1} + \eps_t$,

\item $X_{t,n} = 0.6 \sin(4 \pi t / n) X_{t-1,n} + \eps_t$,

\item $X_{t,n} = (0.5 X_{t-1,n} + \eps_t) \1(t \in \{1,\dots,n/4\} \cup \{3n/4+1,\dots,n\}) + (-0.5 X_{t-1,n} + \eps_t) \1(t \in \{n/4+1,\dots,3n/4\})$,

\item $X_{t,n} = (-0.5 X_{t-1,n} + \eps_t) \1(t \in \{1,\dots,n/2\} \cup \{n/2+n/64+1,\dots,n\}) + 4 \eps_t \1(t \in \{n/2+1,\dots,n/2+n/64\})$,
\end{enumerate}
\vspace{-\medskipamount}

\noindent where $\eps_0,\eps_1,\dots,\eps_n$ are i.i.d.\ standard normal and with the convention that $X_{0,n} = 0$. The next four models under the alternative were considered in \cite{Nas13} (and are denoted P1--P4 therein), the last three ones being locally stationary wavelet (LSW) processes \citep[see, e.g., Equation~(1) in][]{Nas13}:
\vspace{-\medskipamount}
\begin{enumerate}[{A}1 -][10]\setcounter{enumi}{4}\parskip0pt
\item A time-varying AR model $X_t = \alpha_t X_{t-1} + \eps_t$ with i.i.d.\ standard normal innovations and an AR parameter evolving linearly from 0.9 to $-0.9$ over the $n$ observations.
\item A LSW process based on Haar wavelets with spectrum $S_j(z) = 0$ for $j > 1$ and $S_1 (z) = 1/4 - (z - 1/2 )^2$ for $z \in (0, 1)$. This process is a time-varying moving average process.
\item A LSW process based on Haar wavelets with spectrum $S_j (z) = 0$ for $j > 2$, $S_1 (z)$ as for A6 and $S_2 (z) = S_1 (z + 1/2 )$ using periodic boundaries (for the construction of the spectrum only).
\item A LSW process based on Haar wavelets with spectrum $S_j (z) = 0$ for $j = 2$ and $j > 4$. Moreover, $S_1 (z) = \exp\{-4(z - 1/2 )^2 \}$, $S_3 (z) = S_1(z - 1/4 )$ and $S_4 (z) = S_1(z + 1/4 )$, again assuming periodic boundaries.
\end{enumerate}
\vspace{-\medskipamount}

\noindent Models A1--A8 considered thus far are connected to the literature on locally stationary processes. In a second set of experiments, we focused on models that are more in line with the change-point detection literature:
\vspace{-\medskipamount}
\begin{enumerate}[{A}1 -][10]\setcounter{enumi}{8}\parskip0pt
\item An AR(1) model with one break: the $n/2$ first observations are i.i.d.\ from the innovation distribution (standard normal or standardized $t_4$) and the $n/2$ last observations are from an AR(1) model with parameter $\beta \in \{-0.8,-0.4,0,0.4,0.8\}$.
\item An AR(2) model with one break: the $n/2$ first observations are i.i.d.\ from the innovation distribution (standard normal or standardized $t_4$) and the $n/2$ last observations are from an AR(2) model with parameter $(0,\beta)$ with $\beta \in \{-0.8,-0.4,0,0.4,0.8\}$.
\end{enumerate}
\vspace{-\medskipamount}

\noindent Note that both the contemporary distribution and the serial dependence are changing under these scenarios (unless $\beta=0$).  Also note that there is no relationship between $X_t$ and $X_{t-1}$ for Model A10, that is, $H_{0,c}^{\scs (2)}$ in~\eqref{eq:H0:Ch} is met with $C^{\scs (2)}$ the bivariate independence copula.

Finally, we considered two simple models under the alternative for which the contemporary distribution remains unchanged:
\vspace{-\medskipamount}
\begin{enumerate}[{A}1 -][10]\setcounter{enumi}{10}\parskip0pt
\item An AR(1) model with a break affecting the innovation variance: the $n/2$ first observations are i.i.d.\ standard normal and the $n/2$ last observations are drawn
from an AR(1) model with parameter $\beta \in \{0,0.4,0.8\}$ and centered normal innovations with variance $(1-\beta^2)$. The contemporary distribution is thus the standard normal.
\item A max-autoregressive model with one break: the $n/2$ first observations are i.i.d.\ standard Fr\'echet, and the last $n/2$ observations follow the recursion
\begin{equation*}
X_t = \max\{ \beta X_{t-1}, (1-\beta) Z_t\},
\end{equation*}
where $\beta \in \{0,0.4,0.8\}$ and the $Z_t$ are i.i.d\ standard Fr\'echet. The contemporary distribution is standard Fr\'echet regardless of the choice of $\beta$, see, e.g., Example 10.3 in \cite{BeiGoeSegTeu04}.
\end{enumerate}
\vspace{-\medskipamount}

\subsection{Some competitors to our tests}
\label{sec:competitors}

As mentioned in the introduction, many tests of stationarity have been proposed in the literature. Unfortunately, only a few of them seem to have been implemented in statistical software. In this section, we focus on the tests of \cite{PriSub69}, \cite{Nas13} and \cite{CarNas13} that have been implemented in the \textsf{R} packages \texttt{fractal} \citep{fractal}, \texttt{locits} \citep{locits} and \texttt{costat} \citep{costat}, respectively. Note that we did not include the test of \cite{CarNas16} (implemented in the \textsf{R} package \texttt{BootWPTOS}) in our simulations because we were not able to understand how to initialize the arguments of the corresponding \textsf{R} function.

\begin{table}[t!]
\centering
\caption{Percentages of rejection of the null hypothesis of stationarity computed from 1000 samples of size $n \in \{128, 256, 512\}$ generated from Models N1--N10, first with $N(0,1)$ innovations and then with standardized $t_4$ innovations. The column PSR.T corresponds to the test of \cite{PriSub69} implemented in the \textsf{R} package \texttt{fractal}, the column hwtos2 corresponds to the test of \cite{Nas13} implemented in the \textsf{R} package \texttt{locits} and the column BTOS corresponds to the test of \cite{CarNas13} implemented in the \textsf{R} package \texttt{costat}. }
\label{H0others}
\begingroup\footnotesize
\begin{tabular}{lrrrrrrr}
  \hline
  \multicolumn{2}{c}{} & \multicolumn{3}{c}{$N(0,1)$ innovations} & \multicolumn{3}{c}{Standardized $t_4$ innovations} \\ \cmidrule(lr){3-5} \cmidrule(lr){6-8} Model & $n$ & PSR.T & hwtos2 & BTOS & PSR.T & hwtos2 & BTOS \\ \hline
N1 & 128 & 7.0 & 0.4 & 0.0 & 30.6 & 4.2 & 5.8 \\
   & 256 & 5.7 & 3.3 & 0.0 & 47.3 & 15.1 & 8.5 \\
   & 512 & 5.9 & 3.0 & 0.1 & 62.1 & 19.0 & 13.5 \\
  N2 & 128 & 46.3 & 1.0 & 0.0 & 66.6 & 2.4 & 0.0 \\
   & 256 & 22.0 & 3.6 & 0.0 & 59.5 & 12.5 & 0.0 \\
   & 512 & 11.1 & 4.6 & 0.0 & 65.5 & 13.6 & 0.0 \\
  N3 & 128 & 6.8 & 9.7 & 21.1 & 31.1 & 12.5 & 32.2 \\
   & 256 & 6.7 & 15.8 & 35.8 & 46.0 & 22.8 & 41.4 \\
   & 512 & 5.7 & 17.7 & 38.5 & 63.3 & 29.1 & 52.6 \\
  N4 & 128 & 7.2 & 2.0 & 0.0 & 34.8 & 3.5 & 1.0 \\
   & 256 & 7.4 & 4.7 & 0.0 & 46.8 & 18.8 & 2.4 \\
   & 512 & 6.3 & 3.9 & 0.0 & 61.1 & 15.8 & 2.6 \\
  N5 & 128 & 12.7 & 0.0 & 0.4 & 38.1 & 0.1 & 6.6 \\
   & 256 & 9.1 & 0.0 & 0.5 & 48.3 & 7.4 & 11.6 \\
   & 512 & 6.3 & 0.3 & 0.8 & 64.1 & 6.7 & 16.7 \\
  N6 & 128 & 24.3 & 0.2 & 2.0 & 44.0 & 0.2 & 12.9 \\
   & 256 & 8.6 & 0.7 & 2.4 & 46.8 & 6.5 & 16.6 \\
   & 512 & 7.1 & 0.2 & 3.2 & 63.7 & 4.9 & 21.1 \\
  N7 & 128 & 62.8 & 1.1 & 0.5 & 63.8 & 1.3 & 0.6 \\
   & 256 & 63.1 & 7.1 & 5.0 & 73.5 & 9.1 & 8.3 \\
   & 512 & 49.2 & 20.6 & 17.5 & 78.6 & 29.0 & 22.5 \\
  N8 & 128 & 20.1 & 1.9 & 2.5 & 58.1 & 6.0 & 17.4 \\
   & 256 & 38.5 & 8.9 & 4.1 & 81.0 & 30.1 & 33.8 \\
   & 512 & 56.5 & 11.5 & 6.5 & 94.8 & 44.3 & 47.9 \\
  N9 & 128 & 34.8 & 5.9 & 0.0 & 68.9 & 13.6 & 0.4 \\
   & 256 & 25.8 & 16.6 & 0.2 & 71.3 & 34.8 & 1.6 \\
   & 512 & 17.9 & 22.8 & 0.5 & 75.6 & 41.9 & 3.1 \\
  N10 & 128 & 5.2 & 0.4 & 0.0 & 19.8 & 1.6 & 3.7 \\
   & 256 & 4.2 & 2.0 & 0.0 & 36.1 & 12.2 & 5.3 \\
   & 512 & 4.6 & 2.7 & 0.1 & 61.1 & 12.3 & 8.7 \\
   \hline
\end{tabular}
\endgroup
\end{table}

The rejection percentages of the three aforementioned tests were estimated for Models N1--N10 generating observations under the null. These tests were carried out at the 5\% significance level and the empirical levels were estimated from 1000 samples. As one can see in Table~\ref{H0others}, all these tests turn out to be too liberal in at least one scenario. Overall, the empirical levels are even higher when heavy tailed innovations are used instead of standard normal ones.


\subsection{Empirical levels of the proposed tests}

To estimate the levels of the proposed tests, we considered the same setting as in the previous section. We started with the component tests described in Sections~\ref{sec:rank} and~\ref{sec:sotests}, and then considered various combinations of them, based on the weighted version of Stouffer's method and a weighted version of Fisher's method. As explained previously, the former (resp.\ latter) consists of using $\psi_S$ in~\eqref{eq:stouffer} (resp.\ $\psi_F$ in~\eqref{eq:fisher}) as the function $\psi$ in Sections~\ref{sec:rank} and~\ref{sec:sotests}. As Stouffer's method sometimes gave inflated levels, for the sake of brevity, we only report the results for Fisher's method in this section.

As previously, all the tests were carried out at the 5\% significance level and the empirical levels were estimated from 1000 samples generated from Models N1--N10. The values 128, 256 and 512 were considered for the sample size $n$ and the embedding dimension $h$ was taken to be in the set $\{2,3,4,8\}$.  To save space in the forthcoming tables providing the results, each component test is abbreviated by a single letter, as already explained in Section~\ref{sec:MC}.

\begin{sidewaystable}[t!]
\centering
\caption{Percentages of rejection of the null hypothesis of stationarity computed from 1000 samples of size $n \in \{128, 256, 512\}$ generated from Models N1--N5, first with $N(0,1)$ innovations and then with standardized $t_4$ innovations.  The meaning of the abbreviations d, m, v, c, dc, dcp, a, va and mva is given in Section~\ref{sec:MC}.}
\label{H0_1}
\begingroup\footnotesize
\begin{tabular}{llrrrrrrrrrrrrrrrrrrrrr}
  \hline
  \multicolumn{6}{c}{} & \multicolumn{5}{c}{$h=2$ or lag 1} & \multicolumn{5}{c}{$h=3$ or lag 2} & \multicolumn{5}{c}{$h=4$ or lag 3} & \multicolumn{2}{c}{$h=8$}\\ \cmidrule(lr){7-11} \cmidrule(lr){12-16} \cmidrule(lr){17-21} \cmidrule(lr){22-23} Model & Innov. & $n$ & d & m & v & c & dc & a & va & mva & c & dc & dcp & va & mva & c & dc & dcp  & va & mva & c & dc \\ \hline
N1 & N(0,1) & 128 & 4.0 & 4.7 & 3.8 & 3.0 & 3.9 & 3.0 & 3.5 & 3.6 & 4.6 & 4.1 & 4.2 & 3.9 & 4.1 & 3.2 & 3.4 & 5.7 & 4.1 & 4.5 & 0.0 & 0.7 \\
   &  & 256 & 4.9 & 5.3 & 4.1 & 4.4 & 4.1 & 4.5 & 4.8 & 4.6 & 5.5 & 5.2 & 5.9 & 4.4 & 5.1 & 4.2 & 4.7 & 5.9 & 4.4 & 4.9 & 0.0 & 0.5 \\
   &  & 512 & 5.8 & 4.7 & 4.6 & 4.6 & 5.0 & 4.3 & 5.5 & 4.4 & 6.2 & 6.6 & 5.0 & 5.4 & 4.3 & 6.6 & 6.3 & 5.2 & 4.9 & 5.0 & 0.2 & 2.3 \\
   & Stand. $t_4$ & 128 & 4.5 & 3.4 & 1.6 & 4.2 & 5.0 & 5.2 & 3.7 & 3.3 & 4.9 & 5.3 & 5.7 & 4.7 & 3.0 & 3.4 & 4.0 & 5.9 & 3.7 & 2.4 & 0.0 & 0.7 \\
   &  & 256 & 5.1 & 3.8 & 1.4 & 4.2 & 5.9 & 4.2 & 3.6 & 4.0 & 5.9 & 5.5 & 5.4 & 3.0 & 2.6 & 5.3 & 5.4 & 5.4 & 2.8 & 3.0 & 0.0 & 0.4 \\
   &  & 512 & 4.7 & 4.5 & 2.6 & 5.8 & 5.7 & 4.9 & 3.8 & 4.0 & 8.6 & 6.7 & 4.5 & 4.0 & 3.1 & 7.1 & 6.0 & 5.9 & 3.2 & 3.1 & 0.3 & 1.2 \\
  N2 & N(0,1) & 128 & 2.5 & 2.5 & 1.3 & 1.6 & 3.3 & 1.3 & 4.9 & 3.2 & 1.6 & 2.6 & 5.1 & 7.5 & 3.8 & 1.5 & 2.0 & 6.4 & 8.7 & 4.1 & 0.2 & 0.0 \\
   &  & 256 & 5.1 & 4.2 & 1.8 & 2.3 & 3.5 & 1.3 & 5.7 & 3.0 & 2.2 & 3.9 & 5.5 & 8.0 & 4.4 & 2.1 & 3.6 & 6.3 & 9.1 & 4.4 & 0.5 & 1.8 \\
   &  & 512 & 5.1 & 5.0 & 1.5 & 1.5 & 2.2 & 1.7 & 4.4 & 4.2 & 1.9 & 3.1 & 3.8 & 6.9 & 5.2 & 1.4 & 3.1 & 4.8 & 7.5 & 5.3 & 0.7 & 2.5 \\
   & Stand. $t_4$ & 128 & 2.4 & 2.0 & 1.1 & 2.1 & 3.4 & 1.0 & 4.4 & 3.0 & 1.3 & 3.0 & 5.2 & 6.7 & 3.3 & 1.6 & 2.2 & 5.7 & 7.5 & 3.9 & 0.0 & 0.4 \\
   &  & 256 & 5.6 & 4.8 & 1.2 & 1.0 & 1.3 & 1.1 & 4.5 & 2.8 & 0.8 & 1.9 & 4.4 & 7.3 & 4.0 & 0.5 & 2.0 & 4.9 & 8.3 & 4.4 & 0.9 & 1.3 \\
   &  & 512 & 6.1 & 5.9 & 1.9 & 0.6 & 1.9 & 1.4 & 4.5 & 3.5 & 1.0 & 2.7 & 4.1 & 6.4 & 4.5 & 1.1 & 2.7 & 5.0 & 7.6 & 4.7 & 0.9 & 2.8 \\
  N3 & N(0,1) & 128 & 0.5 & 0.0 & 0.5 & 2.2 & 3.2 & 0.5 & 2.8 & 0.7 & 0.7 & 0.9 & 6.5 & 3.3 & 0.7 & 0.0 & 0.0 & 6.3 & 4.1 & 1.0 & 0.0 & 0.0 \\
   &  & 256 & 1.4 & 0.0 & 0.5 & 2.4 & 4.1 & 0.6 & 2.3 & 0.7 & 0.1 & 1.0 & 6.7 & 3.6 & 1.0 & 0.0 & 0.0 & 6.6 & 4.0 & 1.2 & 0.0 & 0.0 \\
   &  & 512 & 2.7 & 0.3 & 1.4 & 2.2 & 4.8 & 1.1 & 4.2 & 1.6 & 0.5 & 2.0 & 7.7 & 5.9 & 1.9 & 0.2 & 0.5 & 7.9 & 6.9 & 2.1 & 0.0 & 0.0 \\
   & Stand. $t_4$ & 128 & 0.4 & 0.1 & 0.4 & 2.0 & 3.7 & 0.1 & 1.8 & 0.4 & 0.1 & 0.3 & 5.0 & 2.5 & 0.3 & 0.0 & 0.0 & 5.4 & 2.6 & 0.2 & 0.0 & 0.0 \\
   &  & 256 & 2.2 & 0.1 & 0.8 & 2.7 & 5.3 & 0.4 & 2.8 & 0.6 & 0.4 & 2.0 & 7.9 & 4.4 & 0.8 & 0.0 & 0.3 & 8.5 & 5.4 & 0.8 & 0.0 & 0.0 \\
   &  & 512 & 2.8 & 0.7 & 1.0 & 2.9 & 6.3 & 0.6 & 3.4 & 1.8 & 0.3 & 2.0 & 7.7 & 5.2 & 2.0 & 0.1 & 0.7 & 8.7 & 5.2 & 2.2 & 0.0 & 0.2 \\
  N4 & N(0,1) & 128 & 4.8 & 5.3 & 3.9 & 2.5 & 4.2 & 2.5 & 7.3 & 5.5 & 3.3 & 4.5 & 7.3 & 7.0 & 4.9 & 3.7 & 4.6 & 6.6 & 6.9 & 4.7 & 0.1 & 1.1 \\
   &  & 256 & 5.7 & 7.3 & 3.5 & 2.5 & 4.6 & 3.7 & 7.9 & 7.7 & 4.0 & 5.3 & 7.0 & 8.2 & 6.6 & 4.9 & 5.7 & 6.5 & 6.8 & 6.4 & 0.9 & 1.9 \\
   &  & 512 & 5.6 & 6.6 & 4.6 & 4.9 & 5.4 & 4.7 & 7.1 & 7.9 & 6.2 & 5.8 & 6.9 & 7.4 & 7.2 & 6.6 & 5.6 & 7.3 & 7.4 & 6.3 & 1.9 & 3.3 \\
   & Stand. $t_4$ & 128 & 3.4 & 3.9 & 2.7 & 3.8 & 4.6 & 1.2 & 4.5 & 2.8 & 3.5 & 5.0 & 6.4 & 6.2 & 3.7 & 3.5 & 4.7 & 6.3 & 6.4 & 3.4 & 0.2 & 1.4 \\
   &  & 256 & 4.4 & 3.6 & 1.5 & 2.4 & 3.7 & 1.5 & 4.3 & 2.3 & 3.8 & 5.0 & 5.9 & 5.4 & 3.1 & 4.9 & 5.5 & 6.4 & 4.8 & 2.5 & 0.3 & 0.8 \\
   &  & 512 & 5.2 & 4.7 & 2.7 & 4.4 & 5.2 & 1.5 & 5.8 & 4.5 & 5.8 & 5.5 & 6.4 & 5.9 & 4.7 & 6.6 & 6.0 & 6.8 & 5.8 & 5.0 & 2.3 & 4.5 \\
  N5 & N(0,1) & 128 & 1.3 & 0.0 & 1.1 & 3.2 & 2.0 & 1.3 & 3.4 & 0.6 & 2.2 & 1.8 & 2.1 & 3.1 & 0.5 & 0.3 & 0.2 & 2.2 & 2.7 & 0.2 & 0.0 & 0.0 \\
   &  & 256 & 2.0 & 0.0 & 1.3 & 3.5 & 3.4 & 1.3 & 3.5 & 0.5 & 2.6 & 2.7 & 3.6 & 3.6 & 0.1 & 0.8 & 1.3 & 3.2 & 3.2 & 0.0 & 0.0 & 0.2 \\
   &  & 512 & 3.2 & 0.1 & 2.0 & 4.7 & 4.4 & 3.1 & 5.9 & 1.6 & 3.6 & 3.7 & 5.0 & 4.9 & 1.3 & 1.7 & 2.6 & 4.6 & 4.9 & 0.9 & 0.0 & 0.2 \\
   & Stand. $t_4$ & 128 & 1.4 & 0.0 & 1.1 & 3.8 & 3.0 & 1.6 & 4.3 & 0.4 & 2.5 & 1.7 & 3.3 & 3.9 & 0.2 & 0.5 & 0.6 & 3.3 & 3.2 & 0.1 & 0.0 & 0.0 \\
   &  & 256 & 2.4 & 0.0 & 0.7 & 3.9 & 2.9 & 1.5 & 3.4 & 0.3 & 2.6 & 2.1 & 4.2 & 3.2 & 0.1 & 0.8 & 0.8 & 2.9 & 2.6 & 0.0 & 0.0 & 0.2 \\
   &  & 512 & 3.9 & 0.0 & 1.1 & 4.0 & 4.0 & 1.3 & 4.0 & 0.6 & 3.6 & 4.3 & 5.1 & 3.2 & 0.3 & 1.5 & 1.8 & 5.1 & 2.9 & 0.1 & 0.0 & 0.7 \\
   \hline
\end{tabular}
\endgroup
\end{sidewaystable}
\begin{sidewaystable}[t!]
\centering
\caption{Continued from Table~\ref{H0_1}. Percentages of rejection of the null hypothesis of stationarity computed from 1000 samples of size $n \in \{128, 256, 512\}$ generated from Models N6--N10, first with $N(0,1)$ innovations and then with standardized $t_4$ innovations.  The meaning of the abbreviations d, m, v, c, dc, dcp, a, va and mva is given in Section~\ref{sec:MC}.}
\label{H0_2}
\begingroup\footnotesize
\begin{tabular}{llrrrrrrrrrrrrrrrrrrrrr}
  \hline
  \multicolumn{6}{c}{} & \multicolumn{5}{c}{$h=2$ or lag 1} & \multicolumn{5}{c}{$h=3$ or lag 2} & \multicolumn{5}{c}{$h=4$ or lag 3} & \multicolumn{2}{c}{$h=8$}\\ \cmidrule(lr){7-11} \cmidrule(lr){12-16} \cmidrule(lr){17-21} \cmidrule(lr){22-23} Model & Innov. & $n$ & d & m & v & c & dc & a & va & mva & c & dc & dcp & va & mva & c & dc & dcp  & va & mva & c & dc \\ \hline
N6 & N(0,1) & 128 & 2.5 & 0.1 & 1.1 & 3.4 & 3.7 & 1.2 & 4.1 & 1.0 & 0.2 & 0.9 & 5.0 & 5.8 & 1.2 & 0.0 & 0.0 & 5.0 & 5.5 & 1.1 & 0.0 & 0.0 \\
   &  & 256 & 4.1 & 1.1 & 2.9 & 2.7 & 5.5 & 2.4 & 7.9 & 3.1 & 1.3 & 3.3 & 7.2 & 8.6 & 3.9 & 0.0 & 0.5 & 6.8 & 7.8 & 3.2 & 0.0 & 0.2 \\
   &  & 512 & 5.5 & 1.6 & 3.8 & 4.2 & 5.8 & 4.3 & 10.3 & 6.5 & 1.5 & 4.0 & 7.1 & 12.1 & 7.1 & 0.2 & 1.1 & 8.2 & 11.8 & 5.8 & 0.0 & 0.8 \\
   & Stand. $t_4$ & 128 & 1.8 & 0.1 & 1.1 & 2.6 & 3.0 & 1.2 & 4.0 & 0.9 & 0.1 & 0.9 & 4.5 & 4.4 & 1.2 & 0.0 & 0.2 & 4.9 & 3.8 & 0.8 & 0.0 & 0.0 \\
   &  & 256 & 4.1 & 0.7 & 2.2 & 2.7 & 4.3 & 1.5 & 5.1 & 2.4 & 1.1 & 2.3 & 5.9 & 7.1 & 2.3 & 0.2 & 0.9 & 6.2 & 6.2 & 1.7 & 0.0 & 0.5 \\
   &  & 512 & 4.6 & 1.5 & 2.5 & 4.7 & 5.5 & 2.2 & 6.9 & 3.2 & 1.1 & 2.5 & 6.9 & 8.4 & 3.4 & 0.0 & 0.6 & 7.1 & 7.5 & 3.1 & 0.0 & 0.5 \\
  N7 & N(0,1) & 128 & 0.0 & 0.0 & 0.0 & 0.1 & 0.0 & 0.0 & 0.0 & 0.0 & 0.4 & 0.4 & 0.8 & 2.1 & 0.0 & 3.5 & 3.4 & 5.1 & 2.8 & 0.0 & 0.1 & 0.1 \\
   &  & 256 & 1.8 & 0.0 & 0.1 & 0.0 & 0.0 & 0.2 & 3.8 & 0.1 & 0.0 & 0.7 & 0.7 & 2.2 & 0.1 & 0.8 & 1.6 & 2.2 & 4.3 & 0.1 & 0.1 & 0.5 \\
   &  & 512 & 6.8 & 0.0 & 0.9 & 0.0 & 0.5 & 0.8 & 4.9 & 0.9 & 0.0 & 2.1 & 3.2 & 2.3 & 0.1 & 1.0 & 3.7 & 5.2 & 5.3 & 0.4 & 0.2 & 1.0 \\
   & Stand. $t_4$ & 128 & 0.0 & 0.0 & 0.0 & 0.1 & 0.0 & 0.0 & 0.2 & 0.0 & 0.6 & 0.3 & 0.9 & 2.2 & 0.0 & 2.8 & 2.4 & 4.0 & 2.7 & 0.0 & 0.0 & 0.0 \\
   &  & 256 & 1.6 & 0.0 & 0.3 & 0.0 & 0.1 & 0.0 & 3.6 & 0.2 & 0.0 & 0.3 & 0.6 & 1.4 & 0.0 & 0.9 & 1.4 & 3.3 & 4.2 & 0.1 & 0.2 & 0.2 \\
   &  & 512 & 6.5 & 0.0 & 0.6 & 0.0 & 1.0 & 0.9 & 3.2 & 0.7 & 0.0 & 2.1 & 3.1 & 2.5 & 0.4 & 0.5 & 3.2 & 5.2 & 3.5 & 0.7 & 0.3 & 0.8 \\
  N8 & N(0,1) & 128 & 5.9 & 5.0 & 23.5 & 4.4 & 6.0 & 4.3 & 17.7 & 11.5 & 4.9 & 5.5 & 6.3 & 22.4 & 11.6 & 3.0 & 4.2 & 7.0 & 23.4 & 12.3 & 0.0 & 0.7 \\
   &  & 256 & 6.3 & 5.7 & 31.9 & 5.8 & 6.6 & 4.0 & 26.2 & 19.1 & 5.6 & 6.4 & 6.7 & 30.8 & 24.1 & 4.3 & 5.6 & 6.9 & 31.9 & 23.2 & 0.0 & 1.2 \\
   &  & 512 & 7.1 & 5.5 & 37.2 & 5.7 & 6.4 & 4.1 & 30.6 & 25.4 & 7.9 & 8.8 & 7.7 & 33.0 & 27.7 & 8.0 & 8.5 & 7.4 & 35.5 & 28.5 & 0.5 & 2.4 \\
   & Stand. $t_4$ & 128 & 6.0 & 3.6 & 13.0 & 4.0 & 4.4 & 4.3 & 11.2 & 8.0 & 4.6 & 5.8 & 4.9 & 16.5 & 7.5 & 3.2 & 4.9 & 6.0 & 15.6 & 7.9 & 0.0 & 0.7 \\
   &  & 256 & 6.8 & 4.3 & 19.5 & 3.4 & 5.7 & 2.8 & 14.4 & 10.4 & 6.2 & 7.1 & 6.6 & 20.2 & 12.3 & 5.2 & 6.1 & 7.7 & 21.5 & 12.3 & 0.0 & 1.5 \\
   &  & 512 & 8.2 & 4.4 & 23.6 & 4.1 & 6.3 & 3.1 & 19.6 & 15.2 & 6.0 & 8.0 & 7.3 & 21.3 & 15.9 & 6.2 & 7.9 & 8.3 & 24.4 & 16.0 & 0.2 & 2.2 \\
  N9 & N(0,1) & 128 & 5.0 & 4.6 & 1.8 & 2.6 & 4.0 & 0.9 & 5.3 & 2.5 & 2.7 & 3.8 & 5.2 & 6.0 & 2.4 & 2.2 & 3.4 & 6.5 & 4.7 & 2.2 & 0.0 & 0.5 \\
   &  & 256 & 5.6 & 4.8 & 1.9 & 2.3 & 4.4 & 2.4 & 4.9 & 3.6 & 2.5 & 4.2 & 6.2 & 6.4 & 3.6 & 2.4 & 4.2 & 6.4 & 6.5 & 3.3 & 0.1 & 1.3 \\
   &  & 512 & 4.8 & 4.1 & 1.4 & 3.0 & 4.0 & 1.0 & 4.8 & 3.0 & 2.9 & 4.2 & 5.3 & 6.3 & 3.5 & 2.6 & 4.2 & 5.1 & 6.7 & 3.5 & 0.7 & 2.3 \\
   & Stand. $t_4$ & 128 & 3.0 & 2.1 & 1.7 & 2.4 & 3.5 & 1.0 & 4.7 & 1.9 & 2.4 & 2.8 & 4.1 & 5.4 & 1.9 & 1.4 & 2.5 & 4.5 & 5.5 & 2.0 & 0.0 & 0.3 \\
   &  & 256 & 5.5 & 4.6 & 1.4 & 2.5 & 2.9 & 1.1 & 3.6 & 2.6 & 2.1 & 3.8 & 4.9 & 5.4 & 2.9 & 1.3 & 3.4 & 4.9 & 5.1 & 2.5 & 0.1 & 0.7 \\
   &  & 512 & 4.6 & 4.2 & 1.8 & 2.9 & 4.3 & 1.6 & 5.1 & 3.9 & 3.3 & 5.2 & 5.8 & 7.0 & 4.3 & 3.2 & 4.8 & 6.8 & 7.3 & 4.1 & 0.5 & 2.8 \\
  N10 & N(0,1) & 128 & 5.9 & 5.9 & 3.9 & 3.7 & 5.2 & 4.0 & 5.2 & 4.5 & 3.8 & 5.8 & 6.5 & 5.7 & 4.7 & 3.1 & 5.9 & 7.1 & 5.1 & 4.4 & 0.2 & 0.8 \\
   &  & 256 & 6.6 & 6.8 & 3.3 & 3.5 & 5.1 & 3.0 & 5.1 & 4.7 & 3.3 & 5.9 & 6.8 & 6.2 & 5.2 & 3.9 & 5.3 & 7.5 & 5.3 & 5.2 & 0.6 & 1.6 \\
   &  & 512 & 7.4 & 6.7 & 3.8 & 4.5 & 6.6 & 4.0 & 5.4 & 6.8 & 5.0 & 7.2 & 7.8 & 6.1 & 6.7 & 5.4 & 6.6 & 8.6 & 6.1 & 6.3 & 3.1 & 3.6 \\
   & Stand. $t_4$ & 128 & 7.0 & 7.8 & 1.6 & 4.0 & 6.1 & 4.8 & 4.6 & 6.1 & 4.4 & 5.9 & 7.7 & 5.5 & 5.4 & 4.0 & 5.9 & 9.4 & 4.1 & 5.5 & 0.1 & 1.7 \\
   &  & 256 & 6.9 & 5.3 & 2.0 & 3.6 & 6.1 & 3.9 & 3.8 & 4.1 & 4.0 & 5.9 & 7.9 & 2.7 & 4.1 & 4.8 & 6.6 & 8.2 & 3.0 & 4.5 & 0.4 & 2.3 \\
   &  & 512 & 6.7 & 5.8 & 2.2 & 3.4 & 5.6 & 3.6 & 3.3 & 4.0 & 5.0 & 6.1 & 6.4 & 3.7 & 4.3 & 5.3 & 6.1 & 7.2 & 3.0 & 4.4 & 4.2 & 4.0 \\
   \hline
\end{tabular}
\endgroup
\end{sidewaystable}

The empirical levels of the tests are reported in Tables~\ref{H0_1} and~\ref{H0_2} for $h \in \{2,3,4,8\}$ (for $h=8$, to save computing time, only the tests c and dc were carried out). As one can see, the rank-based tests d, c, dc and dcp of Section~\ref{sec:rank} never turned out, overall, to be too liberal (unlike their competitors considered in Section~\ref{sec:competitors} -- see Table~\ref{H0others}). Their analogues of Section~\ref{sec:sotests} focusing on second-order characteristics behave reasonably well except for Model N8. The latter is due to the fact that the test v is way too liberal for Model N8 as can be seen from Table~\ref{H0_2}, a probable consequence of the conditional heteroskedasticity of the model. For fixed $n$, as $h$ increases, the empirical autocopula test c (and thus the combined test dc) can be seen to be more and more conservative, as already mentioned in Section~\ref{sec:hchoice}. The latter clearly appears by considering the last vertical blocks of Tables~\ref{H0_1} and~\ref{H0_2} corresponding to $h=8$. Nonetheless, the rejection percentages therein hint at the fact that the empirical levels, as expected theoretically, should improve as $n$ increases further.

\subsection{Empirical powers}

The empirical powers of the proposed tests were estimated under Models A1--A12 from 1000 samples of size $n \in \{128,256,512\}$ for $h \in \{2,3,4,8\}$ (again, for $h=8$, only the tests c and dc were carried out). For Models A1--A8 (those that are connected to the literature on locally stationary processes), the rejection percentages of the null hypothesis of stationarity are reported in Table~\ref{H1}. As one can see, for $h \in \{2,3,4\}$, the rank-based combined tests of Section~\ref{sec:comb:cop:df} (dc and dcp) almost always seem to be more powerful than the combined tests of second-order stationarity that have been considered in Section~\ref{sec:sotests} (va and mva). Furthermore, the tests focusing on the contemporary distribution (d, m and v) hardly have any power overall, suggesting that the distribution of $X_t$ does not change (too) much for the models under consideration (note in passing the very disappointing behavior of the test m for Models A6--A8). The latter explains why the test c is more powerful than the combined tests dc and dcp, and why the test a is almost always more powerful than va and mva for $h=2$. Finally, note that, except for A8, the power of all the tests focusing on serial dependence decreases, overall, as $h$ increases (see also the discussion in Section~\ref{sec:hchoice}). At least for Models A1--A4, the latter is a consequence of the fact that the serial dependence is completely determined by the distribution of $(X_t,X_{t-1})$.

\begin{sidewaystable}[t!]
\centering
\caption{Percentages of rejection of the null hypothesis of stationarity computed from 1000 samples of size $n \in \{128, 256, 512\}$ generated from Models A1--A8. The meaning of the abbreviations d, m, v, c, dc, dcp, a, va and mva is given in Section~\ref{sec:MC}.}
\label{H1}
\begingroup\footnotesize
\begin{tabular}{lrrrrrrrrrrrrrrrrrrrrr}
  \hline
  \multicolumn{5}{c}{} & \multicolumn{5}{c}{$h=2$ or lag 1} & \multicolumn{5}{c}{$h=3$ or lag 2} & \multicolumn{5}{c}{$h=4$ or lag 3} & \multicolumn{2}{c}{$h=8$} \\ \cmidrule(lr){6-10} \cmidrule(lr){11-15} \cmidrule(lr){16-20} \cmidrule(lr){21-22} Model & $n$ & d & m & v & c & dc & a & va & mva & c & dc & dcp & va & mva & c & dc & dcp  & va & mva  & c & dc \\ \hline
A1 & 128 & 3.3 & 2.5 & 4.3 & 5.0 & 4.3 & 5.5 & 6.5 & 5.1 & 6.7 & 5.2 & 5.5 & 6.5 & 4.0 & 3.2 & 3.8 & 6.2 & 6.3 & 4.6 & 0.2 & 0.9 \\
   & 256 & 5.5 & 5.0 & 5.2 & 11.4 & 10.5 & 12.9 & 12.9 & 10.4 & 13.2 & 11.0 & 9.1 & 11.8 & 8.8 & 6.9 & 6.7 & 9.1 & 10.9 & 8.5 & 0.2 & 1.2 \\
   & 512 & 4.2 & 4.4 & 6.4 & 42.5 & 27.5 & 37.7 & 30.9 & 23.1 & 42.5 & 26.6 & 16.0 & 19.9 & 14.9 & 25.2 & 15.6 & 10.6 & 15.3 & 11.5 & 1.4 & 3.5 \\
  A2 & 128 & 4.9 & 5.2 & 2.1 & 9.7 & 9.3 & 13.2 & 10.1 & 9.5 & 10.1 & 8.2 & 8.7 & 9.2 & 8.2 & 4.7 & 5.0 & 8.3 & 7.1 & 6.4 & 0.0 & 0.8 \\
   & 256 & 6.1 & 6.3 & 3.9 & 36.1 & 26.2 & 34.9 & 26.2 & 22.2 & 31.9 & 23.5 & 15.0 & 16.9 & 14.6 & 17.2 & 14.3 & 12.1 & 14.4 & 11.4 & 0.9 & 2.1 \\
   & 512 & 5.6 & 5.5 & 4.1 & 76.1 & 58.2 & 61.6 & 40.6 & 34.1 & 72.4 & 51.0 & 28.1 & 21.6 & 18.1 & 54.8 & 37.6 & 18.1 & 17.4 & 13.0 & 8.7 & 7.1 \\
  A3 & 128 & 6.2 & 6.3 & 4.3 & 13.1 & 12.9 & 18.8 & 16.9 & 15.1 & 11.8 & 11.9 & 10.4 & 13.5 & 11.4 & 6.0 & 7.2 & 10.1 & 11.7 & 10.1 & 0.1 & 1.1 \\
   & 256 & 5.8 & 5.4 & 4.9 & 42.1 & 30.6 & 40.7 & 28.3 & 23.6 & 36.1 & 25.2 & 17.6 & 18.7 & 14.3 & 19.5 & 16.1 & 13.8 & 15.1 & 11.0 & 1.4 & 2.3 \\
   & 512 & 5.7 & 6.0 & 4.6 & 92.4 & 77.0 & 83.8 & 63.6 & 51.0 & 89.3 & 71.4 & 36.9 & 30.7 & 23.5 & 75.2 & 54.4 & 25.8 & 23.4 & 16.8 & 15.9 & 11.5 \\
  A4 & 128 & 2.0 & 1.9 & 1.8 & 3.3 & 2.5 & 2.4 & 4.9 & 2.9 & 3.0 & 2.5 & 4.1 & 5.7 & 3.0 & 1.4 & 1.2 & 3.8 & 6.1 & 2.9 & 0.0 & 0.0 \\
   & 256 & 4.0 & 2.0 & 1.7 & 3.4 & 3.4 & 2.2 & 5.1 & 2.9 & 3.3 & 3.2 & 4.3 & 6.1 & 3.0 & 1.6 & 1.9 & 4.8 & 5.4 & 2.5 & 0.0 & 0.2 \\
   & 512 & 3.4 & 2.1 & 2.3 & 3.9 & 5.1 & 2.5 & 5.8 & 3.7 & 4.0 & 4.9 & 6.0 & 6.5 & 3.7 & 2.8 & 4.0 & 5.5 & 6.8 & 3.4 & 0.0 & 0.7 \\
  A5 & 128 & 4.9 & 5.6 & 3.2 & 64.0 & 49.9 & 45.6 & 44.2 & 37.1 & 53.4 & 44.5 & 28.3 & 31.8 & 24.7 & 38.5 & 29.9 & 19.8 & 26.7 & 20.3 & 1.2 & 2.5 \\
   & 256 & 5.6 & 5.4 & 4.0 & 84.5 & 75.5 & 60.4 & 54.4 & 43.6 & 77.6 & 68.2 & 46.2 & 40.9 & 29.9 & 63.6 & 52.1 & 35.8 & 36.6 & 25.6 & 12.1 & 9.5 \\
   & 512 & 4.9 & 4.4 & 8.4 & 96.1 & 91.5 & 76.5 & 72.4 & 63.2 & 93.0 & 86.8 & 75.0 & 64.6 & 46.5 & 84.4 & 74.9 & 67.2 & 62.9 & 42.1 & 38.8 & 31.2 \\
  A6 & 128 & 1.2 & 0.0 & 0.1 & 3.8 & 3.6 & 0.1 & 0.4 & 0.0 & 2.3 & 2.1 & 3.2 & 0.3 & 0.0 & 0.4 & 0.6 & 3.4 & 0.5 & 0.0 & 0.0 & 0.0 \\
   & 256 & 7.1 & 0.0 & 0.3 & 4.6 & 8.2 & 0.3 & 1.4 & 0.0 & 3.9 & 6.9 & 10.4 & 1.6 & 0.0 & 0.7 & 2.7 & 10.5 & 1.9 & 0.0 & 0.0 & 0.4 \\
   & 512 & 45.6 & 0.0 & 1.2 & 5.1 & 28.0 & 1.4 & 5.3 & 0.7 & 4.2 & 27.9 & 39.2 & 4.6 & 0.3 & 2.4 & 18.3 & 45.0 & 5.3 & 0.3 & 0.0 & 5.8 \\
  A7 & 128 & 0.2 & 0.0 & 2.5 & 3.9 & 0.9 & 2.0 & 3.4 & 0.4 & 6.9 & 1.8 & 0.5 & 4.2 & 0.7 & 4.7 & 1.3 & 0.7 & 3.5 & 0.5 & 0.0 & 0.0 \\
   & 256 & 0.8 & 0.0 & 1.4 & 6.5 & 1.8 & 2.2 & 2.3 & 0.2 & 6.4 & 1.8 & 1.8 & 2.8 & 0.3 & 4.9 & 1.8 & 2.2 & 2.7 & 0.2 & 0.0 & 0.1 \\
   & 512 & 2.8 & 0.0 & 2.0 & 19.9 & 10.3 & 1.9 & 1.9 & 0.1 & 18.1 & 9.3 & 8.8 & 4.3 & 0.2 & 12.7 & 6.5 & 7.1 & 3.6 & 0.1 & 0.0 & 0.2 \\
  A8 & 128 & 0.0 & 0.0 & 12.2 & 12.5 & 4.0 & 7.5 & 15.8 & 7.2 & 20.0 & 7.3 & 3.6 & 18.0 & 4.3 & 23.8 & 9.6 & 5.4 & 24.3 & 5.9 & 7.8 & 2.0 \\
   & 256 & 0.0 & 0.0 & 12.1 & 30.4 & 11.4 & 11.3 & 19.9 & 6.4 & 39.3 & 21.7 & 16.5 & 25.6 & 6.2 & 45.3 & 27.8 & 21.3 & 34.9 & 10.4 & 42.4 & 20.7 \\
   & 512 & 0.4 & 0.0 & 16.4 & 37.6 & 26.6 & 16.8 & 24.8 & 17.1 & 54.2 & 34.8 & 38.3 & 26.3 & 17.1 & 73.6 & 46.7 & 52.7 & 30.6 & 22.5 & 87.3 & 57.2 \\
   \hline
\end{tabular}
\endgroup
\end{sidewaystable}

The results of Table~\ref{H1} allow in principle for a direct comparison with the results reported in \cite{CarNas13} and \cite{DetPreVet11}. Since the tests available in \textsf{R} considered in \cite{CarNas13} and in Section~\ref{sec:competitors} are far from maintaining their levels, a comparison in terms of power with these tests is clearly not meaningful. As far as the tests of \cite{DetPreVet11} are concerned, they appear, overall, to be more powerful for Models A1--A4 (results for Models A5--A8 are not available in the latter reference). It is however unknown whether they hold their levels when applied to stationary heavy-tailed observations as only Gaussian time series were considered in the simulations of \cite{DetPreVet11}.

While Models A1--A8 considered thus far are connected to the literature on locally stationary processes, the remaining Models A9--A12 are more in line with the change-point literature. For the latter, all our tests (except m) turn out to display substantially more power. This should not come as a surprise given that the tests are based on the CUSUM approach and are hence designed to detect alternatives involving one single break.

Table~\ref{H1ar1} reports the empirical powers of the proposed tests for Model A9. Recall that both the contemporary distribution and the serial dependence is changing under this scenario (unless $\beta=0$). As one can see, even in this setting that should possibly be favorable to the tests focusing on second-order stationarity, the rank-based tests involving test c appear more powerful, overall, than those involving test a. Furthermore, with a few exceptions, the test c is always at least slightly more powerful than the combined test~dc. As expected given the data generating model and in line with the discussion of Section~\ref{sec:hchoice}, the increase of $h$ leads to a decrease in the power of c and dc. In addition, for $h \in \{3,4\}$, dcp is more powerful than dc, which can be explained by the fact that the serial dependence in the data generating model is solely of a bivariate nature.

\begin{sidewaystable}[t!]
\centering
\caption{Percentages of rejection of the null hypothesis of stationarity computed from 1000 samples of size $n \in \{128, 256, 512\}$ generated from Model A9 with $\beta \in \{-0.8, -0.4, 0, 0.4, 0.8\}$. The meaning of the abbreviations d, m, v, c, dc, dcp, a, va and mva is given in Section~\ref{sec:MC}.}
\label{H1ar1}
\begingroup\footnotesize
\begin{tabular}{lrrrrrrrrrrrrrrrrrrrr}
  \hline
  \multicolumn{6}{c}{} & \multicolumn{5}{c}{$h=2$ or lag 1} & \multicolumn{5}{c}{$h=3$ or lag 2} & \multicolumn{5}{c}{$h=4$ or lag 3} \\ \cmidrule(lr){7-11} \cmidrule(lr){12-16} \cmidrule(lr){17-21} Innov. & $n$ & $\beta$ & d & m & v & c & dc & a & va & mva & c & dc & dcp & va & mva & c & dc & dcp  & va & mva \\ \hline
N(0,1) & 128 & -0.8 & 10.0 & 1.6 & 14.4 & 61.2 & 59.5 & 42.1 & 52.5 & 37.7 & 49.8 & 52.9 & 58.8 & 56.2 & 36.4 & 20.6 & 25.8 & 49.6 & 54.3 & 33.0 \\
   &  & -0.4 & 4.6 & 4.8 & 4.6 & 32.0 & 22.6 & 34.2 & 28.2 & 20.2 & 26.3 & 18.7 & 14.3 & 21.5 & 13.3 & 14.9 & 10.4 & 11.5 & 15.9 & 9.8 \\
   &  & 0.0 & 3.9 & 3.1 & 3.8 & 3.6 & 4.2 & 3.9 & 4.2 & 3.8 & 5.3 & 5.1 & 5.2 & 4.6 & 2.7 & 3.8 & 3.7 & 5.9 & 4.3 & 2.7 \\
   &  & 0.4 & 8.6 & 8.3 & 3.9 & 28.0 & 26.6 & 31.2 & 24.1 & 20.8 & 27.9 & 25.2 & 19.0 & 15.5 & 12.2 & 17.8 & 16.7 & 16.4 & 12.5 & 9.8 \\
   &  & 0.8 & 6.7 & 6.1 & 12.7 & 52.8 & 45.2 & 39.5 & 47.5 & 32.4 & 49.2 & 43.2 & 41.2 & 50.9 & 32.1 & 39.6 & 34.4 & 33.4 & 45.9 & 27.8 \\
   & 256 & -0.8 & 40.9 & 2.5 & 46.1 & 95.5 & 97.0 & 74.2 & 81.8 & 69.5 & 93.5 & 95.0 & 96.6 & 83.1 & 70.7 & 69.4 & 84.8 & 95.2 & 82.8 & 67.6 \\
   &  & -0.4 & 5.7 & 4.6 & 10.2 & 66.4 & 52.8 & 69.0 & 60.3 & 48.6 & 57.3 & 44.7 & 34.9 & 45.0 & 33.0 & 35.4 & 28.2 & 23.3 & 32.1 & 22.3 \\
   &  & 0.0 & 3.7 & 4.0 & 4.6 & 4.4 & 3.3 & 4.7 & 4.1 & 4.1 & 6.2 & 5.1 & 4.6 & 4.9 & 4.8 & 5.3 & 4.2 & 4.8 & 5.1 & 4.8 \\
   &  & 0.4 & 6.8 & 6.6 & 6.6 & 63.6 & 54.9 & 69.8 & 57.3 & 48.4 & 60.8 & 51.7 & 37.6 & 39.7 & 28.7 & 50.6 & 42.3 & 25.1 & 27.3 & 19.5 \\
   &  & 0.8 & 10.0 & 6.9 & 37.6 & 89.4 & 85.0 & 74.7 & 81.2 & 68.2 & 86.9 & 82.2 & 80.8 & 84.0 & 67.0 & 80.5 & 76.5 & 72.1 & 82.0 & 63.3 \\
   & 512 & -0.8 & 87.7 & 2.0 & 85.9 & 100.0 & 100.0 & 95.2 & 97.7 & 93.9 & 100.0 & 100.0 & 100.0 & 98.3 & 94.7 & 99.2 & 99.9 & 99.9 & 98.4 & 93.7 \\
   &  & -0.4 & 6.8 & 5.6 & 17.3 & 95.7 & 89.8 & 98.1 & 93.8 & 88.6 & 89.6 & 83.9 & 75.8 & 81.4 & 67.2 & 75.2 & 62.9 & 48.8 & 63.2 & 48.4 \\
   &  & 0.0 & 4.2 & 4.5 & 4.2 & 5.2 & 5.2 & 5.0 & 4.1 & 4.0 & 6.5 & 5.6 & 4.9 & 4.2 & 4.7 & 6.3 & 4.9 & 5.1 & 4.6 & 4.9 \\
   &  & 0.4 & 7.0 & 5.5 & 12.8 & 95.0 & 89.7 & 96.2 & 90.5 & 86.4 & 91.9 & 86.4 & 74.9 & 77.8 & 63.4 & 86.8 & 78.2 & 50.5 & 60.0 & 48.0 \\
   &  & 0.8 & 11.2 & 6.8 & 84.2 & 99.7 & 99.7 & 96.1 & 98.5 & 94.3 & 99.6 & 99.5 & 99.5 & 99.1 & 94.4 & 99.0 & 98.3 & 98.8 & 98.6 & 93.7 \\
  St.\ $t_4$ & 128 & -0.8 & 15.8 & 1.0 & 7.7 & 61.5 & 67.1 & 34.4 & 37.6 & 23.5 & 51.5 & 59.9 & 67.7 & 42.5 & 23.5 & 21.8 & 35.1 & 60.3 & 41.0 & 22.1 \\
   &  & -0.4 & 4.4 & 2.5 & 2.6 & 38.3 & 27.8 & 25.9 & 21.0 & 11.2 & 33.5 & 23.5 & 16.4 & 13.2 & 6.1 & 15.9 & 11.7 & 12.0 & 10.6 & 4.3 \\
   &  & 0.0 & 5.1 & 3.6 & 1.2 & 3.2 & 4.6 & 4.4 & 4.1 & 2.8 & 5.8 & 3.9 & 5.3 & 3.9 & 1.9 & 3.9 & 3.2 & 6.1 & 3.5 & 1.7 \\
   &  & 0.4 & 9.4 & 7.0 & 2.1 & 37.3 & 32.0 & 26.1 & 19.8 & 14.2 & 32.0 & 29.1 & 23.8 & 12.5 & 8.9 & 22.3 & 21.2 & 21.0 & 10.6 & 6.7 \\
   &  & 0.8 & 9.4 & 7.4 & 7.9 & 56.3 & 52.9 & 33.5 & 37.2 & 27.3 & 50.2 & 48.0 & 47.2 & 39.4 & 25.2 & 39.2 & 38.9 & 39.7 & 36.8 & 21.1 \\
   & 256 & -0.8 & 56.3 & 2.2 & 26.6 & 96.1 & 98.5 & 65.3 & 69.8 & 54.7 & 95.8 & 98.6 & 97.7 & 73.4 & 56.2 & 73.3 & 92.2 & 96.9 & 72.1 & 53.0 \\
   &  & -0.4 & 5.4 & 3.0 & 3.9 & 74.5 & 61.4 & 58.9 & 44.6 & 31.9 & 65.2 & 54.4 & 43.3 & 30.3 & 18.1 & 41.5 & 32.4 & 28.0 & 21.5 & 12.3 \\
   &  & 0.0 & 4.7 & 4.8 & 1.7 & 4.8 & 4.1 & 4.8 & 3.5 & 3.2 & 6.9 & 6.4 & 5.1 & 2.2 & 2.5 & 6.6 & 6.2 & 5.3 & 2.3 & 2.4 \\
   &  & 0.4 & 6.4 & 5.6 & 3.2 & 71.5 & 62.9 & 58.1 & 45.3 & 35.6 & 67.7 & 57.8 & 44.4 & 27.8 & 19.4 & 56.5 & 47.1 & 30.6 & 18.1 & 13.2 \\
   &  & 0.8 & 11.2 & 6.1 & 20.6 & 92.4 & 87.6 & 62.1 & 65.6 & 51.2 & 89.2 & 86.1 & 84.2 & 68.7 & 51.5 & 83.9 & 81.0 & 77.0 & 66.0 & 47.5 \\
   & 512 & -0.8 & 96.1 & 2.8 & 58.4 & 100.0 & 100.0 & 88.4 & 92.3 & 83.9 & 100.0 & 100.0 & 100.0 & 93.9 & 84.6 & 99.9 & 100.0 & 100.0 & 93.8 & 84.0 \\
   &  & -0.4 & 8.8 & 4.7 & 5.1 & 98.6 & 96.2 & 89.3 & 81.5 & 69.7 & 95.3 & 91.3 & 85.2 & 65.7 & 46.4 & 84.6 & 74.8 & 60.0 & 46.3 & 27.3 \\
   &  & 0.0 & 4.9 & 3.3 & 1.6 & 5.1 & 5.6 & 4.3 & 3.9 & 5.1 & 6.7 & 6.2 & 5.4 & 3.9 & 3.4 & 7.2 & 6.2 & 6.0 & 2.4 & 2.3 \\
   &  & 0.4 & 7.0 & 5.4 & 4.9 & 98.1 & 95.8 & 87.8 & 79.3 & 69.7 & 97.2 & 92.8 & 81.4 & 62.6 & 46.1 & 93.1 & 85.4 & 59.5 & 40.8 & 28.1 \\
   &  & 0.8 & 14.9 & 5.7 & 57.1 & 99.9 & 100.0 & 87.5 & 89.7 & 82.5 & 99.7 & 99.6 & 99.9 & 91.7 & 83.1 & 99.7 & 99.4 & 99.5 & 91.9 & 82.2 \\
   \hline
\end{tabular}
\endgroup
\end{sidewaystable}

The rejection percentages for Model A10 are reported in Table~\ref{H1ar2}. As expected, the empirical powers of tests c and a are very low for $h=2$ since there is no relationship between $X_t$ and $X_{t-1}$. The tests focusing on the contemporary distribution are more powerful, in particular the test v. Consequently, the combined tests at lag 1 involving v do display some power. For $h \in \{3,4\}$, the two most powerful tests are dcp and va. The fact that dcp is more powerful than dc can again be explained by the bivariate nature of the serial dependence.

\begin{sidewaystable}[t!]
\centering
\caption{Percentages of rejection of the null hypothesis of stationarity computed from 1000 samples of size $n \in \{128, 256, 512\}$ generated from Model A10 with $\beta \in \{-0.8, -0.4, 0, 0.4, 0.8\}$. The meaning of the abbreviations d, m, v, c, dc, dcp, a, va and mva is given in Section~\ref{sec:MC}.}
\label{H1ar2}
\begingroup\footnotesize
\begin{tabular}{lrrrrrrrrrrrrrrrrrrrr}
  \hline
  \multicolumn{6}{c}{} & \multicolumn{5}{c}{$h=2$ or lag 1} & \multicolumn{5}{c}{$h=3$ or lag 2} & \multicolumn{5}{c}{$h=4$ or lag 3} \\ \cmidrule(lr){7-11} \cmidrule(lr){12-16} \cmidrule(lr){17-21} Innov. & $n$ & $\beta$ & d & m & v & c & dc & a & va & mva & c & dc & dcp & va & mva & c & dc & dcp  & va & mva \\ \hline
N(0,1) & 128 & -0.8 & 2.6 & 0.4 & 11.6 & 1.5 & 3.7 & 2.8 & 10.2 & 4.4 & 16.1 & 21.9 & 26.7 & 35.4 & 15.3 & 11.3 & 16.2 & 20.0 & 29.9 & 11.0 \\
   &  & -0.4 & 2.4 & 2.5 & 7.0 & 3.7 & 2.6 & 2.3 & 4.2 & 2.6 & 9.3 & 5.6 & 9.3 & 19.4 & 9.8 & 8.7 & 4.9 & 6.2 & 15.5 & 7.2 \\
   &  & 0.0 & 3.9 & 3.1 & 3.8 & 3.6 & 4.2 & 3.9 & 4.2 & 3.8 & 5.3 & 5.1 & 5.2 & 4.6 & 2.7 & 3.8 & 3.7 & 5.9 & 4.3 & 2.7 \\
   &  & 0.4 & 8.8 & 11.3 & 4.3 & 5.2 & 9.1 & 6.2 & 6.2 & 7.8 & 9.0 & 11.3 & 21.0 & 19.0 & 15.8 & 7.4 & 10.1 & 17.3 & 14.8 & 12.7 \\
   &  & 0.8 & 7.3 & 8.6 & 11.3 & 1.5 & 5.6 & 4.1 & 14.1 & 13.4 & 2.5 & 5.6 & 25.2 & 38.7 & 27.4 & 2.8 & 5.0 & 22.7 & 36.7 & 23.9 \\
   & 256 & -0.8 & 28.8 & 0.8 & 33.9 & 1.2 & 13.0 & 1.3 & 18.5 & 10.0 & 67.5 & 80.8 & 74.7 & 63.8 & 38.0 & 52.5 & 73.3 & 61.4 & 56.6 & 30.9 \\
   &  & -0.4 & 4.0 & 3.2 & 8.8 & 3.4 & 4.2 & 3.2 & 6.4 & 4.3 & 21.6 & 15.4 & 24.9 & 39.0 & 25.3 & 20.4 & 13.9 & 15.8 & 27.5 & 18.1 \\
   &  & 0.0 & 3.7 & 4.0 & 4.6 & 4.4 & 3.3 & 4.7 & 4.1 & 4.1 & 6.2 & 5.1 & 4.6 & 4.9 & 4.8 & 5.3 & 4.2 & 4.8 & 5.1 & 4.8 \\
   &  & 0.4 & 9.1 & 9.2 & 8.9 & 4.7 & 7.9 & 6.0 & 9.2 & 10.8 & 11.6 & 15.6 & 33.2 & 41.5 & 35.8 & 14.5 & 17.1 & 25.6 & 30.2 & 26.0 \\
   &  & 0.8 & 10.4 & 7.6 & 36.2 & 1.1 & 6.1 & 3.2 & 27.2 & 21.2 & 4.4 & 8.4 & 41.8 & 71.0 & 53.3 & 5.8 & 11.1 & 31.1 & 66.1 & 45.2 \\
   & 512 & -0.8 & 81.3 & 1.5 & 78.2 & 2.9 & 56.0 & 1.9 & 53.0 & 35.4 & 99.5 & 99.8 & 98.7 & 90.7 & 75.7 & 97.4 & 99.5 & 97.0 & 88.1 & 69.3 \\
   &  & -0.4 & 4.1 & 3.6 & 17.4 & 3.7 & 3.0 & 4.4 & 10.9 & 7.7 & 47.8 & 35.3 & 58.9 & 74.0 & 51.7 & 49.0 & 35.3 & 32.7 & 53.4 & 35.1 \\
   &  & 0.0 & 4.2 & 4.5 & 4.2 & 5.2 & 5.2 & 5.0 & 4.1 & 4.0 & 6.5 & 5.6 & 4.9 & 4.2 & 4.7 & 6.3 & 4.9 & 5.1 & 4.6 & 4.9 \\
   &  & 0.4 & 8.0 & 7.1 & 15.8 & 4.2 & 7.8 & 5.3 & 11.7 & 11.8 & 23.7 & 20.7 & 62.5 & 72.2 & 57.6 & 32.4 & 25.9 & 36.0 & 51.9 & 40.0 \\
   &  & 0.8 & 17.1 & 8.1 & 77.7 & 1.4 & 8.3 & 3.8 & 60.4 & 46.3 & 16.2 & 27.7 & 83.2 & 93.8 & 84.2 & 27.8 & 40.3 & 58.5 & 91.9 & 78.3 \\
  St.\ $t_4$ & 128 & -0.8 & 5.7 & 0.4 & 6.9 & 1.4 & 5.1 & 4.4 & 7.6 & 4.5 & 17.0 & 29.5 & 36.0 & 26.6 & 10.2 & 11.5 & 21.6 & 28.1 & 21.7 & 6.5 \\
   &  & -0.4 & 2.0 & 1.9 & 2.5 & 2.8 & 1.8 & 4.2 & 4.3 & 3.2 & 13.4 & 6.7 & 12.2 & 14.4 & 6.9 & 10.1 & 6.6 & 8.1 & 11.1 & 4.5 \\
   &  & 0.0 & 5.1 & 3.6 & 1.2 & 3.2 & 4.6 & 4.4 & 4.1 & 2.8 & 5.8 & 3.9 & 5.3 & 3.9 & 1.9 & 3.9 & 3.2 & 6.1 & 3.5 & 1.7 \\
   &  & 0.4 & 11.8 & 10.3 & 2.8 & 5.5 & 9.8 & 7.4 & 7.7 & 8.4 & 8.0 & 11.6 & 23.9 & 15.9 & 14.9 & 8.2 & 10.9 & 21.5 & 11.6 & 11.5 \\
   &  & 0.8 & 9.9 & 7.9 & 8.5 & 1.6 & 7.1 & 8.4 & 16.0 & 14.5 & 1.9 & 6.4 & 30.3 & 35.1 & 22.7 & 1.6 & 5.4 & 26.1 & 33.4 & 22.4 \\
   & 256 & -0.8 & 44.7 & 0.2 & 19.7 & 0.6 & 20.1 & 3.2 & 12.5 & 5.9 & 69.4 & 88.5 & 83.0 & 46.8 & 24.9 & 55.7 & 82.3 & 74.3 & 40.4 & 18.4 \\
   &  & -0.4 & 3.0 & 1.9 & 4.3 & 3.9 & 3.0 & 2.9 & 4.5 & 2.8 & 23.3 & 16.0 & 28.8 & 27.8 & 15.1 & 22.9 & 14.9 & 16.3 & 17.8 & 7.9 \\
   &  & 0.0 & 4.7 & 4.8 & 1.7 & 4.8 & 4.1 & 4.8 & 3.5 & 3.2 & 6.9 & 6.4 & 5.1 & 2.2 & 2.5 & 6.6 & 6.2 & 5.3 & 2.3 & 2.4 \\
   &  & 0.4 & 10.0 & 9.1 & 2.9 & 4.4 & 9.1 & 5.6 & 6.5 & 7.8 & 11.0 & 15.3 & 41.3 & 28.0 & 21.0 & 14.8 & 17.6 & 28.9 & 18.2 & 15.7 \\
   &  & 0.8 & 13.7 & 7.9 & 20.3 & 1.3 & 6.7 & 5.1 & 18.5 & 14.3 & 3.0 & 8.3 & 50.1 & 53.4 & 39.2 & 4.7 & 11.7 & 35.8 & 49.0 & 32.5 \\
   & 512 & -0.8 & 92.3 & 2.1 & 52.5 & 2.9 & 70.6 & 2.1 & 29.2 & 18.8 & 99.8 & 100.0 & 99.8 & 77.7 & 55.1 & 98.3 & 99.7 & 98.6 & 72.0 & 46.4 \\
   &  & -0.4 & 7.5 & 2.6 & 8.1 & 4.3 & 6.0 & 3.3 & 5.8 & 4.9 & 60.1 & 49.9 & 69.8 & 55.7 & 32.6 & 57.6 & 47.5 & 41.9 & 33.5 & 18.9 \\
   &  & 0.0 & 4.9 & 3.3 & 1.6 & 5.1 & 5.6 & 4.3 & 3.9 & 5.1 & 6.7 & 6.2 & 5.4 & 3.9 & 3.4 & 7.2 & 6.2 & 6.0 & 2.4 & 2.3 \\
   &  & 0.4 & 10.5 & 8.3 & 6.5 & 6.0 & 8.2 & 5.4 & 6.5 & 6.6 & 27.1 & 28.2 & 74.4 & 53.6 & 39.7 & 36.0 & 33.6 & 49.4 & 33.0 & 25.7 \\
   &  & 0.8 & 19.1 & 6.7 & 51.8 & 1.1 & 9.7 & 4.3 & 37.7 & 27.8 & 14.2 & 28.3 & 86.5 & 80.7 & 64.6 & 27.1 & 42.7 & 64.2 & 75.0 & 57.7 \\
   \hline
\end{tabular}
\endgroup
\end{sidewaystable}

Finally, the rejection percentages for Models A11 and A12 are given in Table~\ref{H1otherAR}. The columns c2 and c3 report the results for the bivariate analogues of the tests based on $S_{n,C^{\scs (2)}}$ defined by~\eqref{eq:SnCh} for lags 2 and 3 (these tests arise in the combined test dcp). To save computing time, we did not include the tests of second-order stationarity as these were found less powerful, overall, in the previous experiments (for Models A12, moments do not exist, whence an application would not even be meaningful). Comparing the results for Model A11 with those of Table~\ref{H1ar1} for the same values of $h$ reveals, as expected, a higher power of the test c. In addition, the test c for lag 1 is more powerful than the test c2, which, in turn, is more powerful than the test c3, a consequence of the data generating models. Finally, the fact that the test d displays some power for $\beta = 0.8$ seems to be only a consequence of the sample sizes under consideration and the very strong serial dependence in the second half of the observations.

\begin{table}[t!]
\centering
\caption{Percentages of rejection of the null hypothesis of stationarity computed from 1000 samples of size $n \in \{128, 256, 512\}$ generated from Models A11 and A12 with $\beta \in \{0, 0.4, 0.8\}$. The meaning of the abbreviations d, c, dc, dcp is given in Section~\ref{sec:MC}. The columns c2 and c3 report the results for the bivariate analogues of the test based on $S_{n,C^{\scs (2)}}$ defined by~\eqref{eq:SnCh} for lags 2 and 3 (these tests arise in the combined test dcp).}
\label{H1otherAR}
\begingroup\footnotesize
\begin{tabular}{lrrrrrrrrrrrrr}
  \hline
  \multicolumn{4}{c}{} & \multicolumn{2}{c}{$h=2$ or lag 1} & \multicolumn{4}{c}{$h=3$ or lag 2} & \multicolumn{4}{c}{$h=4$ or lag 3} \\ \cmidrule(lr){5-6} \cmidrule(lr){7-10} \cmidrule(lr){11-14} Model & $n$ & $\beta$ & d & c & dc & c & dc & c2 & dcp & c & dc & c3 & dcp \\ \hline
A11 & 128 & 0.0 & 4.3 & 4.9 & 4.5 & 4.8 & 3.8 & 5.2 & 5.1 & 4.1 & 3.2 & 6.4 & 5.1 \\
   &  & 0.4 & 8.3 & 69.3 & 60.7 & 60.6 & 53.3 & 8.4 & 41.8 & 48.6 & 40.4 & 7.1 & 30.3 \\
   &  & 0.8 & 22.1 & 91.9 & 90.4 & 86.4 & 84.6 & 80.0 & 91.2 & 80.8 & 79.5 & 55.5 & 88.6 \\
   & 256 & 0.0 & 4.3 & 3.8 & 4.3 & 6.3 & 5.9 & 4.8 & 5.0 & 5.9 & 4.9 & 5.6 & 5.5 \\
   &  & 0.4 & 6.4 & 99.5 & 96.1 & 96.9 & 91.2 & 20.0 & 78.4 & 92.3 & 84.5 & 6.2 & 54.7 \\
   &  & 0.8 & 25.3 & 99.3 & 98.5 & 96.4 & 94.9 & 95.8 & 98.6 & 92.4 & 90.4 & 85.9 & 96.8 \\
   & 512 & 0.0 & 4.6 & 6.6 & 5.9 & 7.0 & 6.5 & 4.2 & 5.2 & 5.6 & 6.2 & 5.0 & 5.0 \\
   &  & 0.4 & 6.4 & 100.0 & 100.0 & 99.9 & 100.0 & 56.7 & 99.2 & 99.8 & 99.6 & 10.7 & 83.6 \\
   &  & 0.8 & 22.1 & 100.0 & 100.0 & 100.0 & 100.0 & 100.0 & 100.0 & 99.9 & 99.7 & 99.7 & 100.0 \\
  A12 & 128 & 0.0 & 5.6 & 3.7 & 4.7 & 4.9 & 4.9 & 4.4 & 5.9 & 3.3 & 4.1 & 6.3 & 7.4 \\
   &  & 0.4 & 6.5 & 30.8 & 24.9 & 31.3 & 25.5 & 6.8 & 18.7 & 23.5 & 20.1 & 6.5 & 14.3 \\
   &  & 0.8 & 10.0 & 69.9 & 60.0 & 67.6 & 61.1 & 32.8 & 55.8 & 62.0 & 57.8 & 15.3 & 48.1 \\
   & 256 & 0.0 & 4.5 & 4.4 & 3.7 & 6.8 & 5.2 & 5.9 & 4.2 & 5.1 & 4.8 & 5.5 & 4.7 \\
   &  & 0.4 & 7.1 & 67.0 & 55.2 & 66.3 & 53.7 & 11.1 & 35.6 & 54.8 & 42.5 & 6.0 & 24.0 \\
   &  & 0.8 & 7.9 & 97.8 & 93.5 & 97.1 & 92.9 & 74.3 & 91.5 & 95.3 & 91.2 & 41.0 & 85.0 \\
   & 512 & 0.0 & 4.8 & 4.7 & 4.6 & 8.6 & 7.5 & 4.7 & 4.9 & 7.4 & 6.9 & 4.4 & 4.6 \\
   &  & 0.4 & 6.5 & 96.0 & 90.8 & 94.3 & 88.5 & 22.4 & 75.5 & 89.4 & 83.1 & 8.2 & 51.6 \\
   &  & 0.8 & 7.2 & 100.0 & 99.7 & 100.0 & 99.6 & 99.3 & 99.6 & 99.9 & 99.2 & 89.2 & 99.6 \\
   \hline
\end{tabular}
\endgroup
\end{table}

\end{document}